\documentclass[11pt]{article}
\usepackage{amsmath,amsfonts,amssymb,amsthm}
\usepackage{color}

\usepackage{algorithmicx}
\usepackage{enumitem}
\usepackage{graphicx}
\usepackage{bm}

\def\PAY{\bP\bA\bY}
\def\algzero{\text{{\em Algorand}}_0}
\def\alg1{\text{{\em Algorand1}}}
\def\alg2{\text{{\em Algorand2}}}
\def\alg3{\text{{\em Algorand3}}}

\def\alg{\text{{\em Algorand}}\,'}

\def\BBA*{BBA^\star}
\def\BA*{BA^\star}

\def\rs{^{r,s}}

\makeatletter
\renewcommand\section{\@startsection {section}{1}{\z@}%
                                   {-3.5ex \@plus -1ex \@minus -.2ex}%
                                   {2.3ex \@plus.2ex}%
                                   {\normalfont\LARGE\bfseries}}
\renewcommand\subsection{\@startsection{subsection}{2}{\z@}%
                                     {-3.25ex\@plus -1ex \@minus -.2ex}%
                                     {1.5ex \@plus .2ex}%
                                     {\normalfont\Large\bfseries}}
\renewcommand\subsubsection{\@startsection{subsubsection}{3}{\z@}%
                                     {-3.25ex\@plus -1ex \@minus -.2ex}%
                                     {1.5ex \@plus .2ex}%
                                     {\normalfont\large\bfseries}}
\makeatother

\usepackage{tikz}
\usetikzlibrary{trees}

\usepackage{latexsym}
\usepackage{color}
\usepackage{url,array,epsfig,fullpage}
\usepackage[colorlinks=false,citecolor=blue,linkcolor=red,hidelinks]{hyperref}

\usepackage{epsfig}
\newcommand{\remove}[1]{}
\newcommand{\ignore}[1]{}
\newtheorem{theorem}{Theorem}[section]

\newtheorem{definition}{Definition}[section]

\newtheorem{lemma}[theorem]{Lemma}

\newcommand{\scparagraph}[1]{\medskip\noindent\textsc{#1}}

\def\blackslug{\hbox{\hskip 1pt \vrule width 8pt height 8pt depth 0pt
\hskip 1pt}}

\def\bqed{\quad\blackslug\lower 8.5pt\null\par}

\def\wqed
{\quad\raisebox{-.3ex}{\Large$\Box$}\lower 8.5pt\null\par}
\long\gdef\boxit#1{\begingroup\vbox{\hrule\hbox{\vrule\kern3pt
      \vbox{\kern3pt#1\kern3pt}\kern3pt\vrule}\hrule}\endgroup}

\def\cI{{\cal I}}

\def\cP{{\cal P}}

\newcommand{\bit}{ \{0,1\}}

\newcommand{\bA}{\mathbb{A}}

\newcommand{\bE}{\mathbb{E}}

\newcommand{\bP}{\mathbb{P}}

\newcommand{\bY}{\mathbb{Y}}
\newcommand{\bZ}{\mathbb{Z}}

\def\bit{\{0,1\}}

\newcommand{\upr}{^r}
\newcommand{\uprm}{^{r-1}}
\newcommand{\uprp}{^{r+1}}

\def\bydef{\triangleq}

\DeclareMathOperator*{\argmin}{arg\,min}

\def\bydef{\triangleq}

\def\ShowAuthNotes{1}

\ifnum\ShowAuthNotes=1

\newcommand{\authnote}[2]{\textcolor{red}{\parbox{0.9\linewidth}{[{\footnotesize {\bf #1:} { {#2}}}]}}}

\else
\newcommand{\authnote}[2]{}
\fi

\newcommand{\absnewline}{\ifnum\full=0 \\ \fi}

\setlist[description]{leftmargin=18pt}
\setlist[enumerate]{leftmargin=18pt}
\setlist[itemize]{leftmargin=18pt}

\def\full{1} 

\newlength{\saveparindent}
\newlength{\saveparskip}
\newcounter{ctr}

\newenvironment{newitemize}{%
\begin{list}{$\bullet$}{\labelwidth=18pt%
\labelsep=5pt \leftmargin=23pt \topsep=1pt%
\setlength{\listparindent}{0pt}%
\setlength{\parsep}{\saveparskip}%
\setlength{\itemsep}{3pt}}}{\end{list}}

\newenvironment{newcenter}{%
\begin{center}%
\vspace{-5pt}%
\setlength{\parskip}{0pt}}{\end{center}\vspace{-5pt}}

\DeclareMathOperator{\lsb}{{\tt lsb}}

\title{ALGORAND%
\footnote{This is the more formal (and asynchronous) version of the ArXiv paper by the second author \cite{ALG}, a paper itself based on that of Gorbunov and Micali \cite{democoin}.
Algorand's technologies are the object of the following patent applications: US62/117,138   US62/120,916   US62/142,318   US62/218,817   US62/314,601  PCT/US2016/018300   US62/326,865   62/331,654   US62/333,340   US62/343,369   US62/344,667   US62/346,775   US62/351,011   US62/653,482   US62/352,195   US62/363,970  US62/369,447   US62/378,753   US62/383,299   US62/394,091   US62/400,361   US62/403,403  US62/410,721   US62/416,959   US62/422,883   US62/455,444   US62/458,746   US62/459,652   US62/460,928   US62/465,931
}
}

\author{Jing Chen \\ Computer Science Department \\ Stony Brook University \\ Stony Brook, NY 11794, USA \\ jingchen@cs.stonybrook.edu \and Silvio Micali \\ CSAIL \\ MIT \\ Cambridge, MA 02139, USA \\ silvio@csail.mit.edu
}

\begin{document}

\date{}
\maketitle

\begin{center}
Abstract
\end{center}

\noindent
A public ledger is a tamperproof sequence of data that can be read and augmented by everyone. Public ledgers have innumerable and compelling uses. They can secure, in plain sight, all kinds of transactions ---such as titles, sales, and payments--- in the exact order in which they occur.
Public ledgers not only curb corruption, but also enable very sophisticated applications ---such as cryptocurrencies and smart contracts. They stand to revolutionize the way a democratic society operates. As currently implemented, however, they scale poorly and cannot achieve their  potential.

\medskip

\noindent
{\em Algorand} is a  truly democratic and  efficient way to implement a public ledger. Unlike prior implementations based on proof of work, it requires a negligible amount of computation, and generates a transaction history that will not ``fork'' with overwhelmingly high probability.

\medskip

Algorand is based on (a novel and super fast) message-passing Byzantine agreement.

\bigskip

\smallskip

\begin{center}
For  concreteness,  we shall describe Algorand only as a money platform.
\end{center}

\section{Introduction}

Money is becoming increasingly virtual.
It has been estimated that about 80\% of United States dollars today only exist as
ledger entries~\cite{ledger-dollars}.
Other financial instruments are following suit.

In an ideal world, in which we could count on a universally trusted central entity, immune to all possible cyber attacks,
money and other financial transactions could be  solely electronic.
Unfortunately, we do not live in such a world.   Accordingly,
decentralized cryptocurrencies, such as Bitcoin~\cite{bitcoin},
and ``smart contract'' systems, such as Ethereum, have been proposed~\cite{Ethereum}.
At the heart of these systems is a shared \emph{ledger} that reliably records a
 sequence of transactions, as varied as payments and contracts, in a tamperproof way.
The technology of choice to guarantee such tamperproofness is the {\em blockchain}. Blockchains are behind applications such as cryptocurrencies~\cite{bitcoin},
financial applications~\cite{Ethereum}, and
the Internet of Things~\cite{coindesk-iot}.
Several techniques to manage blockchain-based  ledgers have been proposed: \emph{proof of work}~\cite{bitcoin},
\emph{proof of stake}~\cite{pos},
\emph{practical Byzantine fault-tolerance} \cite{PBFT},
or some combination.

Currently, however,
ledgers can be inefficient to manage.
For example,
Bitcoin's  \emph{proof-of-work}  approach (based on the original concept of \cite{DN}) requires a vast amount of computation,
is wasteful and scales poorly \cite{bitcoinwaste}. In addition, it {\em de facto} concentrates power in very few hands.

We therefore wish to put forward a new method to implement a public ledger that offers the
convenience and efficiency of a centralized system run by a trusted and inviolable authority,
without the inefficiencies and weaknesses of current decentralized implementations.
We call our approach {\em Algorand}, because we use algorithmic randomness to select, based on the ledger constructed so far,  a set of {\em verifiers} who are in charge of constructing the next block of valid transactions.
Naturally, we ensure that such selections are provably immune from manipulations and unpredictable until the last minute, but also that they ultimately are universally clear.

Algorand's approach is quite democratic, in the sense that neither in principle nor {\em de facto} it creates different classes of users (as ``miners" and ``ordinary users" in Bitcoin). In Algorand ``all power resides with the set of all users".

One notable property of Algorand is that its transaction history may fork only with very small probability (e.g., one in a trillion, that is, or even $10^{-18}$). Algorand can also address some legal and political concerns.

The Algorand approach  applies to blockchains and, more generally, to any method of generating a tamperproof sequence of blocks. We actually put forward a new  method ---alternative to, and more efficient than, blockchains--- that may be of independent interest.

\subsection{Bitcoin's  Assumption and Technical Problems}\label{BitcoinProblems}

Bitcoin is a very ingenious  system and has inspired a great amount of subsequent research. Yet, it  is also problematic.  Let us summarize  its underlying assumption  and  technical problems ---which are actually shared by essentially all cryptocurrencies that, like Bitcoin, are based on {\em proof-of-work}.

For this summary, it suffices to recall that, in Bitcoin, a user may own multiple public keys of a digital signature scheme, that money is associated with public keys, and that a payment is a digital signature that transfers some amount of money  from one public key to another. Essentially,  Bitcoin organizes all processed payments in a chain of blocks, $B_1,B_2,\ldots$,
each consisting of multiple payments, such that,
all payments of $B_1$, taken in any order,
followed by those of $B_2$, in any order, etc., constitute a sequence of valid
payments. Each block is generated, on average, every 10 minutes.

This sequence of blocks is  a \emph{chain}, because it is structured so as to ensure that any change, even in a single block,  percolates into all subsequent blocks, making it easier to spot any alteration of the payment history.
(As we shall see, this is achieved by including in each block a {\em cryptographic hash} of the previous one.)
Such block structure is referred to as a \emph{blockchain}.

\paragraph{Assumption: Honest Majority of Computational Power}
Bitcoin  assumes that no malicious entity (nor a coalition of coordinated malicious entities)  controls the majority of the computational power devoted to block generation. Such an entity, in fact, would be able to  modify the blockchain, and thus  re-write the payment history, as it pleases. In particular, it could make a payment $\wp$, obtain the benefits paid for, and then ``erase'' any trace of $\wp$.

\paragraph{Technical Problem 1: Computational Waste}
Bitcoin's proof-of-work approach to block generation requires  an extraordinary amount of computation. Currently, with just a few hundred thousands public keys in the system, the top 500 most powerful supercomputers can only muster a mere 12.8\% percent of the total computational power required from the Bitcoin players.
This amount of computation would greatly increase,  should significantly more users join the system.

\paragraph{Technical Problem 2: Concentration of Power}

Today, due to the exorbitant amount of computation required,  a user, trying to generate a new block using an ordinary desktop (let alone a cell phone), expects to lose money. Indeed, for computing a new block with an ordinary computer, the expected cost of the necessary electricity  to power the computation exceeds the expected reward. Only using \emph{pools} of specially built computers (that do nothing  other than ``mine new blocks''), one might expect to make a profit by generating new blocks.
Accordingly, today there are, {\em de facto},  two disjoint classes of users: ordinary users, who only make payments, and specialized mining pools, that only search for new blocks.

It  should therefore not be a surprise that, as of recently, the total computing
power for block generation lies within just five pools.
In such conditions,
the assumption that a majority of the computational power is honest becomes
less credible.

\paragraph{Technical Problem 3:  Ambiguity}
In Bitcoin,  the blockchain is not necessarily unique. Indeed its latest portion often {\em forks}: the blockchain may be ---say---
$B_1,\ldots,B_k,B_{k+1}',B_{k+2}'$, according to one user, and
$B_1,\ldots,B_k,B_{k+1}'',B_{k+2}'',B_{k+3}''$ according another user.
Only after several blocks have been added to the chain, can one be reasonably sure that the first $k+3$ blocks will be the same for all users.
Thus, one cannot rely right away on the payments contained in the last block of the chain. It is more prudent to wait and see whether the block becomes sufficiently deep  in the blockchain and thus sufficiently stable.

\bigskip
Separately, {\em law-enforcement} and {\em monetary-policy} concerns have also been raised about Bitcoin.%
\footnote{The (pseudo) anonymity offered by Bitcoin payments may be misused for money laundering and/or the financing of criminal individuals or terrorist organizations.
Traditional banknotes or gold bars, that in principle offer perfect anonymity, should pose the same challenge, but the physicality of these currencies substantially slows down money transfers, so as to  permit some degree of monitoring by law-enforcement agencies.

The ability to ``print money'' is one of the very basic powers of a nation state. In principle, therefore, the massive adoption of an independently floating currency may curtail this power.
Currently, however, Bitcoin is far from being a threat to governmental monetary policies, and, due to its scalability problems,  may never be.}

\subsection{Algorand, in a Nutshell}

\paragraph{Setting}
Algorand works in a very tough setting.
Briefly,

\begin{itemize}

\item[(a)] {\em Permissionless and Permissioned Environments.} Algorand works efficiently  and securely even in a totally permissionless environment, where arbitrarily many users are allowed to join the system at any time, without any vetting or permission of any kind. Of course, Algorand works even better in a {\em permissioned} environment.

\item[(b)] {\em Very Adversarial Environments.} Algorand withstands a very powerful Adversary, who can

    (1) {\em instantaneously} corrupt {\em any user} he wants, at {\em any time} he wants, provided  that, in a permissionless environment, 2/3 of the money in the system belongs to honest user. (In a permissioned environment, irrespective of money, it suffices that 2/3 of the users are honest.)

    (2) {\em totally control and perfectly  coordinate} all corrupted users; and

    (3) {\em schedule the delivery of all messages,} provided that each message $m$ sent by a honest user reaches 95\% of the honest users within
    a time $\lambda_m$, which solely depends on the size of $m$.

\end{itemize}

\paragraph{Main Properties}

Despite the presence of our powerful adversary, in Algorand
\begin{itemize}
  \item {\em The amount of computation required is minimal.} Essentially, no matter how many users are present in the system, each of fifteen hundred users must perform at most a few seconds of computation.

      \item {\em A New Block is Generated in less than 10 minutes, and will {\em de facto} never leave the blockchain.}
For instance, in expectation, the time to generate a block in the first embodiment is less than $\Lambda +12.4\lambda$, where $\Lambda$ is  the time necessary to propagate a block, in a peer-to-peer gossip fashion, {\em no matter what  block size one may choose}, and $\lambda$ is the time to propagate 1,500 200B-long messages.
(Since in a truly decentralized system, $\Lambda$ essentially is an intrinsic latency, in Algorand the limiting factor in block generation  is network speed.)
{\em The second embodiment has actually been tested experimentally ( by ?), indicating that a block is generated in less than 40 seconds.}

In addition, Algorand's blockchain {\em may fork only with negligible probability} (i.e.,  less than one in  a trillion), and thus users can relay on the payments contained in a new block as soon as the block appears.

\item {\em All power resides with the users themselves.} Algorand is a truy distributed system. In particular, there are no exogenous entities  (as the ``miners" in Bitcoin), who can control which transactions are recognized.

\end{itemize}

\paragraph{Algorand's Techniques.}$ $

\scparagraph{1. A New and Fast Byzantine Agreement Protocol.}
 Algorand generates a new block via a new cryptographic,  {\em message-passing},  binary Byzantine agreement (BA) protocol, $\BA*$. Protocol
 $\BA*$ not only satisfies some additional properties (that we shall soon discuss), but is also very fast. Roughly said, its binary-input version consists of a 3-step loop, in which a player $i$ sends a {\em single} message $m_i$ to all other players. Executed in a complete and synchronous network, with more than 2/3 of the players being honest,  with probability $>1/3$, after each loop the protocol ends in agreement. (We stress that protocol $\BA*$ satisfies the original definition of Byzantine agreement of Pease, Shostak, and Lamport \cite{PSL}, without any weakenings.)

Algorand leverages this binary BA protocol to reach agreement, in our different communication model, on each new block. The agreed upon block is then {\em certified}, via a prescribed number of digital signature of the proper verifiers, and propagated through the network.

\medskip

\scparagraph{2. Cryptographic Sortition.}
Although very fast,  protocol $\BA*$ would benefit from further speed when played by millions of users.
Accordingly, Algorand chooses the players of $\BA*$  to be a much smaller subset of the set of all users. To avoid a different kind of concentration-of-power problem, each new block $B\upr$ will be constructed and agreed upon, via a new execution of $\BA*$, by a separate set of {\em selected verifiers},~$SV\upr$.
In principle,
selecting such a set might be as hard as
selecting $B\upr$ directly. We traverse this potential problem by an approach that we term, embracing the insightful suggestion of Maurice Herlihy, {\em cryptographic sortition}. Sortition is the practice of selecting officials at random from a large set of eligible individuals \cite{Wiki}.  (Sortition was practiced across centuries: for instance,  by the  republics of Athens, Florence, and Venice.
In modern judicial systems, random selection is often used to choose juries. Random sampling  has also been recently advocated for elections by David Chaum \cite{ChaumRSE}.) In a decentralized system, of course, choosing the random coins necessary
to randomly select the members of each verifier set $SV\upr$ is problematic. We thus resort to cryptography in order to select each verifier set, from the population of all users, in a way that is guaranteed to be automatic (i.e., requiring no message exchange) and random. In essence, we use a
cryptographic function to automatically determine, from the previous block $B\uprm$,  a user, the {\em leader}, in charge of proposing the new block $B\upr$, and the verifier set $SV\upr$, in charge to reach agreement on the block proposed by the leader. Since malicious users can affect the composition of $B\uprm$ (e.g., by choosing some of its payments), we  specially construct and use additional inputs so as to prove
that the leader for the $r$th block and the verifier set $SV\upr$ are indeed randomly chosen.

\medskip

\scparagraph{3. The Quantity (Seed) ${Q\upr}$.}
We use the  the last block $B\uprm$ in the blockchain in order to automatically determine the next verifier set and leader in charge of constructing the new block $B\upr$.
The challenge with this approach is that, by just choosing a slightly different payment in the previous round, our powerful  Adversary gains a tremendous control over the next leader. Even if he only controlled only 1/1000 of the players/money in the system, he could ensure that all leaders are malicious. (See the Intuition Section \ref{sec:BasicAlgorand}.)
This challenge is central to all proof-of-stake approaches, and, to the best of our knowledge, it has not, up to now, been satisfactorily solved.

To meet this challenge, we purposely construct, and continually update, a separate and carefully defined  quantity, $Q\upr$,
which {\em provably} is, not only unpredictable, but also not influentiable, by our powerful Adversary.
We may refer to $Q\upr$ as the $r$th {\em seed}, as it is from $Q\upr$ that Algorand selects, via secret cryptographic sortition, all the users that will play a special role in the generation of the $r$th block.

\medskip

\scparagraph{4. Secret Crytographic Sortition and Secret Credentials.}
Randomly and  unambiguously using the current last block, $B\uprm$, in order to choose the verifier set and the leader in charge of constructing the new block, $B\upr$, is not enough. Since  $B\uprm$ must be known before generating $B\upr$, the last non-influentiable quantity $Q\uprm$ contained in $B\uprm$ must be known too. Accordingly, so are the verifiers and the leader in charge to compute the block $B\upr$. Thus, our powerful  Adversary might immediately corrupt {\em all of them},
before they engage in any discussion about $B\upr$, so as to get full control over the block they certify.

To prevent this problem, leaders (and actually verifiers too) {\em secretly} learn of their role, but can compute a proper  {\em credential}, capable of proving to everyone that indeed have that role.
When a user privately realizes that he is the leader for the next block, first he secretly assembles his own proposed  new block, and then disseminates it (so that can be certified) together with  his own credential.
This way, though the Adversary will immediately realize who the leader of the next block is, and although he can corrupt him right away, it will be too late for the Adversary to influence the choice of a new block. Indeed, he cannot ``call back" the leader's message no more than a powerful government can put back into the bottle a message virally spread by WikiLeaks.

As we shall see, we cannot guarantee leader uniqueness, nor that everyone is sure who the leader is, including the leader himself! But, in Algorand,  unambiguous progress will be guaranteed.

\medskip

\scparagraph{5. Player Replaceability.}
After he proposes a new block, the leader might as well ``die" (or be corrupted by the Adversary), because his job is done.  But, for the verifiers in $SV\upr$, things are less simple. Indeed, being in charge of certifying the new block $B\upr$ with sufficiently many signatures, they must first run Byzantine agreement on the block proposed by the leader. The problem is that, no matter how efficient it is,  $\BA*$ requires {\em multiple} steps and the honesty of $>2/3$ of its players. This is a problem, because, for efficiency reasons,  the player set of $\BA*$ consists  the small set $SV\upr$ randomly selected among the set of all users. Thus, our powerful Adversary, although unable to corrupt  1/3 of {\em  all the users}, can certainly corrupt {\em all members of $SV\upr$}!

Fortunately we'll prove that protocol $\BA*$, executed by propagating messages in a peer-to-peer fashion, is {\em player-replaceable}. This novel requirement means that the protocol correctly and efficiently reaches consensus even if each of its step is executed by a totally new, and randomly and independently selected,  set of players. Thus, with millions of users, each small set of players associated to a step of $\BA*$ most probably has empty intersection with the next set.

In addition, the sets of players of different steps of $\BA*$ will probably have  totally different {\em cardinalities}. Furthermore, the members of each set do not know who the next set of players will be, and do not secretly pass any internal state.

The replaceable-player property is actually crucial to defeat the dynamic and very powerful Adversary we envisage.
We believe that replaceable-player protocols will prove crucial in lots of contexts and applications. In particular, they will be crucial to execute securely small sub-protocols embedded in a larger universe of players with a dynamic adversary, who, being able to corrupt even a small fraction of the total players, has no difficulty in corrupting all the players in the smaller sub-protocol.

\paragraph{An Additional Property/Technique: Lazy Honesty}

A honest user follows his prescribed instructions, which include being online and run the protocol. Since, Algorand has only modest computation and communication requirement, being online and running the protocol ``in the background" is not a major sacrifice. Of course, a few ``absences" among honest players, as those due to sudden loss of connectivity or the need of rebooting, are automatically tolerated (because we can always consider such few players to be temporarily malicious). Let us point out, however, that Algorand can be simply adapted so as to work in a new model, in which honest users to be offline most of the time. Our new model can be informally introduced as follows.

\begin{itemize}
  \item[] {\em Lazy Honesty.} Roughly speaking, a user $i$ is lazy-but-honest if (1) he follows all his prescribed
instructions, when he is asked to participate to the protocol, and (2) he is asked to participate
to the protocol only  rarely, and with a suitable advance notice.
\end{itemize}
With such a relaxed notion of honesty, we may be even more confident that honest people will be at hand when we need them, and Algorand guarantee that, when this is the case,
\begin{newcenter}
  {\em The system operates securely even if,  at a given point in time,\\
  the majority of the participating players are malicious.}
\end{newcenter}

\subsection{Closely Related work}

Proof-of-work approaches (like the cited \cite{bitcoin} and \cite{Ethereum}) are quite orthogonal to our ours. So are the approaches based on message-passing Byzantine agreement or practical Byzantine fault tolerance (like the cited \cite{PBFT}). Indeed, these protocols cannot be run among the set of all users and cannot, in our model, be restricted to a suitably small set of users. In fact, our powerful adversary my immediately corrupt all the users involved in a small set charged to actually running a BA protocol.

Our approach could be considered related to proof of stake~\cite{pos}, in the sense that users' ``power" in block building is proportional to {\em the money they own in the system} (as opposed to ---say--- to the money they have put in ``escrow").

The paper closest to ours is the Sleepy Consensus Model of Pass and Shi \cite{PS}. To avoid the heavy computation required in the proof-of-work approach, their paper relies upon (and kindly credits) Algorand's secret cryptographic sortition. With this crucial aspect in common, several significant differences exist between our papers.
In particular,

(1) {\em Their setting is only permissioned}.
  By contrast, Algorand is also a permissionless system.

(2) {\em They use a Nakamoto-style protocol, and thus their blockchain  forks frequently}.  Although dispensing with proof-of-work,  in their protocol a secretly selected leader is  asked to elongate
the longest valid (in a richer sense) blockchain. Thus, forks are unavoidable and  one has to wait that the block is sufficiently ``deep" in the chain. Indeed, to achieve their goals with an adversary capable of adaptive corruptions, they require a block to be $poly(N)$ deep, where $N$ represents the total number of users in the system. Notice that, even assuming that a block could be produced in a minute, if there were $N=1M$ users, then one would have to wait for about 2M years for a block to become $N^2$-deep, and for about 2 years for a block to become $N$-deep.
By contrast, Algorand's blockchain forks only with negligible probability, even though the Adversary corrupt users immediately and adaptively, and its new blocks can immediately be relied upon.

(3) {\em They do not handle individual Byzantine agreements}.  In a sense, they only guarantee ``eventual consensus on a growing sequence of values". Theirs is a {\em state replication} protocol, rather than a BA one, and cannot be used to reach Byzantine agreement on an individual value of interest. By contrast, Algorand can also be used only once, if so wanted, to enable millions of users to quickly reach Byzantine agreement on a specific value of interest.

(4) {\em They require weakly synchronized clocks.}
      That is,  all users' clocks are offset by a small time $\delta$.
      By contrast, in Algorand, clocks need only have (essentially) the same ``speed".

(5) {\em Their protocol works with lazy-but-honest users or with honest majority of online users.}
They kindly credit Algorand for raising the issue of honest users going offline en masse, and for putting forward the lazy honesty model in response.
Their protocol not only works in the lazy honesty model, but also in their {\em adversarial sleepy model}, where an adversary chooses which users are online and which are offline, provided that, at all times, the majority of online users are honest.%
\footnote{The original version of their paper actually considered only security in their adversarial sleepy model. The original version of Algorand, which precedes theirs,  also explicitly  envisaged assuming that a given majority of the online players is always honest, but  explicitly excluded it from consideration, in favor of the lazy honesty model. (For instance, if at some point in time half of the honest users choose to go off-line, then the majority of the users on-line may very well be malicious. Thus, to prevent this from happening, the Adversary should {\em force} most of his corrupted players to go off-line too, which clearly is against his own interest.)  Notice that a protocol with a majority of lazy-but-honest players works just fine if the majority of the users on-line are always malicious. This is so, because a sufficient number of honest players, knowing that they are going to be crucial at some rare point in time, will elect not to go off-line in those moments, nor can they be forced off-line by the Adversary, since he does not know who the crucial honest players might be.}

(6) {\em They require a simple honest majority.} By contrast, the current version of Algorand requires a 2/3 honest majority.

Another paper close to us is {\em Ouroboros: A Provably Secure Proof-of-Stake Blockchain Protocol}, by
Kiayias, Russell, David, and Oliynykov \cite{Kiayias+}.
Also their system appeared after ours. It also uses crytpographic sortition to dispense with proof of work in a provable manner. However, their system is, again,  a Nakamoto-style protocol, in which forks are both unavoidable and frequent. (However, in their model, blocks need not as deep as the sleepy-consensus model.) Moreover, their system relies on the following assumptions: in the words of the authors themselves,
``(1) the network is
highly synchronous,
(2) the majority of the selected stakeholders is available as needed to participate
in each epoch,
(3) the stakeholders do not remain offline for long periods of time,
(4) the adaptivity
of corruptions is subject to a small delay that is measured in rounds linear in the security parameter."
By contrast,  Algorand is, with overwhelming probability, fork-free, and does not rely on any of these 4 assumptions.
In particular, in Algorand, the Adversary is able to instantaneously corrupt the users he wants to control.

\section{Preliminaries}\label{sec:Preliminaries}

\subsection{\Large{Cryptographic Primitives}}\label{sec:Crypto}

\paragraph{Ideal Hashing.} We shall rely on an efficiently computable cryptographic hash function, $H$,
that maps arbitrarily long strings to binary strings of fixed length. Following a long tradition, we model  $H$ as a {\em random oracle}, essentially a function mapping each possible string $s$ to a randomly and independently selected (and then fixed)  binary string, $H(s)$, of the chosen length.

In this paper,  $H$ has 256-bit long outputs. Indeed, such length is short enough to make the system efficient and long enough to make the system secure.
For instance, we want  $H$ to be {\em collision-resilient}. That is, it should be hard to find two different strings $x$ and $y$ such that $H(x)=H(y)$. When $H$ is a random oracle with 256-bit long outputs,  finding any such  pair of strings is indeed difficult. (Trying at random, and relying on the birthday paradox, would require $2^{256/2}=2^{128}$ trials.)

\paragraph{Digital Signing.}

Digital signatures allow users to to authenticate information to each other without sharing any sharing any secret keys.
A  {\em digital signature scheme} consists of three fast algorithms:
a probabilistic {\em key generator} $G$, a  {\em signing algorithm} $S$, and  a {\em verification algorithm} $V$.

Given a security parameter $k$, a sufficiently high integer, a user $i$ uses $G$ to produce a pair of $k$-bit keys (i.e., strings): a ``public" key $pk_i$ and a matching ``secret"
signing key $sk_i$. Crucially, a public key does not ``betray" its corresponding secret key. That is, even given knowledge of $pk_i$, no one other than $i$ is able to compute $sk_i$  in less than astronomical time.

User $i$ uses  $sk_i$ to digitally sign messages. For each possible message (binary string) $m$,
 $i$ first hashes $m$  and then runs algorithm $S$ on inputs $H(m)$ and $sk_i$ so as to produce the $k$-bit string

 $$ sig_{pk_i}(m)\bydef S(H(m),sk_i)\enspace.%
 \footnote{Since $H$ is collision-resilient it is practically impossible that, by signing $m$ one ``accidentally signs" a different message $m'$.} $$
The binary string $sig_{pk_i}(m)$ is referred to as $i$'s digital signature of $m$ (relative to $pk_i$), and can be more simply denoted by $sig_i(m)$, when the public key $pk_i$ is clear from context.

Everyone knowing $pk_i$ can use it to verify the digital signatures produced by $i$.
Specifically, on inputs (a) the public key $pk_i$  of a player $i$, (b) a message $m$, and (c) a string $s$, that is, $i$'s alleged digital signature of the message $m$, the verification algorithm $V$ outputs either YES or NO.

The properties we require from a digital signature scheme are:
\begin{itemize}

\item[1.] {\em Legitimate signatures are always verified:}  If $s=sig_i(m)$, then $V(pk_i, m,s)=YES$; and

\item[2.] {\em Digital signatures are  hard to forge:} Without knowledge of $sk_i$ the  time to find a string $s$ such that
$V(pk_i, m, s)=YES$, for a message $m$ never signed by $i$, is astronomically long.

(Following the strong security requirement of Goldwasser, Micali, and Rivest \cite{GMR}, this is true even if one can obtain the signature of any other message.)

\end{itemize}
Accordingly, to prevent anyone else from signing messages on his behalf, a player $i$ must keep his signing key $sk_i$ secret (hence the term ``secret key"), and to enable anyone to verify the messages he does sign, $i$ has an interest in publicizing his key $pk_i$  (hence the term ``public key").

In general, a message $m$ is not retrievable from its signature $sig_i(m)$.  In order to virtually deal with digital signatures
 that satisfy the conceptually convenient {\em ``retrievability" property} (i.e., to guarantee that the signer and the  message are easily computable from a signature, we  define
$$SIG_{pk_i}(m)=(i, m,sig_{pk_i}(m)) \quad \text{ and } \quad SIG_i(m)=(i,m,sig_i(m)), \text{ if  $pk_i$ is clear}.$$

\paragraph{Unique Digital Signing.}

We also consider digital signature schemes $(G,S,V)$ satisfying the following additional property.

\begin{itemize}

\item[3.] {\em  Uniqueness.} It is hard to find  strings $pk'$, $m$, $s$, and $s'$  such that
    \begin{newcenter}
      $s\neq s'$ \quad and \quad $V(pk',m,s)=V(pk',m,s')=1$.
    \end{newcenter}

(Note that the uniqueness property holds also for strings $pk'$ that are not legitimately generated public keys. In particular, however, the uniqueness property implies
    that, if one used the specified key generator $G$ to compute a  public key $pk$ together with a matching secret key $sk$, and thus knew $sk$, it would be essentially impossible also for him to find two different digital signatures of a same message relative to $pk$.)
\end{itemize}

\paragraph{Remarks}

\begin{itemize}
  \item {\sc From Unique  signatures to verifiable random functions.}
      Relative to a digital signature scheme with the uniqueness property, the mapping $m\rightarrow H(sig_i(m))$
associates to each possible string $m$, a unique, randomly selected, 256-bit string, and the correctness of this mapping can be proved given the signature $sig_i(m)$.

That is, ideal hashing and digital signature scheme satisfying the uniqueness property essentially provide an elementary implementation of a {\em verifiable random function}, as introduced and by Micali, Rabin, and Vadhan \cite{VRF}. (Their original implementation was necessarily more complex, since they did not rely on ideal hashing.)

\item {\sc Three different needs for digital signatures.}
In Algorand, a user $i$ relies on digital signatures for

 (1) {\em Authenticating $i$'s own payments.}
  In this application, keys can be ``long-term" (i.e., used to sign many messages over a long period of time) and come from a ordinary signature scheme.

(2) {\em Generating  credentials proving that $i$ is entitled to act at some step $s$ of a round $r$.}
      Here, keys can  be long-term, but must come from a scheme satisfying the uniqueness property.

    (3) {\em Authenticating the message $i$ sends in each step in which he acts.}
          Here, keys must be ephemeral (i.e., destroyed after their first use), but can come from an ordinary signature scheme.

  \item{\sc A small-cost simplification.}
For simplicity, we envision each user $i$ to have a single long-term key. Accordingly, such a key must come from a signature scheme with the uniqueness property. Such simplicity has  a small  computational cost. Typically, in fact, unique digital signatures are slightly more expensive to produce and verify than ordinary signatures.

\end{itemize}

\subsection{The Idealized Public Ledger}\label{IdealizedScheme}

 Algorand tries to mimic the following payment system, based on an {\em idealized public ledger}.

\begin{itemize}

\item[1.]
{\em The Initial Status.}
Money is associated with individual public keys (privately generated and  owned by users). Letting $pk_1,\ldots,pk_j$ be the initial public keys and $a_1,\ldots,a_j$ their respective initial
amounts of money units,
then the {\em initial status} is $$S_0=(pk_1,a_1),\ldots,(pk_j,a_j)\enspace,$$
 which is assumed to be common knowledge in the system.

\item[2.] {\em Payments.}
Let $pk$ be a public key currently having $a\geq 0$ money units, $pk'$ another public key, and $a'$ a non-negative number no greater than $a$. Then, a (valid) payment $\wp$ is a digital signature, relative to $pk$, specifying the transfer of $a'$ monetary units from $pk$ to $pk'$, together with some additional information. In symbols,

$$\wp=SIG_{pk}(pk,pk',a',I,H(\cI)),$$

where $I$ represents any additional information deemed useful but not sensitive (e.g., time information and a payment identifier), and $\cI$ any additional information deemed sensitive (e.g., the reason for the payment, possibly the identities of the owners of $pk$ and the $pk'$, and so on).

We refer to $pk$ (or its owner) as the {\em payer},  to each $pk'$ (or its owner) as a {\em payee}, and to $a'$ as the {\em  amount} of the payment $\wp$.

{\bf Free Joining Via Payments.} Note that users may join the system whenever they want by generating their own public/secret key pairs. Accordingly, the public key $pk'$ that appears in the payment $\wp$ above may be a newly generated public key that had never ``owned'' any money before.

\item[3.]
{\em The Magic Ledger.}
In the Idealized System,  all payments are valid and appear in a tamper-proof list $L$ of sets of payments ``posted on the sky'' for everyone to see:
$$L= PAY^1, PAY^2,\ldots, $$
 Each block $PAY\uprp$ consists of the set of all payments made since the appearance of block $PAY\upr$. In the ideal system, a new block  appears after a fixed (or finite) amount of time.
\end{itemize}

\paragraph{Discussion.}

\begin{itemize}

\item
{\em More General Payments and Unspent Transaction Output.}
More generally, if a public key $pk$ owns an amount $a$, then a valid payment $\wp$ of $pk$ may transfer the amounts $a_1',a_2',\ldots$, respectively to the keys $pk_1',pk_2',\ldots$, so long as $\sum_ja_j'\leq a$.

In Bitcoin and similar systems, the money owned by a public key $pk$ is  segregated into separate amounts,  and a payment $\wp$ made by $pk$ must transfer such a segregated amount $a$ in its entirety. If $pk$ wishes to transfer only a fraction  $a'<a$ of $a$ to another key, then it must also transfer the balance,  the {\em unspent transaction output}, to another key, possibly $pk$ itself.

Algorand also works with keys having segregated amounts. However, in order to focus on the novel aspects of Algorand, it is conceptually simpler to stick to our simpler forms of payments and keys having a single amount associated to them.

\item {\em Current Status.}
The Idealized Scheme does not directly provide information about the current status of the system (i.e., about how many money units each public key has).
This  information is deducible from the  Magic Ledger.

In the ideal system, an active user continually stores and updates the latest status information, or he  would otherwise have to reconstruct it, either  from scratch, or from the last time he computed it.
(In the next version of this paper, we shall augment Algorand so as to enable its users to reconstruct the current status in an efficient manner.)

\item {\em Security and ``Privacy".} Digital signatures guarantee that no one can forge a
payment by another user. In a payment $\wp$, the public keys and the amount are not hidden, but the sensitive information $\cI$ is. Indeed, only $H(\cI)$ appears in $\wp$, and since $H$ is an ideal hash function,  $H(\cI)$ is  a random 256-bit value,
and thus there is no way to figure out what $\cI$ was better than by simply guessing it. Yet, to prove what $\cI$ was (e.g., to prove the reason for the payment)
the payer may just reveal  $\cI$. The correctness of the revealed $\cI$ can be verified by computing $H(\cI)$ and comparing the resulting value with the last item of $\wp$.
In fact,  since $H$ is {\em collision resilient,}  it is hard to find a second value $\cI'$ such that $H(\cI)=H(\cI')$.
\end{itemize}

\subsection{Basic Notions and  Notations}\label{sec:Notions}

\paragraph{Keys, Users, and Owners}

Unless otherwise specified, each public key (``key" for short)
is long-term and relative to a digital signature scheme with the uniqueness property. A public key $i$ joins the system when another public key $j$ already in the system
makes a payment to $i$.

For color, we personify keys. We refer to a key $i$ as a ``he", say that $i$ is honest, that $i$ sends and receives messages, etc. {\em User} is a synonym for key.
When we want to distinguish a key  from the person to whom it belongs, we respectively  use the term ``digital key" and ``owner".

\paragraph{Permissionless and Permissioned Systems.}

A system is {\em permissionless}, if a digital key is free to join at any time and an owner can own multiple digital keys; and its {\em permissioned}, otherwise.

\paragraph{Unique Representation}
Each object  in Algorand has a unique representation. In particular, each set $\{(x,y,z,\ldots):x\in X, y\in Y, z\in Z,\ldots\}$ is ordered in a pre-specified  manner: e.g., first lexicographically in $x$, then in $y$, etc.

\paragraph{Same-Speed Clocks}
There is no global clock: rather, each user has his own clock. User clocks need not be synchronized in any way.  We assume, however, that they all have the same speed.

For instance, when it is  12pm according to the clock of a user $i$, it may be 2:30pm according to the clock of another user $j$, but when it will be 12:01 according to $i$'s clock, it will 2:31 according to $j$'s clock. That is, ``one minute is the same (sufficiently, essentially the same) for every user".

\paragraph{Rounds}
Algorand  is organized in logical units, $r=0,1,\ldots$, called {\em rounds}.

We consistently use  superscripts to indicate rounds. To indicate that a non-numerical quantity~$Q$ (e.g., a string, a public key, a set, a digital signature, etc.) refers to a round $r$, we simply write $Q^r$. Only when $Q$ is a genuine number (as opposed to a binary string interpretable as a number), do we write $Q^{(r)}$, so that  the symbol $r$ could not be interpreted as the exponent of $Q$.

At (the start of a) round $r>0$,  the set of all  public keys is $PK\upr$, and the  system status is
$$  S\upr=\left\{\left(i,a_i^{(r)},\ldots\right):i\in PK\upr\right\},$$
where  $a_i^{(r)}$ is the amount of money available to the public key $i$. Note that $PK\upr$ is deducible from $S\upr$,
and that $S\upr$ may also specify other components for each public key $i$.

 For round 0, $PK^0$ is the set of {\em initial public keys}, and $S^0$ is {\em the initial status}. Both $PK^0$ and $S^0$ are assumed to be common knowledge in the system. For simplicity, at the start of round $r$, so are $PK^1,\ldots,PK\upr$ and $S^1,\ldots,S\upr$.

In a round $r$, the system status transitions from $S\upr$ to $S\uprp$: symbolically,

\begin{newcenter}
Round $r$: $S^r \longrightarrow S^{r+1}$.
\end{newcenter}

\paragraph{Payments} In Algorand, the users continually make payments (and disseminate them in the way described in subsection \ref{CommunicationModel}). A payment $\wp$ of a user $i\in PK\upr$ has  the same format and semantics as in the Ideal System. Namely,
$$\wp=SIG_{i}( i,i',a,I, H(\cI))\enspace.$$

Payment  $\wp$  is {\em individually  valid  at a round $r$} (is a {\em round-$r$ payment}, for short)  if (1) its  amount $a$ is less than or equal to $a^{(r)}_i$, and (2) it does not appear in any official payset $PAY^{r'}$  for $r'<r$.  (As explained below, the second condition means that $\wp$ has not already become effective.

A set of round-$r$ payments of $i$ is {\em collectively valid}
if the sum of their amounts  is at most $ a^{(r)}_i$.

\paragraph{Paysets} A round-$r$ {\em payset} $\cP$ is a set of round-$r$ payments such that, for each user $i$, the payments of $i$ in $\cP$ (possibly none) are collectively valid. The set of all round-$r$ paysets is $\PAY (r)$. A round-$r$ payset $\cP$ is {\em maximal} if no superset of $\cP$ is a round-$r$ payset.

We actually suggest that a payment $\wp$ also specifies a round $\rho$, $\wp=SIG_{i}(\rho, i,i',a,I, H(\cI))\enspace$, and cannot be  valid at any round outside $[\rho,\rho+k]$, for some fixed non-negative integer $k$.%
\footnote{This simplifies checking whether $\wp$ has become ``effective" (i.e., it simplifies determining whether some payset $PAY\upr$ contains $\wp$. When $k=0$, if $\wp=SIG_{i}( r, i,i',a,I, H(\cI))\enspace$, and $\wp\notin PAY\upr$, then $i$ must re-submit $\wp$.}

\paragraph{Official Paysets}

For every round $r$, Algorand publicly selects (in a manner described later on) a single (possibly empty)  payset, $PAY\upr$,  the round's {\em official payset}.
(Essentially, $PAY\upr$ represents the round-$r$ payments that have {\em ``actually"} happened.)

As in the Ideal System (and Bitcoin), (1) the only way for a new user $j$ to enter
the system is to be the recipient of a payment belonging to the official payset $PAY\upr$ of a given round $r$; and (2)
$PAY\upr$ determines the status of the next round, $S\uprp$, from that of the current round, $S\upr$.
Symbolically,
\begin{newcenter}
$PAY^r: S^r \longrightarrow S^{r+1}$.
\end{newcenter}
Specifically,
\begin{newitemize}

\item[1.]
the set of public keys of round $r+1$, $PK^{r+1}$, consists of the union of $PK\upr$ and the set of all payee keys that appear, for the first time, in the payments of $PAY\upr$; and

 \item[2.]
the amount of money $a_i^{(r+1)}$ that a user $i$ owns in round $r+1$ is the sum of $a_i{(r)}$ ---i.e., the amount of money $i$ owned in the previous round (0 if $i\not\in PK\upr$)--- and the sum of amounts paid to $i$ according to
the payments of $PAY\upr$.

\end{newitemize}
In sum, as in the Ideal System, each status $S\uprp$ is deducible from the
previous payment history:
\begin{newcenter}
  $PAY^0,\ldots,PAY\upr$.
\end{newcenter}

\subsection{Blocks and Proven Blocks}

In $\algzero$,
the block $B^r$ corresponding to a round $r$ specifies: $r$ itself; the set of payments of round $r$, $PAY\upr$; a quantity $Q\upr$, to be explained, and the hash of the previous block, $H(B\uprm)$.
Thus, starting from some fixed block $B^0$, we have a traditional blockchain:
$$B^1=(1,PAY^1,Q^0, H(B^0)), \quad B^2=(2,PAY^2,Q^1, H(B^1)), \quad B^3=(3,PAY^3,Q^2, H(B^2)), \quad \ldots $$

In Algorand, the authenticity of a block is actually  vouched by a separate piece of information, a ``block certificate" $CERT\upr$, which turns $B^r$ into a {\em proven block}, $\overline{B^r}$.
The Magic Ledger, therefore, is  implemented by the sequence of the proven blocks,
$$\overline{B^1},\overline{B^2},\ldots$$

\paragraph{Discussion}

As we shall see, $CERT\upr$ consists of a set of digital signatures for $H(B\upr)$, those of a majority of the members of $SV\upr$, together with a proof that each of those members indeed belongs to $SV\upr$.
We could, of course, include the certificates $CERT\upr$ in the blocks themselves, but find it conceptually cleaner to keep it separate.)

In Bitcoin each block  must satisfy a special property, that is, must ``contain a solution of a crypto puzzle", which makes block generation computationally intensive and forks both inevitable and not rare.
By contrast, Algorand's blockchain has two main advantages: it is generated with minimal computation, and it will not fork with overwhelmingly high probability. Each block~$B^i$ is safely {\em final} as soon as it enters the blockchain.

\subsection{Acceptable Failure Probability}
To analyze  the security of Algorand we specify the probability, $F$, with which we are willing to accept that something goes wrong (e.g., that a verifier set $SV\upr$ does not have an honest majority). As in the case of the output length of the cryptographic hash function $H$, also $F$ is a parameter. But, as in that case,   we find it useful to set $F$ to a concrete value, so as to get a more intuitive grasp of the fact that it is indeed possible, in Algorand,  to enjoy simultaneously sufficient security and sufficient efficiency. To emphasize that $F$ is parameter that can be set as desired, in the first and second embodiments we respectively set
$$F=10^{-12} \quad \text{and} \quad F=10^{-18}\enspace.$$

\paragraph{Discussion}

Note that $10^{-12}$ is actually less than one in a trillion, and we believe that such a choice of $F$ is adequate in our application. Let us emphasize  that $10^{-12}$ is {\em not} the probability with which the Adversary can forge the payments of an honest user. All payments are digitally signed, and thus, if the proper digital signatures are used, the probability of forging a payment is far lower than $10^{-12}$, and is, in fact, essentially 0. The bad event that we are willing to tolerate with probability $F$ is that Algorand's  blockchain {\em forks}.
Notice that, with our setting of $F$  and one-minute long rounds, a fork is expected to occur in Algorand's blockchain  as infrequently as (roughly) once in 1.9 million years. By contrast, in Bitcoin, a forks occurs quite often.

A  more demanding person may  set $F$ to a  lower value.
To this end, in our second embodiment we consider setting $F$ to
$10^{-18}$. Note that, assuming that a block is generated every
{\em second}, $10^{18}$ is the estimated number of seconds taken by the Universe so far: from the Big Bang to present time.
Thus, with $F=10^{-18}$, if a block is generated in a second, one should expect for the age of the Universe to see a fork.

\subsection{The Adversarial Model}
Algorand is designed to be secure in a very adversarial model.
Let us explain.

\paragraph{Honest and Malicious Users}
A user is {\em honest} if he follows all his protocol instructions, and is perfectly capable of sending and receiving messages. A user is {\em malicious} (i.e., {\em Byzantine}, in the parlance of distributed computing) if he can  deviate arbitrarily from his prescribed instructions.

\medskip

\paragraph{The Adversary} The {\em Adversary} is an efficient (technically polynomial-time)   algorithm, personified for color, who can {\em immediately} make malicious {\em any user}  he wants, at {\em any  time} he wants (subject only to an upperbound to the number of the users he can corrupt).

The Adversary totally controls and perfectly coordinates all malicious users. He takes all actions on their behalf, including receiving and sending all their messages, and can let them deviate from their prescribed instructions in arbitrary ways. Or he can simply isolate a  corrupted user sending and receiving messages.
Let us clarify that no one else automatically learns that a user $i$ is malicious, although $i$'s maliciousness  may transpire by  the actions the Adversary has him  take.

This powerful adversary however,
\begin{itemize}
\item Does not have unbounded computational power and cannot successfully forge the digital signature of an honest user, except with negligible probability; and

 \item Cannot interfere in any way with the messages exchanges among honest users.

 \end{itemize}
 Furthermore, his ability to attack honest users is bounded by one of the following assumption.

\paragraph{Honesty Majority of Money}

We consider a continuum of Honest Majority of Money (HMM) assumptions: namely, for each non-negative integer $k$ and real $h>1/2$,

\begin{itemize}

      \item[] $HHM_{k}>h$: {\em  the honest users in every round $r$ owned a fraction greater than $h$ of all money in the system at round $r-k$}.
\end{itemize}

\paragraph{Discussion.}

Assuming that all malicious users  perfectly coordinate their actions (as if controlled by a single entity, the Adversary) is  a rather pessimistic hypothesis. Perfect coordination among too many individuals is difficult to achieve.
Perhaps  coordination only occurs within separate groups of malicious players. But, since one cannot be sure about the level of coordination malicious users may enjoy, we'd  better  be safe than sorry.

Assuming that the Adversary can secretly, dynamically, and immediately  corrupt  users  is also pessimistic. After all, realistically, taking full control of a user's operations should take some time.

The assumption $HMM_k>h$ implies, for instance, that,
if a round (on average) is implemented in one minute, then,
 the majority of the money at a given round will remain in honest hands for at least two hours, if $k=120$, and at least one week, if $k=10,000$.

Note that the HMM assumptions and the previous Honest Majority of Computing Power assumptions are related in the sense that, since computing power can be bought with money, if  malicious users own most of the money, then they can obtain most of the computing power.

\subsection{The Communication Model}\label{CommunicationModel}

We envisage message propagation ---i.e.,  ``peer-to-peer gossip''%
\footnote{Essentially, as in Bitcoin, when a user propagates  a message $m$, every active user $i$ receiving $m$ for the first time,  randomly and independently selects  a suitably small number  of active users, his {\em ``neighbors"}, to whom he forwards $m$, possibly until he receives an acknowledgement from them.
The propagation of $m$ terminates when no user receives $m$ for the first time.}--- to be the only means of communication.

\paragraph{Temporary Assumption: Timely Delivery of Messages in the Entire Network.}
For most part of this paper we assume that
every propagated message reaches almost all honest users in a timely fashion.
We shall remove this assumption in Section \ref{sec:partition}, where we deal with network partitions,
either naturally occurring or adversarially induced. (As we shall see, we only assume timely delivery of messages within
each connected component of the network.)

One concrete way to capture timely delivery of propagated messages (in the entire network) is the following:

\begin{itemize}
  \item[] {\em  For all
  {\em reachability} $\rho >95\%$ and  {\em message size}  $\mu\in \bZ_+$, there exists  $\lambda_{\rho,\mu}$ such that,

  if a honest user propagates $\mu$-byte message $m$ at  time $t$,

  then $m$ reaches,  by time $t+\lambda_{\rho,\mu}$, at least a fraction $\rho$ of the honest users. }
\end{itemize}
The above property, however, cannot support our Algorand protocol, without explicitly and separately envisaging a mechanism to obtain the latest blockchain ---by another user/depository/etc. In fact, to construct a new block $B\upr$ not only should a proper set of verifiers timely receive round-$r$ messages, but also the messages of previous rounds, so as to know $B\uprm$ and all other previous blocks, which is necessary to determine whether the payments in  $B\upr$ are valid. The following assumption instead suffices.

\paragraph{Message Propagation (MP) Assumption:}

 {\em  For all $\rho >95\%$ and  $\mu \in \bZ_+$, there exists  $\lambda_{\rho,\mu}$ such that, for all times $t$ and all $\mu$-byte messages $m$ propagated by an honest user before $t-\lambda_{\rho,\mu}$, $m$ is received, by time $t$, by at least a fraction $\rho$ of the honest users. }

\smallskip

Protocol $\alg$ actually instructs each of a small number of users (i.e., the verifiers of a given step of a round in $\alg$, to
propagate a separate message of a (small) prescribed size, and we need to bound the time required to fulfill these instructions.  We do so by enriching the MP assumption as follows.

\smallskip

{\em  For all $n$, $\rho >95\%$, and  $\mu \in \bZ_+$, there exists  $\lambda_{n,\rho,\mu}$ such that, for all times $t$ and all $\mu$-byte messages $m_1,\ldots, m_n$, each propagated by an honest user before $t-\lambda_{n,\rho,\mu}$, $m_1,\ldots,m_n$ are received, by time $t$, by at least a fraction $\rho$ of the honest users. }

\paragraph{Note}
\begin{itemize}
  \item
  The above assumption is deliberately simple, but also stronger than needed in our paper.%
\footnote{Given the honest percentage $h$ and the acceptable failure probability $F$,
Algorand computes an upperbound, $N$, to the maximum number of member of verifiers in a step. Thus, the MP assumption need only  hold for $n\leq N$.

In addition,  as stated, the MP assumption  holds no matter how many other messages may be propagated alongside the $m_j$'s. As we shall see, however, in Algorand   messages at are propagated  in essentially non-overlapping time intervals, during which either a single block is propagated, or at most $N$ verifiers propagate a small (e.g., 200B) message. Thus, we could restate the MP assumption in a weaker, but also more complex, way.}

\item
For simplicity, we assume $\rho=1$, and thus drop mentioning $\rho$.

\item
We pessimistically assume that, provided he does not violate the MP assumption,  the Adversary totally controls the delivery of all messages. In particular, without being noticed by the honest users, the Adversary he can arbitrarily decide which honest player receives which message when, and arbitrarily accelerate the delivery of any message he wants.%
      \footnote{For instance, he can immediately learn the messages sent by honest players. Thus,  a malicious user $i'$, who is asked to propagate a message simultaneously with a honest user $i$, can always choose his own message $m'$ based on the message $m$ actually propagated by $i$. This ability is related to {\em rushing}, in the parlance of distributed-computation literature.}

\end{itemize}

\section{The   BA Protocol $\bm{\BA*}$
in a Traditional Setting}\label{app:BA}

As already emphasized, Byzantine agreement is a key ingredient of Algorand. Indeed, it is through the use of such a BA protocol that Algorand is unaffected by forks. However, to be secure against our powerful Adversary, Algorand must rely on a BA protocol that satisfies the new player-replaceability constraint. In addition, for Algorand to be efficient, such a BA protocol must be very efficient.

BA protocols were first defined for an idealized communication model,  {\em synchronous complete networks} (SC networks). Such a  model allows for a simpler design and  analysis of BA protocols. Accordingly,
in this section, we introduce a new BA protocol, $\BA*$, for SC networks and ignoring the issue of player replaceability altogether. The protocol $\BA*$ is a contribution of separate value. Indeed, it is the most efficient cryptographic BA protocol for SC networks known so far.

To use it within our Algorand protocol, we modify $\BA*$ a bit, so as to account for our different communication model and context, but make sure, in section X, to highlight how $\BA*$ is used within our actual protocol $\alg$.

We start by recalling the model in which $\BA*$ operates and the notion of a Byzantine agreement.

\subsection{Synchronous Complete Networks and Matching Adversaries}

In a SC network,  there is a common clock, ticking at each integral times $r=1,2,\ldots$

At each even time click $r$, each player $i$ instantaneously and simultaneously sends a single message $m_{i,j}^r$ (possibly the empty message) to each player $j$, including himself. Each $m_{i,j}^r$ is received at time click $r+1$ by player $j$, together with the identity of the sender $i$.

Again, in a communication protocol, a player is {\em honest}
if he follows all his prescribed instructions, and {\em malicious} otherwise. All malicious players are totally controlled and perfectly coordinated by the Adversary, who, in particular, immediately receives all messages addressed to malicious players, and chooses the messages they send.

The Adversary can immediately make malicious any honest user he wants at any odd time click he wants, subject only to a possible upperbound $t$ to the number of malicious players.
That is, the Adversary  ``cannot interfere with the messages already sent by an honest user $i$", which will be delivered as usual.

The Adversary also has the additional ability to see instantaneously, at each even round, the messages that the currently honest players send, and instantaneously use this information to choose the messages the malicious players send at the same time tick.

\medskip

\paragraph{Remarks}

\begin{itemize}
  \item {\em  Adversary Power.} The above setting is very adversarial.
Indeed, in the Byzantine agreement literature, many settings are less adversarial. However, some more adversarial settings have also been considered, where  the Adversary, after seeing  the messages sent by an honest player $i$ at a given time click $r$, has the ability to erase all these messages from the network, immediately corrupt $i$,  choose the message that the now malicious $i$ sends at time click $r$, and have them delivered as usual.
The envisaged power of the Adversary  matches that he has in our setting.

\item {\em Physical Abstraction.} The envisaged communication model abstracts a more physical model, in which each pair of players $(i,j)$ is linked by a separate and private communication line $l_{i,j}$.  That is, no one else can inject, interfere with, or gain information about  the messages sent over  $l_{i,j}$. The only way for the Adversary to have access to $l_{i,j}$ is to corrupt either $i$ or $j$.

\item {\em Privacy and Authentication.}
In SC networks message privacy and authentication are guaranteed by assumption. By contrast, in our communication network, where messages are propagated from peer to peer, authentication is guaranteed by digital signatures, and privacy is non-existent.
Thus, to adopt protocol $\BA*$ to our setting, each message exchanged should be digitally signed (further identifying the state at which it was sent).
Fortunately, the BA protocols that we consider using in Algorand do not require message privacy.

\end{itemize}

\subsection{The Notion of a Byzantine Agreement
}\label{sec:BANotion}

The notion of Byzantine agreement was introduced  by Pease Shostak and Lamport \cite{PSL} for the {\em binary} case, that is,  when every initial value consists of a bit.
However, it was quickly extended to arbitrary initial values. (See  the surveys of
Fischer \cite{Fi83} and
Chor and Dwork \cite{ChDw}.) By a BA protocol, we mean an arbitrary-value one.

\begin{definition}
In a synchronous network, let $\cP$ be a $n$-player protocol, whose player set is common knowledge among the players, $t$ a positive integer such that $n\geq 2t+1$.
We say that $\cP$ is an  {\em arbitrary-value} (respectively, {\em binary}) {\em $(n,t)$-Byzantine agreement protocol}
 with {\em soundness} $\sigma \in (0,1)$
if, for every set of values $V$ not containing the special symbol $\bot$ (respectively, for $V=\{0,1\}$), in an execution in which at most $t$ of the players are malicious and in which every player $i$ starts with an {\em initial value} $v_i\in V$, every honest player $j$ halts with probability 1, outputting a value $out_i \in V\cup \{\bot\}$ so as to satisfy, with  probability at least $\sigma$, the following two conditions:

\begin{itemize}

\item[1.] {\em Agreement:} There exists $out\in V\cup \{\bot\}$ such that $out_i=out$ for all honest players~$i$.

\item[2.] {\em Consistency:} if, for some value $v\in V$, $v_i=v$ for all honest players, then $out=v$.

\end{itemize}
We refer to $out$ as {\em $\cP$'s output}, and to each $out_i$ as {\em player $i$'s output}.

\end{definition}

\subsection{The BA Notation $\bm{\#}$} \label{sec:NotationSharp}

In our BA protocols, a player is required to count
how many players sent him a given message in a given step. Accordingly, for each possible value $v$ that might be sent,
$$\#^s_i(v)$$
(or just $\#_i(v)$ when $s$ is clear) is  the number of players $j$ from which $i$ has received $v$  in step $s$.

Recalling that a player $i$ receives exactly one message from each player $j$, if the number of players is $n$, then, for all $i$ and $s$, $\sum_v \#_i^s(v)=n.$

\subsection{The Binary BA Protocol $\bm{\BBA*}$ }\label{sec:BBA*}

In this section we present a new {\em binary} BA protocol, $\BBA*$, which relies on the honesty of  more than two thirds of the players and is very fast: no matter what the malicious players might do, each execution  of its main loop brings the players into agreement with probability 1/3.

Each player has his own public key of a digital signature scheme satisfying the unique-signature property. Since this protocol is intended to be run on
synchronous complete network, there is no need for a
player $i$ to sign each of his messages.

Digital signatures are used to generate a sufficiently common random bit in Step 3. (In Algorand, digital signatures are used to authenticate all other messages as well.)

The protocol requires a minimal set-up: a common random string $r$, independent of the players' keys.
(In Algorand, $r$
      is actually replaced by the quantity $Q\upr$.)

Protocol $BBA^\star$ is a 3-step loop, where the
players repeatedly exchange Boolean values,
and different players may exit this loop at different times.
A player $i$ exits this loop by propagating, at some step, either a special value $0*$ or a special value $1*$, thereby instructing all players to ``pretend'' they respectively receive 0 and 1
from $i$ in all future steps. (Alternatively said: assume that the last message received by a player $j$ from another player $i$ was a bit $b$. Then, in any step in which he does not  receive any message from $i$, $j$ acts as if $i$ sent him the bit $b$.)

The protocol  uses a counter $\gamma$, representing how many times its 3-step loop has been executed. At the start of $BBA^\star$, $\gamma=0$.  (One may think of $\gamma$ as a global counter, but it is actually increased by each individual player every time that the loop is executed.)

There are $n\geq 3t+1$, where $t$ is the maximum possible number of malicious players.
A binary string $x$ is identified with the integer whose binary representation (with possible leadings 0s) is $x$; and ${\tt lsb}(x)$ denotes the least significant bit of $x$.

\begin{center}
  {\sc Protocol $BBA^\star$}
\end{center}

\begin{description}

\item[{\sc (Communication) Step 1.}] [Coin-Fixed-To-0 Step]
{\em Each player $i$ sends $b_i$.}

{\em

\begin{itemize}
\item[1.1] If  $\#^1_i(0)\geq 2t+1$, then $i$ sets $b_i = 0$,\,
sends $0*$,\,
outputs $out_i=0$, \, and HALTS.

\item[1.2] If $\#^1_i(1)\geq 2t+1$, then, then $i$ sets $b_i = 1$.

\item[1.3] Else, $i$ sets $b_i = 0$.

\end{itemize}

}

\item[{\sc (Communication) Step 2.}] [Coin-Fixed-To-1 Step]
{\em Each player $i$ sends $b_i$.}

{\em

\begin{itemize}

\item[2.1] If $\#^2_i(1)\geq 2t+1$, then $i$ sets $b_i = 1$, \, sends $1*$, \,
outputs $out_i=1$, \, and HALTS.

\item[2.2] If $\#^2_i(0)\geq 2t+1$, then $i$ set $b_i = 0$.

\item[2.3] Else, $i$ sets $b_i = 1$.

\end{itemize}
}

\item[{\sc (Communication) Step 3.}] [Coin-Genuinely-Flipped Step]
{\em Each player $i$ sends $b_i$ and $SIG_i(r,\gamma)$.}

{\em

\begin{itemize}
  \item[3.1]If  $\#^3_i(0) \geq 2t+1$, then $i$ sets $b_i = 0$.

  \item[3.2] If  $\#^3_i(1) \geq 2t+1$, then $i$ sets $b_i = 1$.

  \item[3.3] Else, letting $S_i= \{j\in N \text{  who have sent $i$ a proper message in this step 3 } \}$, \\
$i$ sets $b_i = c \triangleq \lsb(\min_{j \in S_i} H(SIG_i(r,\gamma)))$;\, increases $\gamma_i$ by 1;\, and  returns  to Step 1.

\end{itemize}

}

\end{description}

\begin{theorem}\label{thm:BBA*}
Whenever $n\geq 3t+1$, $BBA^\star$ is a binary $(n,t)$-BA protocol with soundness 1.
\end{theorem}

\noindent

A proof of Theorem \ref{thm:BBA*} is given in \cite{BBA-Proof}. Its adaptation to
our setting, and its player-replaceability property are novel.

\paragraph{Historical Remark}

Probabilistic binary BA protocols were first proposed by Ben-Or in asynchronous settings \cite{Ben-Or}.
Protocol  $BBA^\star$ is a  novel adaptation, to our public-key setting, of the binary BA protocol of Feldman and Micali \cite{FM}. Their protocol was the first to work in an expected constant number of steps. It worked by having the players themselves implement a {\em common coin}, a notion proposed by Rabin, who implemented it via an external trusted party \cite{Rabin}.

\subsection{Graded Consensus and the Protocol $\bm{GC$}}

Let us recall, for arbitrary values,  a notion of consensus much weaker than Byzantine agreement.

\begin{definition}
  Let $\cP$ be a protocol in which the set of all players is common knowledge, and each player $i$ privately knows an arbitrary {\em initial  value} $v_i'$.

We say that $\cP$ is an {\em $(n,t)$-graded consensus protocol}
if, in every execution with $n$ players, at most $t$ of which are malicious, every honest player $i$ halts outputting a {\em value-grade} pair  $(v_i,g_i)$, where $g_i\in  \{0,1,2\}$, so as to satisfy the following three conditions:

\begin{itemize}

\item[1.] For all honest players $i$ and $j$, $|g_i-g_j|\leq 1$.

\item[2.] For all honest players $i$ and $j$,
$g_i,g_j>0 \Rightarrow v_i=v_j$.

\item[3.] If  $v_1'=\cdots =v_n'=v$ for some value $v$, then $v_i=v$ and $g_i=2$ for all honest players $i$.

\end{itemize}

\end{definition}

\paragraph{Historical Note}

The notion of a graded consensus is simply derived from that of a {\em graded broadcast}, put forward by
Feldman and Micali
in \cite{FM}, by strengthening the notion of a {\em crusader agreement}, as introduced by Dolev \cite{D}, and refined by Turpin and Coan \cite{TC}.%
\footnote{In essence, in a graded-broadcasting protocol, (a) the input of every player is the identity of a distinguished player, the {\em sender}, who has  an arbitrary value $v$ as an additional private input, and (b) the outputs must satisfy the same properties 1 and 2 of graded consensus, plus the following property $3'$: {\em if the sender is honest, then $v_i=v$ and $g_i=2$ for all honest player $i$.} }

 In \cite{FM}, the authors also provided a 3-step $(n,t)$-graded broadcasting protocol, {\em gradecast}, for $n\geq 3t+1$. A more complex $(n,t)$-graded-broadcasting protocol for $n>2t+1$
has later been found by Katz and Koo \cite{KK}.

The following two-step protocol $GC$
consists of the last two steps of {\em gradecast},  expressed in our notation. To emphasize this fact, and to match the steps of protocol $\alg$ of section \ref{sec:BasicAlgorand}, we respectively name 2 and 3 the steps of $GC$.

\begin{center}
  {\sc Protocol $GC$}
  \end{center}

\begin{description}

  \item[{\sc Step 2.}] {\em Each player $i$ sends $v_i'$ to all players.}

   \item[{\sc Step 3.}] {\em Each player $i$ sends to all players the string $x$ if and only if $\#_i^2(x)\geq 2t+1$.}

   \item[{\sc Output Determination.}]
       {\em  Each player $i$ outputs the pair $(v_i,g_i)$ computed as follows:}

       \begin{itemize}
         \item {\em If, for some $x$,  $\#_i^3(x)\geq 2t+1$, then $v_i=x$ and $g_i=2$.}

         \item {\em If, for some $x$, $\#_i^3(x)\geq t+1$, then $v_i=x$ and $g_i=1$.}

         \item {\em Else, $v_i=\bot$ and $g_i=0$.}

       \end{itemize}

\end{description}

\begin{theorem}\label{thm:GB}
 If $n\geq3t+1$, then $GC$ is a $(n,t)$-graded broadcast protocol.
\end{theorem}

\smallskip

The proof immediately follows from that of the protocol {\em gradecast} in \cite{FM}, and is thus omitted.%
\footnote{Indeed, in their protocol, in step 1, the sender sends his own private value $v$ to all players, and each player $i$ lets $v_i'$ consist of the value he has actually received from the sender in step 1.}

\subsection{The Protocol $\bm{\BA*}$}

We now describe the arbitrary-value BA protocol
$\BA*$ via the binary BA protocol $\BBA*$ and the graded-consensus protocol $GC$.
Below, the initial value of each player $i$ is $v_i'$.

\begin{center}
  {\sc Protocol $\BA*$}
  \end{center}

\begin{description}

  \item[{\sc Steps 1 and 2.}] {\em Each player $i$ executes $GC$, on input $v_i'$, so as to compute a pair $(v_i,g_i)$.}

   \item[{\sc Step 3, $\ldots$}] {\em Each player $i$ executes $\BBA*$ ---with initial input 0, if $g_i=2$, and 1 otherwise--- so as to compute the bit $out_i$.  }

   \item[{\sc Output Determination.}]
       {\em  Each player $i$ outputs $v_i$, if $out_i=0$, and $\bot$ otherwise.}

\end{description}

\begin{theorem}\label{thm:BA*}
Whenever $n\geq 3t+1$, $BA^\star$ is a $(n,t)$-BA protocol with soundness 1.
\end{theorem}

\medskip
\noindent
{\em Proof.} We first prove Consistency, and then Agreement.

\scparagraph{Proof of Consistency.}
Assume that, for some value $v\in V$, $v_i'=v$.
Then, by property 3 of graded consensus, after the execution of $GC$, all honest players output $(v,2)$.
Accordingly, 0 is the initial bit of all honest players in the end of the execution of $\BBA*$. Thus, by the
Agreement property of binary Byzantine agreement, at the end of the execution of $\BA*$,
$out_i=0$ for all honest players. This implies that the output of each honest player $i$ in $\BA*$ is $v_i=v$. \wqed
\scparagraph{Proof of Agreement.}
Since $\BBA*$ is a binary BA protocol, either

(A) $out_i=1$ for all honest player $i$, or

(B)
$out_i=0$ for all honest player $i$.

\noindent
In  case A, all honest players output $\bot$ in $\BA*$, and thus Agreement holds.
Consider now case B. In this case, in the execution of $\BBA*$, the initial bit of at least one honest player $i$ is 0. (Indeed, if initial bit of all honest players were 1, then, by the Consistency property of $\BBA*$, we would have $out_j=1$ for all honest $j$.)
Accordingly, after the execution of $GC$, $i$ outputs the pair $(v,2)$ for some value $v$.
Thus, by property 1 of graded consensus, $g_j>0$ for all honest players $j$. Accordingly, by property 2 of graded consensus, $v_j=v$ for all honest players $j$.
This implies that, at the end of $\BA*$, every honest player $j$ outputs $v$. Thus, Agreement holds also in case B.
 \wqed
Since both Consistency and Agreement hold,
$\BA*$ is an arbitrary-value BA protocol. \bqed

\paragraph{Historical Note} Turpin and Coan were the first to show that, for $n\geq 3t+1$, any binary $(n,t)$-BA protocol can be converted to an arbitrary-value $(n,t)$-BA protocol.
The reduction arbitrary-value Byzantine agreement to
binary Byzantine agreement via graded consensus is more modular and cleaner, and simplifies the analysis of our Algorand protocol $\alg$.

\paragraph{Generalizing $\bm{\BA*}$ for use in Algorand}
Algorand works even when all communication is via gossiping.
However, although presented in a traditional and familiar communication network, so as to enable a better comparison with the prior art and an easier understanding, protocol $\BA*$ works also in gossiping networks. In fact, in our detailed embodiments  of Algorand, we shall present it directly for gossiping networks. We shall also point out that it satisfies the player replaceability property that is crucial for Algorand to be secure in the envisaged very adversarial model.

Any BA player-replaceable protocol working in a gossiping communication network can be securely employed within the inventive Algorand system. In particular, Micali and Vaikunthanatan have extended $\BA*$ to work very efficiently also with a simple majority of honest players. That protocol too could be used in Algorand.

\section{Two Embodiments of Algorand}

As discussed, at a very high level,  a round of Algorand ideally proceeds as follows. First,
a randomly selected user, the leader, proposes and circulates a new block.
(This process includes initially selecting a few potential leaders and then ensuring that, at least a good fraction of the time, a single common leader emerges.)
Second, a randomly selected committee of users is selected, and reaches Byzantine agreement on the block proposed by the leader. (This process includes that each step of the BA protocol is run by a separately selected committee.)
The agreed upon block is then digitally signed by a given threshold ($T_H$) of committee members. These digital signatures are circulated so that everyone is assured of which is the new block. (This includes circulating the credential of the signers, and authenticating just the hash of the new block, ensuring that everyone is guaranteed to learn the block, once its hash is made clear.)

In the next two sections, we present two embodiments of  Algorand,
${\alg_1}$ and ${\alg_2}$, that work under a majority-of-honest-users assumption. In Section \ref{sec:hmm} we show how to adopts these embodiments to work under a honest-majority-of-money assumption.

$\alg_1$ only envisages that $>2/3$ of the committee members are honest. In addition, in $\alg_1$, the number of steps for reaching Byzantine agreement is capped at a suitably high number, so that agreement is guaranteed to be reached with overwhelming probability within a fixed number of steps (but potentially requiring longer time than the steps of $\alg_2$). In the remote case in which agreement is not yet reached by the last step, the committee agrees on the empty block, which is always valid.

$\alg_2$ envisages that the number of honest members in a committee is always greater than or equal to a fixed threshold $t_H$ (which guarantees that, with overwhelming probability, at least 2/3 of the committee members are honest). In addition, $\alg_2$ allows Byzantine agreement to be reached in an arbitrary number of steps (but potentially in a shorter time than $\alg_1$).

It is easy to derive many variants of these basic embodiments. In particular, it is easy, given $\alg_2$, to modify $\alg_1$ so as to enable to reach Byzantine agreement in an arbitrary number of steps.

Both embodiments share the following common core, notations, notions, and parameters.

\subsection{A Common Core}
\label{sec:BasicAlgorand}

\paragraph{Objectives}

Ideally, for each round $r$, Algorand would satisfy the following properties:

\begin{enumerate}

\item {\em Perfect Correctness.}  All honest users agree on the same block $B\upr$.

\item {\em Completeness 1.} With probability 1,  the payset of $B\upr$, $PAY\upr$, is maximal.%
     \footnote{Because paysets are defined to  contain valid payments, and honest users to make only valid payments, a maximal $PAY\upr$ contains the ``currently  outstanding" payments of all honest users.}
\end{enumerate}

Of course, guaranteeing perfect correctness alone is trivial: everyone always chooses the official payset $PAY\upr$ to be empty. But in this case, the system would have completeness 0. Unfortunately, guaranteeing both perfect correctness and completeness 1 is not easy in the presence of malicious users. Algorand thus adopts a more realistic objective. Informally, letting $h$ denote the percentage of users who are honest, $h>2/3$, the goal of  Algorand is
\begin{center}
{\em Guaranteeing, with overwhelming probability, perfect correctness and  completeness close to $h$.}
\end{center}
Privileging  correctness over  completeness seems a reasonable choice: payments not processed in one round can  be processed in the next, but one should avoid {\em forks}, if possible.

\paragraph{Led Byzantine Agreement}

 Perfect Correctness could be guaranteed as follows. At the start of round $r$, each user $i$ constructs his own candidate block $B\upr_i$, and then all users reach Byzantine agreement on one candidate block. As per  our introduction, the BA protocol employed requires a 2/3 honest majority and is player replaceable. Each of its step can be executed by a small and randomly selected set of {\em verifiers}, who do not share any inner variables.

 Unfortunately, this approach has no completeness guarantees. This is so, because the candidate blocks of the honest users are most likely totally different from each other. Thus, the ultimately agreed upon block might always be one with a non-maximal payset.
 In fact, it may always be the empty block, $B_\varepsilon$, that is, the block whose payset is empty.  well be the default, empty one.

$\alg$ avoids this completeness problem as follows. First,  a leader for round $r$, $\ell\upr$, is selected. Then,  $\ell\upr$  propagates his own candidate block, $B\upr_{\ell\upr}$. Finally, the users reach agreement on the block they actually receive from $\ell\upr$. Because, whenever $\ell\upr$ is honest,
Perfect Correctness and Completeness 1 both hold,
$\alg$ ensures that $\ell\upr$ is honest with probability close to $h$. (When the leader is malicious, we do not care whether the agreed upon block is one with an empty payset. After all, a malicious leader $\ell\upr$ might always maliciously choose $B\upr_{\ell\upr}$ to be the empty block, and then honestly propagate it, thus forcing the honest users to agree on the empty block.)

\paragraph{Leader Selection}

In Algorand's, the $r$th block
is of the form $B\upr=(r,PAY\upr, Q\upr, H(B\uprm )$.
As already mentioned in the introduction, the quantity $Q\uprm$ is carefully constructed so as to be essentially non-manipulatable by our very powerful Adversary. (Later on in this section, we shall  provide some intuition about why this is the case.)
At the start  of a round $r$, all users  know the blockchain so far, $B^0,\ldots,B\uprm$, from which they  deduce the set of users  of every prior round: that is, $PK^1,\ldots,PK\uprm$.
A {\em potential leader} of round $r$ is a user $i$ such that $$.H\left(SIG_i\left(r,1,Q\uprm\right)\right)\leq p\enspace.$$
Let us explain.
Note that, since the quantity $Q\uprm$ is part of block $B\uprm$,  and the underlying signature scheme satisfies the uniqueness property, $SIG_i\left(r,1,Q\uprm\right)$ is a binary string uniquely associated to $i$ and $r$.
Thus, since $H$ is a random oracle,
$ H\left(SIG_i\left(r,1,Q\uprm\right)\right)$ is a random 256-bit long string uniquely associated to $i$ and $r$. The symbol ``." in front of $H\left(SIG_i\left(r,1,Q\uprm\right)\right)$ is the {\em decimal} (in our case, {\em binary})  {\em point}, so that $r_i\triangleq .H\left(SIG_i\left(r,1,Q\uprm\right)\right)$ is the binary expansion of a random 256-bit number between 0 and 1 uniquely associated to $i$ and $r$. Thus the probability that $r_i$ is less than or equal to $p$ is essentially~$p$.
(Our potential-leader selection mechanism has been inspired by the micro-payment scheme of Micali and Rivest \cite{MR}.)

The probability $p$ is chosen so that, with overwhelming (i.e., $1-F$) probability, at least one potential verifier is honest.
(If fact, $p$ is chosen to be the smallest such probability.)

Note that, since $i$ is the only one capable of
computing his own signatures, he alone can determine whether he is a potential verifier of round 1. However, by revealing his own {\em credential},
$\sigma_i\upr \triangleq SIG_i\left(r,1,Q\uprm\right)$, $i$ can prove to anyone  to be a potential verifier of round~$r$.

The leader $\ell\upr$ is defined to be the potential leader whose hashed credential is smaller that the hashed credential of all other potential leader $j$: that is, $H(\sigma_{\ell\upr}\rs ) \leq H(\sigma_j\rs)$.

Note that, since a malicious $\ell\upr$ may not reveal his credential, the correct leader of round $r$ may never be known, and that, barring improbable ties, $\ell\upr$ is indeed the only leader of round $r$.

Let us finally bring up  a last but important detail: a user $i$ can be a potential leader  (and thus the leader) of a round $r$ only if he belonged to the system for at least $k$  rounds.
This  guarantees the
non-manipulatability of $Q\upr$ and all  future $Q$-quantities. In fact, one of the potential leaders will actually determine $Q\upr$.

\paragraph{Verifier Selection}

Each step $s>1$ of round $r$ is executed by a small set of verifiers, $SV\rs$. Again, each verifier $i \in SV\rs$  is randomly selected among the users already in the system $k$ rounds before $r$, and again via the special quantity $Q\uprm$.
Specifically,
$i\in PK^{r-k}$ is
a {\em verifier} in $SV\rs$, if $$.H\left(SIG_i\left(r,s,Q\uprm\right)\right)\leq p'\enspace.$$
Once more, only $i$ knows whether he belongs to $SV\rs$,but, if this is the case, he could prove it by exhibiting his credential $\sigma_i\rs\triangleq H(SIG_i\left(r,s,Q\uprm\right))$.
A verifier $i\in SV\rs$ sends a  message, $m_i\rs$, in step $s$ of round $r$, and this message includes his credential $\sigma_i\rs$, so as to enable the verifiers f the nest step to recognize that $m_i\rs$ is a legitimate step-$s$ message.

The probability $p'$  is chosen so as to ensure that, in $SV\rs$, letting $\#good$ be the number of honest users and $\#bad$ the number of malicious users, with overwhelming probability the following two conditions hold.

For embodiment $\alg_1$:

(1) $\#good >2\cdot\#bad$ and

  (2) $\#good+4 \cdot  \#bad<2n$, where $n$ is the expected cardinality of $SV\rs$.

  \bigskip

  For embodiment $\alg_2$:

(1) $\#good >t_H$ and

  (2) $\#good+ 2 \#bad<2t_H$, where $t_H$ is a specified threshold.

\bigskip

\noindent
These conditions imply that,
 with sufficiently high probability, (a) in the last step of the BA protocol, there will be at least given number of honest players to digitally sign the new block $B\upr$, (b) only one block per round may have the necessary number of signatures, and (c) the used BA protocol has (at each step) the required 2/3 honest majority.

\paragraph{Clarifying Block Generation}

If the round-$r$ leader $\ell\upr$ is honest, then the corresponding block is of the form
$$B\upr =\left(r,PAY\upr, SIG_{\ell\upr}\left(Q\uprm\right), H\left(B\uprm \right)\right)\enspace,$$
where the payset $PAY\upr$ is maximal. (recall that all paysets are, by definition, collectively valid.)

Else (i.e., if $\ell\upr$ is malicious), $B\upr$ has one of the following two possible forms:
$$B\upr =\left(r,PAY\upr, SIG_{i}\left(Q\uprm\right), H\left(B\uprm \right)\right) \quad \text{ and } \quad B\upr =B_\varepsilon\upr \triangleq \left(r,\emptyset, Q\uprm, H\left(B\uprm \right)\right)\enspace.$$
In the first form,   $PAY\upr$ is a (non-necessarily maximal)  payset and it may be $PAY\upr=\emptyset$; and $i$ is a potential leader of round $r$. (However, $i$ may not be the leader $\ell\upr$. This may indeed happen if if $\ell\upr$ keeps secret his credential and does not reveal himself.)

The second form arises when, in the round-$r$ execution of the BA protocol, all honest players output the default value, which is the empty block $B_\varepsilon\upr$ in our application. (By definition, the possible outputs of a BA protocol include a default value, generically denoted by $\bot$. See  section~\ref{sec:BANotion}.)

Note that, although the paysets are empty in both cases, $B\upr =\left(r,\emptyset, SIG_{i}\left(Q\uprm\right), H\left(B\uprm \right)\right)$
and $B_\varepsilon\upr$  are syntactically different blocks and arise in two different situations: respectively, ``all went smoothly enough in the execution of the BA protocol", and  ``something went wrong in the BA protocol, and the default value was output".

Let us now intuitively describe how the generation of block $B\upr$ proceeds in round $r$ of $\alg$.
In the first step, each eligible player, that is, each player $i\in PK^{r-k}$, checks whether he is a potential leader. If this is the case, then $i$ is asked, using of all the payments he has seen so far, and the current blockchain, $B^0,\ldots,B\uprm$, to secretly prepare a maximal payment set, $PAY\upr_i$, and secretly assembles his candidate block,
 $B\upr =\left(r,PAY\upr_i, SIG_{i}\left(Q\uprm\right), H\left(B\uprm \right)\right)$. That,is, not only does he include in $B\upr_i$, as its second component the
just prepared payset, but also, as its third component, his own signature of $Q\uprm$, the third component of the last block, $B\uprm$. Finally, he
propagate his round-$r$-step-1 message, $m_i^{r,1}$,
which includes (a) his candidate block $B\upr_i$, (b)
his proper signature of his candidate block (i.e., his signature of the hash of $B\upr_i$, and (c) his own credential $\sigma_i^{r,1}$, proving that he is indeed a potential verifier of round $r$.

(Note that, until an honest $i$ produces his message $m_i^{r,1}$, the Adversary has no clue that $i$ is a potential verifier. Should he wish to corrupt honest potential leaders, the Adversary might as well corrupt random honest players. However, once he sees $m_i^{r,1}$, since it contains $i$'s credential, the Adversary knows and could corrupt $i$, but cannot prevent $m_i^{r,1}$, which is virally propagated, from reaching all users in the system.)

In the second step, each  selected verifier $j\in SV^{r,2}$ tries to identify the leader of the round. Specifically,
 $j$ takes the step-1 credentials, $\sigma_{i_1}^{r,1},\ldots, \sigma_{i_n}^{r,1}$,  contained in the proper step-1 message $m_i^{r,1}$ he has received; hashes all of them, that is, computes
 $H\left(\sigma_{i_1}^{r,1}\right),\ldots, H\left(\sigma_{i_n}^{r,1}\right)$; finds the credential, $\sigma_{\ell_j}^{r,1}$, whose hash is lexicographically minimum; and considers $\ell_j\upr$ to be the leader of round $r$.

Recall that each considered credential is a digital signature of $Q\uprm$, that $SIG_i\left(r,1,Q\uprm\right)$ is uniquely determined by $i$ and $Q\uprm$,  that $H$ is random oracle, and thus that each $H(SIG_i\left(r,1,Q\uprm\right)$ is a random 256-bit long string unique to each potential leader $i$ of round $r$.

From this we can conclude that, if the 256-bit string $Q\uprm$ were itself {\em randomly and independently selected}, than so would be the hashed credentials of all potential leaders of round $r$.
In fact, all potential leaders are well defined, and so are their  credentials (whether actually computed or not). Further, the set of potential leaders of round $r$ is a random subset of the users of round $r-k$, and
an honest potential leader $i$ always properly constructs and propagates his message $m_i\upr$, which contains $i$'s credential. Thus, since the percentage of honest users is $h$,  no matter what the malicious potential leaders might do (e.g., reveal or conceal their own credentials), the minimum hashed potential-leader credential belongs to a honest user, who is necessarily identified by everyone to be the leader $\ell\upr$ of the round $r$.
Accordingly, if the 256-bit string $Q\uprm$ were itself {\em randomly and independently selected}, with probability exactly $h$ (a) the leader  $\ell\upr$ is honest  and (b) $\ell_j=\ell\upr$ for all honest step-2 verifiers $j$.

In reality, the hashed credential are, yes, randomly selected, but depend on $Q\uprm$, which is not randomly and independently selected. We shall prove in our analysis, however, that $Q\uprm$ is sufficiently non-manipulatable to guarantee that the leader of a round is honest with probability $h'$ sufficiently close to $h$: namely, $h'> h^2(1+h-h^2)$. For instance, if $h=80\%$, then $h'> .7424$.

Having identified the leader of the round (which they correctly do when the leader $\ell\upr$ is honest),
the task of  the step-2 verifiers is to start executing the BA using as initial values what they believe to be the block of the leader.
Actually, in order to minimize the amount of communication required,
a verifier $j\in SV^{r,2}$ does not use, as his input value $v_j'$ to the Byzantine protocol, the block  $B_j$ that he has actually received from $\ell_j$ (the user  $j$ believes to be the leader), but the the leader, but the hash of that block, that is, $v_j'=H(B_i)$.
Thus, upon termination of the BA protocol, the verifiers of the last step do not compute the desired round-$r$ block $B\upr$, but compute (authenticate and propagate) $H(B\upr)$.
Accordingly, since $H(B\upr)$ is digitally signed by sufficiently many verifiers of the last step of the BA protocol, the users in the system will realize that $H(B\upr)$ is the hash of the new block. However, they must also retrieve (or wait for, since the execution is quite asynchronous) the block $B\upr$ itself, which the protocol ensures that is indeed available, no matter what the Adversary might do.

\paragraph{Asynchrony and Timing}

 $\alg_1$ and $\alg_2$ have a significant degree of asynchrony. This is so because the Adversary has large latitude in scheduling the delivery of the messages being propagated. In addition, whether the total number of steps in a round is capped or not, there is the variance contribute by the number of steps actually taken.

As soon as he learns  the certificates of $B^0,\ldots, B\uprm$, a user $i$ computes $Q\uprm$  and starts working on round $r$, checking whether he is a potential leader, or a verifier in some step $s$ of round $r$.

Assuming that $i$ must act at step $s$,  in light of the discussed asynchrony, $i$ relies on various strategies to ensure that he has sufficient information before he acts.

For instance, he might wait to receive at least a given number of messages from the verifiers  of the previous step, or wait for a sufficient time to ensure that he receives the messages of sufficiently many verifiers of the previous step.

\paragraph{The  Seed ${\mathbf Q}^{\mathbf r}$ and the Look-Back Parameter $\bm{k}$}

Recall that, ideally, the quantities $Q\upr$ should random and independent, although it will suffice for them to be sufficiently non-manipulatable by the Adversary.

At a first glance,
we could choose $Q\uprm$ to coincide with $H\left(PAY\uprm\right)$, and thus avoid to specify $Q\uprm$ explicitly in $B\uprm$.
An elementary analysis reveals, however, that  malicious users may take advantage of this selection mechanism.%
\footnote{We are at the start of round $r-1$. Thus, $Q^{r-2}=PAY^{r-2}$ is publicly known, and  the Adversary privately knows who are the potential leaders he controls. Assume that the Adversary controls 10\% of the users, and that, with very high probability, a malicious user $w$ is the potential leader of round $r-1$. That is, assume that $H\left(SIG_w\left(r-2,1,Q^{r-2}\right)\right)$ is so small that it is highly improbable an honest potential leader will actually be the leader of round $r-1$. (Recall that, since we choose potential leaders via a secret cryptographic sortition mechanism, the Adversary does not know who the honest potential leaders are.) The Adversary, therefore, is in the enviable position of choosing the payset $PAY'$ he wants, and have it become the official payset of round $r-1$. However, he can do more. He can also ensure that, with high probability, (*) one of his malicious users will be the leader also of round $r$, so that he can freely select what $PAY\upr$ will be. (And so on. At least for a long while, that is, as long as these high-probability events really occur.)   To guarantee (*), the Adversary acts as follows.
Let $PAY'$ be the payset the Adversary prefers for round $r-1$.
Then, he computes $H(PAY')$ and checks whether, for some already malicious player $z$, $SIG_z(r,1,H(PAY'))$ is particularly small, that is, small enough that with very high probability $z$ will be the leader of round $r$. If this is the case, then he instructs $w$ to choose his candidate block to be $B\uprm_i=(r-1,PAY',H(B^{r-2})$.
Else, he has two other malicious users $x$ and $y$ to keep on generating a new payment $\wp'$, from one to the other, until, for some malicious user $z$ (or even for some fixed user $z$)
$H\left(SIG_z\left(PAY'\cup \{\wp\}\right)\right)$
is particularly small too.
This experiment will stop quite quickly. And when it does the Adversary asks $w$ to propose the candidate block $B\uprm_i=(r-1,PAY'\cup \{\wp\},H(B^{r-2})$.}\label{BadLeaderSelection}
Some additional effort shows that myriads of other alternatives, based on traditional block quantities are easily exploitable by the Adversary to ensure that malicious leaders are very frequent.
We instead specifically and inductively define our brand new quantity $Q\upr$ so as to be able to prove that it is non-manipulatable by the Adversary.
Namely,
$$Q^r \triangleq H(SIG_{\ell^r}(Q^{r-1}), r), \text{ if $B\upr$ is not the empty block, and }Q^r \triangleq H(Q^{r-1}, r) \text{ otherwise.}$$

The intuition of why this construction of $Q\upr$ works is as follows. Assume for a moment that $Q\uprm$ is truly randomly and independently selected. Then, will so be $Q\upr$? When $\ell\upr$ is honest the answer is (roughly speaking) yes. This is so because
$$ H(SIG_{\ell^r}(\cdot), r): \bit^{256} \longrightarrow \bit^{256}$$
is a random function.
When $\ell\upr$ is malicious, however, $Q\upr$ is no longer univocally defined from $Q\uprm$ and $\ell\upr$.
There are at least two separate values for $Q\upr$.
One continues to be  $Q^r \triangleq H(SIG_{\ell^r}(Q^{r-1}), r)$, and the other is
$H(Q^{r-1}, r)$. Let us first argue that, while the second choice is somewhat arbitrary, a second choice is absolutely mandatory.
The reason for this is that a malicious $\ell\upr$ can always cause totally different candidate blocks to be received by the honest verifiers of the second step.%
\footnote{For instance, to keep it simple (but extreme), ``when the time of the second step is about to expire", $\ell\upr$ could directly email a different candidate block $B_i$ to each user $i$. This way, whoever the step-2 verifiers might be, they will have received totally different blocks.}
Once this is the case, it is easy to ensure that
the block ultimately agreed upon via the BA protocol
of round $r$ will be the default one, and thus will not contain
anyone's digital signature of $Q\uprm$. But the system must continue, and for this, it needs a  leader for round $r$. If this leader is automatically and openly selected, then the Adversary will trivially corrupt him. If it is selected by the previous $Q\uprm$ via the same process, than $\ell\upr$ will again be the leader in round $r+1$. We specifically propose to use the same secret cryptographic sortition mechanism, but applied
to a new $Q$-quantity: namely, $H(Q^{r-1}, r)$.
By having this quantity to be the output of $H$ guarantees that the output is random, and by including $r$ as the second input of $H$, while all other uses of $H$ have one or $3+$ inputs, ``guarantees" that such a $Q\upr$ is independently selected. Again, our specific choice of alternative $Q\upr$ does not matter, what matter is that $\ell\upr$ has two choice for $Q\upr$, and thus he can double his chances to have another malicious user as the next leader.

The options for $Q\upr$ may even be more numerous for
the Adversary who controls a malicious $\ell\upr$. For instance, let $x$, $y$, and $z$ be three malicious potential leaders of round $r$ such that
$$H\left(\sigma_x^{r,1}\right)< H\left(\sigma_y^{r,1}\right)<
H\left(\sigma_z^{r,1}\right)$$
and $H\left(\sigma_z^{r,1}\right)$ is particulary small. That is, so small that there is a good chance
that $H\left(\sigma_z^{r,1}\right)$ is smaller of the hashed credential of every honest potential leader.
Then, by asking $x$ to hide his credential, the Adversary has a good chance of having $y$ become the leader of round $r-1$. This implies that
he has another option for $Q\upr$: namely,
$SIG_y\left( Q\uprm \right)$.
Similarly, the Adversary may ask both $x$ and $y$ of
withholding their credentials, so as to have $z$ become the leader of round $r-1$ and gaining another option for $Q\upr$: namely, $SIG_z\left( Q\uprm \right)$.

Of course, however, each of these and other options
has a non-zero chance to fail, because the Adversary cannot predict the hash of the digital signatures of the honest potential users.

A careful, Markov-chain-like analysis  shows that, no matter what options the Adversary chooses to make at round $r-1$, {\em as long as he cannot inject new users in the system}, he cannot decrease the probability of an honest user to be the leader of round $r+40$ much below $h$.
This is the reason for which we demand that the potential leaders of round $r$ are users already existing in round $r-k$. It is a way to ensure that, at round $r-k$, the Adversary cannot alter by much the
probability that an honest user become the leader of round $r$. In fact, no matter what users he may
add to the system in rounds $r-k$ through $r$, they are ineligible to become potential leaders (and {\em a fortiori} the leader) of round $r$.
Thus the look-back parameter $k$ ultimately is a security parameter. (Although, as we shall see in section \ref{sec:LazyHonesty},
it can also be a kind of ``convenience parameter" as well.)

\paragraph{Ephemeral Keys}

Although the execution of our protocol cannot generate a fork, except with negligible probability,
the Adversary could generate a fork, at the $r$th block,
after the legitimate block $r$ has been generated.

Roughly, once $B\upr$ has been generated, the Adversary has learned who the verifiers of each step of round $r$ are. Thus, he could therefore corrupt all of them
and oblige them to certify a new block $\widetilde{B\upr}$. Since this fake block might be propagated only  after the legitimate one, users that have been paying attention would not be fooled.%
\footnote{Consider corrupting the news anchor of a major TV network, and producing and broadcasting today a newsreel showing secretary Clinton winning the last presidential election. Most of us would recognize it as a hoax. But someone getting out of a coma might be fooled.}
Nonetheless, $\widetilde{B\upr}$ would be syntactically correct and we want to prevent from being manufactured.

We do so by means of a new rule. Essentially, the members of  the verifier set $SV\rs$ of a step $s$ of round $r$  use  ephemeral public keys $pk_i\rs$ to digitally sign their messages. These keys are single-use-only and their corresponding secret keys $sk^{r, s}_i$ are destroyed once used. This way, if a verifier is corrupted later on, the Adversary cannot force him to
sign anything else he did not originally sign.

 Naturally,  we must ensure that it is impossible for the Adversary to compute a new key $\widetilde{p_i\rs}$ and convince an honest user that it is the right ephemeral key of  verifier $i\in SV\rs$ to use in step $s$.

\subsection{Common Summary of Notations, Notions, and Parameters}\label{sec:para}

\paragraph{Notations}

\begin{newitemize}

\item
$r\geq 0$: the current round number.

\item
$s\geq 1$: the current step number in round $r$.

\item
$B^{r}$: the block generated in round $r$.

\item
$PK^r$: the set of public keys by the end of round $r-1$ and at the beginning of round $r$.

\item
$S^r$: the system status by the end of round $r-1$ and at the beginning of round $r$.%
\footnote{In a system that is not synchronous, the notion of ``the end of round $r-1$'' and ``the beginning
of round $r$'' need to be carefully defined. Mathematically, $PK^r$ and $S^r$ are computed from
the initial status $S^0$ and the blocks $B^1, \dots, B^{r-1}$.}

\item
$PAY^r$: the payset contained in $B^r$.

\item
$\ell^r$: round-$r$ leader. $\ell^r$ chooses the payset $PAY^r$ of round $r$ (and determines the next $Q^r$).

\item
$Q^r$: the seed of round $r$, a quantity (i.e.,  binary string)
that is generated at the end of round $r$ and is used to choose verifiers
for round $r+1$.  $Q^r$ is  independent  of the paysets in the blocks
 and cannot be manipulated by $\ell^r$.

\item
$SV^{r, s}$: the set of verifiers chosen for step $s$ of round $r$.

\item
$SV^r$: the set of verifiers chosen for round $r$, $SV^r = \cup_{s\geq 1} SV^{r, s}$.

\item
$MSV^{r, s}$ and $HSV^{r, s}$: respectively, the
set of malicious verifiers and the set of honest verifiers in $SV^{r, s}$.
$MSV^{r, s}\cup HSV^{r, s} = SV^{r, s}$ and $MSV^{r, s}\cap HSV^{r, s} = \emptyset$.

\item
$n_1\in \bZ^+$ and $n\in \bZ^+$: respectively,  the expected numbers of potential leaders in each $SV^{r, 1}$, and the expected numbers of  verifiers in each
 $SV^{r, s}$, for $s>1$.

 Notice that $n_1<<n$, since we need at least one honest honest member in $SV^{r,1}$, but at least a majority of honest members in each $SV^{r, s}$ for $s>1$.

\item
$h\in (0, 1)$:
a constant greater than 2/3.
$h$ is
the honesty ratio in the system. That is, the fraction of honest users or honest money, depending on the assumption used,  in each $PK^{r}$ is at least $h$.

\item
$H$: a cryptographic hash function, modelled as a random oracle.

\item
$\bot$: A special string of the same length as the output of $H$.

\item
$F\in (0, 1)$: the parameter specifying the allowed error probability. A probability $\leq F$ is considered ``negligible'', and
a probability $\geq 1-F$ is considered ``overwhelming''.

\item
$p_h\in (0, 1)$: the probability that the leader of a round $r$, $\ell^r$, is honest. Ideally $p_h = h$. With the existence of the Adversary,
the value of $p_h$ will be determined in the analysis.

\item
$k\in \bZ^+$: the look-back parameter. That is, round $r-k$ is where the verifiers
for round $r$ are chosen from ---namely, $SV^r\subseteq PK^{r-k}$.%
\footnote{Strictly speaking, ``$r-k$'' should be ``$\max\{0, r-k\}$''.}

\item
$p_1\in (0, 1)$: for the first step of round $r$, a user in round $r-k$ is chosen to be
in $SV^{r, 1}$ with probability $p_1 \triangleq \frac{n_1}{|PK^{r-k}|}$.

\item
$p\in (0, 1)$: for each step $s>1$ of round $r$, a user in round $r-k$ is chosen to be
in $SV^{r, s}$ with probability $p \triangleq \frac{n}{|PK^{r-k}|}$.

\item
$CERT^r$: the certificate for $B^r$. It is a set of $t_H$
signatures of $H(B^r)$ from proper verifiers
in  round $r$.

\item
$\overline{B^r} \triangleq (B^r, CERT^r)$ is a proven block.

A user $i$ {\em knows} $B^r$ if he possesses (and successfully verifies) both parts of the proven block.
Note that the $CERT^r$ seen by different users may be different.

\item
$\tau_i^r$: the (local) time at which a user $i$ knows $B^r$.
In the Algorand protocol each user has his own clock. Different users' clocks need not be synchronized, but must have the same speed.
Only for the purpose of the analysis, we consider a reference clock and measure the players' related times with respect to it.

\item
$\alpha^{r, s}_i$ and $\beta^{r,s}_i$: respectively the (local) time a  user $i$ starts and ends his execution of Step $s$ of round $r$.

\item
$\Lambda$ and $\lambda$:
essentially, the upper-bounds to, respectively, the  time needed to execute Step 1
and the time needed for any other step of the Algorand protocol.

Parameter $\Lambda$ upper-bounds the time to propagate
a single 1MB block. (In our notation, $\Lambda = \lambda_{\rho, 1MB}$. Recalling our notation, that we set $\rho = 1$ for simplicity, and that blocks are chosen to be at most 1MB-long, we have  $\Lambda = \lambda_{1,1, 1MB}$.)

Parameter $\lambda$ upperbounds the time
to propagate one small message per verifier in a Step $s>1$.
(Using, as in Bitcoin,  elliptic curve signatures with 32B keys,
a verifier message is 200B long. Thus, in our notation, $\lambda = \lambda_{n,\rho, 200B}$.)

We assume that $\Lambda=O(\lambda)$.

\end{newitemize}

\paragraph{Notions}
\begin{itemize}
\item
Verifier selection.

For each round $r$ and step $s>1$,
$SV^{r, s}\triangleq \{i\in PK^{r-k}: \  .H(SIG_i(r, s, Q^{r-1}))\leq p\}$.
Each user $i\in PK^{r-k}$ privately computes his signature using his long-term key
and decides whether $i\in SV^{r, s}$ or not.
If $i\in SV^{r, s}$, then $SIG_i(r, s, Q^{r-1})$ is $i$'s $(r, s)$-credential, compactly denoted by $\sigma_i^{r, s}$.

For the first step of round $r$,
$SV^{r,1}$ and $\sigma_i^{r, 1}$ are similarly defined, with $p$ replaced by $p_1$.
The verifiers in $SV^{r, 1}$ are {\em potential leaders}.

\item
Leader selection.

User $i\in SV^{r, 1}$ is the leader of round $r$, denoted by $\ell^r$, if
$H(\sigma_i^{r, 1})\leq H(\sigma_j^{r,1})$ for all potential leaders $j\in SV^{r, 1}$.
Whenever the hashes of two players' credentials are compared,
in the unlikely event of ties, the protocol always breaks ties lexicographically according to the (long-term public keys of the) potential leaders.

By definition, the hash value of player $\ell^r$'s credential is also the smallest among all users in $PK^{r-k}$.
Note that a potential leader cannot privately decide whether he is the leader or not, without seeing the other potential leaders' credentials.

Since the hash values are uniform at random, when $SV^{r, 1}$ is non-empty,
$\ell^r$ always exists and is honest with probability at least $h$. The parameter $n_1$ is large enough so as to ensure that each $SV^{r, 1}$ is
non-empty with overwhelming probability.

\item
Block structure.

A non-empty block is of the form $B^r = (r, PAY^r, SIG_{\ell^r}(Q^{r-1}), H(B^{r-1}))$,
and an empty block is of the form $B^r_\epsilon = (r, \emptyset, Q^{r-1}, H(B^{r-1}))$.

Note that a non-empty block may still contain an empty payset $PAY^r$, if
no payment occurs in this round or if the leader is malicious.
However, a non-empty block implies
that the identity of $\ell^r$, his credential $\sigma^{r, 1}_{\ell^r}$
and $SIG_{\ell^r}(Q^{r-1})$ have all been
timely revealed.
The protocol guarantees that, if the leader is honest, then the block will be non-empty with overwhelming probability.

\item
Seed $Q^r$.

If $B^{r}$ is non-empty,
then
$Q^r \triangleq H(SIG_{\ell^r}(Q^{r-1}), r)$,
otherwise $Q^r \triangleq H(Q^{r-1}, r)$.

\end{itemize}

\paragraph{Parameters}

\begin{newitemize}

\item
{\em Relationships among various parameters.}

\begin{newitemize}

\item[---] The verifiers and potential leaders of round $r$ are selected from the users in $PK^{r-k}$, where $k$ is chosen so that
the Adversary
cannot predict $Q^{r-1}$ back at round $r-k-1$ with probability better than $F$:
otherwise, he will be able to introduce malicious users for round $r-k$, all of which
will be potential leaders/verifiers in round $r$, succeeding in having a malicious leader or a malicious
majority in
$SV^{r, s}$ for some steps $s$ desired by him.

\item[---] For Step 1 of each round $r$, $n_1$ is chosen so that with overwhelming probability, $SV^{r, 1}\neq \emptyset$.

\end{newitemize}

\item
{\em Example choices of important parameters.}
\begin{newitemize}

\item[---]
The outputs of $H$ are 256-bit long.

\item[---]
$h=80\%$, $n_1=35$.

\item[---] $\Lambda =$ 1 minute and $\lambda =$ 10 seconds.

  \end{newitemize}

  \item {\em Initialization of the protocol.}

  The protocol starts at time $0$ with $r=0$. Since there does not exist ``$B^{-1}$'' or ``$CERT^{-1}$'', syntactically
 $B^{-1}$ is a public parameter with its third component specifying $Q^{-1}$, and
all users know $B^{-1}$ at time 0.

\end{newitemize}

\section{$\bm{\alg_1}$ }\label{sec:Intuitionalg}

In this section, we construct a version of  $\alg$ working under the following assumption.
\begin{description}
  \item[{\sc Honest Majority of Users  Assumption:}] {\em More than 2/3 of the users in each $PK\upr$ are honest.}

\end{description}
In Section \ref{sec:hmm}, we show how to replace the above assumption with the desired Honest Majority of Money assumption.

\subsection{Additional Notations and Parameters}\label{sec:para}

\paragraph{Notations}

\begin{newitemize}

\item
$m\in \bZ^+$: the maximum number of steps in the binary BA protocol, a multiple of 3.

\item
$L^r\leq m/3$:  a random variable
representing the number of Bernoulli trials needed to see a 1,
when each trial is 1 with probability $\frac{p_h}{2}$ and
there are at most $m/3$ trials.
If all trials fail then $L^r\triangleq m/3$.
$L^r$ will be used to upper-bound the time needed to generate block $B^r$.

\item
$t_H=\frac{2n}{3}+1$: the number of signatures needed in the ending conditions of the protocol.

\item
$CERT^r$: the certificate for $B^r$. It is a set of $t_H$
signatures of $H(B^r)$ from proper verifiers
in  round $r$.

\end{newitemize}

\paragraph{Parameters}

\begin{newitemize}

\item
{\em Relationships among various parameters.}

\begin{newitemize}

\item[---] For each step $s>1$ of round $r$,
$n$
is chosen so that, with overwhelming probability,
\begin{newcenter}
$|HSV^{r, s}|>2|MSV^{r, s}|$ \quad and \quad $|HSV^{r,s}|+4|MSV^{r, s}|<2n$.
\end{newcenter}

The closer to 1 the value of $h$ is, the smaller $n$ needs to be.
In particular, we use (variants of) Chernoff bounds to ensure the desired conditions hold with overwhelming probability.

\item[---]
$m$ is chosen such that $L^r<m/3$ with overwhelming probability.

\end{newitemize}

\item
{\em Example choices of important parameters.}
\begin{newitemize}

\item[---]
$F=10^{-12}$.

\item[---]
$n\approx 1500$, $k=40$ and $m =180$.

\end{newitemize}

\end{newitemize}

\subsection{Implementing Ephemeral Keys in $\alg_1$}

As already mentioned, we wish that a verifier $i\in SV\rs$
digitally signs his message $m_i\rs$ of step $s$ in round $r$,
relative to an ephemeral public key $pk_i\rs$, using an ephemeral secrete key $sk^{r, s}_i$ that
 he promptly destroys after using.  We thus need an efficient method to ensure
that every user can verify that $pk_i\rs$ is indeed the key to use to verify
$i$'s signature of $m_i\rs$.
We do so by a (to the best of our knowledge) new use of identity-based signature schemes.

At a high level, in such a scheme, a central authority $A$ generates a
public master key, $PMK$, and a corresponding secret master key, $SMK$. Given the identity, $U$, of a player $U$, $A$ computes, via $SMK$, a secret signature key $sk_{U}$ relative to the public key $U$, and privately gives $sk_U$ to $U$. (Indeed, in an identity-based digital signature scheme, the public key of a user $U$ is $U$ itself!) This way, if $A$ destroys
$SMK$ after computing the secret keys of the users he wants to enable to produce digital signatures, and does not keep any computed secret key,
then $U$ is the only one who can digitally sign messages relative to the public key $U$. Thus, anyone who knows ``$U$'s name", automatically knows $U$'s public key, and thus can verify $U$'s signatures (possibly using also the public master key $PMK$).

In our application, the authority $A$ is user $i$, and the set of all possible users  $U$ coincides
with the round-step pair $(r,s)$ in ---say---
$S=\{i\}\times \{r',\ldots,r'+10^6\}\times \{1, \ldots, m+3\}$, where $r'$ is a given round, and $m+3$ the upperbound to the number of steps that may occur within a round.  This way, $pk_i\rs \triangleq (i, r,s)$, so that everyone seeing $i$'s signature $SIG_{pk^{r, s}_i}\rs(m_i\rs)$ can, with overwhelming probability, immediately verify it for the first million rounds $r$ following $r'$.

In other words, $i$ first generates $PMK$ and $SMK$.
Then, he publicizes that $PMK$ is $i$'s master public key for any round $r\in [r',r'+10^6]$, and uses $SMK$
 to privately produce and store the secret key $sk^{r,s}_i$ for each triple $(i, r,s)\in S$. This done, he  destroys $SMK$.
If he determines that he is not part of $SV\rs$, then $i$ may leave $sk^{r,s}_i$ alone (as the protocol does not require that he aunthenticates any message in Step $s$ of round $r$). Else, $i$ first uses  $sk^{r,s}_i$ to digitally sign his message $m_i\rs$, and then destroys $sk^{r,s}_i$.

Note that $i$ can publicize his first public master key  when he first enters the system. That is, the same payment $\wp$ that brings $i$ into the system (at a round $r'$ or at a round close to $r'$), may also specify, at $i$'s request, that $i$'s public master key for any round $r\in [r',r'+10^6]$ is $PMK$ ---e.g., by including a pair of the form $(PMK,[r',r'+10^6])$.

Also note that, since
$m+3$ is the maximum number of steps in a round,
assuming that a round takes a minute, the stash of ephemeral keys so produced will last $i$ for almost two years. At the same time, these ephemeral  secret keys will not take $i$ too long to produce. Using an elliptic-curve based system
with 32B keys, each secret key is computed in a few microseconds. Thus, if $m+3=180$, then all 180M secret keys can be computed in less than one hour.

When the current round is getting close to $r'+10^6$, to handle the next million rounds, $i$ generates a new $(PMK', SMK')$ pair, and informs what his next stash of ephemeral keys is by ---for example--- having  $SIG_i(PMK', [r'+10^6+1,r'+2\cdot 10^6+1])$ enter a new block, either as a separate ``transaction" or as some additional information that is part of a payment. By so doing, $i$ informs everyone that
he/she should use $PMK'$ to verify $i$'s ephemeral signatures in the next million rounds. And so on.

(Note that, following this basic approach, other ways for implementing ephemeral keys without using identity-based signatures are certainly possible. For instance, via Merkle trees.%
\footnote{In this method, $i$ generates a public-secret key pair $(pk_i\rs,sk_i\rs)$ for each round-step pair $(r,s)$ in ---say---
$ \{r',\ldots, r'+10^6\}\times \{1,\ldots, m+3\}$.
Then he orders these public keys in a canonical way, stores the $j$th public key in the $j$th leaf of a Merkle tree, and  computes the root value $R_i$, which he publicizes. When he wants to sign a message relative to key $pk_i\rs$, $i$ not only provides the actual signature, but also the authenticating path for $pk_i\rs$ relative to $R_i$. Notice that this authenticating path also proves that $pk_i\rs$ is stored in the $j$th leaf.   The rest of the details can be easily filled.}%
)

Other ways for implementing ephemeral keys are certainly possible ---e.g., via Merkle trees.

\subsection{Matching the Steps of $\bm{\alg_1}$ with those of $\bm{\BA*}$}

As we said, a round in $\alg_1$ has at most $m+3$ steps.

\begin{description}
  \item[{\sc Step 1.}]

  In this step, each potential leader $i$  computes and propagates his candidate block $B^r_i$, together with his own credential, $\sigma_i^{r,1}$.

  Recall that this credential
  {\em explicitly} identifies $i$. This is so, because
  $\sigma_i^{r,1}\triangleq SIG_i(r,1,Q\uprm)$.

  Potential verifier $i$ also propagates, as part of his message, his proper digital signature of $H(B^r_i)$. Not dealing with a payment or a credential, this signature of $i$ is relative to his ephemeral public key $pk_i^{r,1}$: that is, he propagates $sig_{pk_i^{r, 1}}(H(B^r_i))$.

  Given our conventions, rather than propagating $B^r_i$ and $sig_{pk_i^{r, 1}}(H(B^r_i))$, he could have
  propagated $SIG_{pk_i^{r, 1}}(H(B^r_i))$. However, in our analysis we need to have explicit access to $sig_{pk_i^{r, 1}}(H(B^r_i))$.

   \item[{\sc Steps 2.}]

   In this step, each verifier $i$ sets  $\ell_i\upr$ to be the potential leader whose hashed credential is the smallest, and  $B^r_i$ to be the block proposed by $\ell\upr_i$.
   Since, for the sake of efficiency,
   we wish to agree on $H(B\upr)$, rather than directly on $B\upr$, $i$ propagates the message
   he would have propagated in the first step of $\BA*$ with initial value $v_i'=H(B_i\upr)$.
   That is, he propagates $v_i'$, after ephemerally signing it, of course. (Namely, after signing it relative to the right ephemeral public key, which in this case is $pk_i^{r,2}$.)
      Of course too, $i$ also transmits his own credential.

   Since the first step of $\BA*$ consists of the
   first step of the graded consensus protocol $GC$, Step 2 of $\alg$ corresponds to the first step of $GC$.

  \item[{\sc Steps 3.}]

   In this step, each verifier $i\in SV^{r,2}$ executes the second step of $\BA*$. That is, he sends the same message he would have sent in the second step of $GC$. Again, $i$'s message is ephemerally signed and accompanied by $i$'s credential. (From now on, we shall omit
   saying that a verifier ephemerally signs his message and also propagates his credential.)

  \item[{\sc Step 4.}]

In this step, every verifier $i\in SV^{r,4}$
computes the output of $GC$, $(v_i,g_i)$,
and ephemerally signs and sends the same message he would have sent
in the third step of $\BA*$, that is, in the first
step of $\BBA*$, with initial bit 0 if $g_i=2$, and 1 otherwise.

  \item[{\sc Step $s=5,\ldots,m+2$.}]

Such a step, if ever reached, corresponds to step $s-1$ of $\BA*$, and thus to
step $s-3$ of $\BBA*$.

Since our propagation model is sufficiently asynchronous, we must account for the possibility that,
in the middle of such a step $s$, a verifier $i\in SV\rs$ is reached by information proving him that
block $B\upr$ has already been chosen. In this case,
$i$ stops his own execution of round $r$ of $\alg$, and starts executing his round-$(r+1)$ instructions.

Accordingly,  the instructions of a verifier $i\in SV\rs$, in addition to the instructions corresponding to Step $s-3$ of $\BBA*$,
include checking whether
the execution of $\BBA*$ has halted in a prior Step $s'$. Since $\BBA*$ can only halt is a Coin-Fixed-to-0 Step or in a Coin-Fixed-to-1 step, the instructions distinguish whether

A (Ending Condition 0): $s'-2\equiv  0 \; mod\; 3$, or

B (Ending Condition 1):  $s'-2\equiv 1 \; mod\; 3$.

\noindent
In fact, in case A, the block $B\upr$ is non-empty, and thus additional instructions are necessary to ensure that $i$ properly reconstructs $B\upr$, together with its proper certificate $CERT\upr$.
In case B, the block $B\upr$ is empty, and thus $i$ is instructed to set $B\upr=B\upr_\varepsilon =(r,\emptyset,H(Q\uprm,r),H(B\uprm))$,
and  to compute $CERT\upr$.

If, during his execution of step $s$, $i$ does not see any evidence that the block $B\upr$ has already
been generated, then he sends the same message he would have sent in step $s-3$ of $\BBA*$.

\item[{\sc Step $m+3$.}]

If, during step $m+3$, $i\in SV^{r,m+3}$ sees
that the block $B\upr$ was already generated in a prior step $s'$, then he proceeds just as explained above.

Else, rather then sending the same message he would have sent in step $m$ of $\BBA*$, $i$ is  instructed, based on the information in his possession, to compute $B\upr$ and its corresponding certificate $CERT\upr$.

Recall, in fact, that we upperbound by $m+3$ the total number of steps of a round.
\end{description}

\subsection{The Actual Protocol}

Recall that, in each step $s$ of a round $r$, a verifier $i\in SV^{r, s}$
uses his long-term public-secret key pair to produce his credential, $\sigma_i^{r,s}\triangleq SIG_i(r, s, Q^{r-1})$,
as well as $SIG_i\left(Q^{r-1}\right)$ in case $s=1$. Verifier $i$ uses his ephemeral secret key $sk_i^{r,s}$
to sign his $(r, s)$-message $m_i\rs$.
For simplicity, when $r$ and $s$ are clear, we write $esig_i(x)$ rather than $sig_{pk^{r, s}_i}(x)$ to denote $i$'s proper ephemeral signature of a value $x$ in step $s$ of round $r$, and write $ESIG_i(x)$ instead of $SIG_{pk_i^{r,s}}(x)$ to denote
$(i, x, esig_i(x))$.

\begin{center}
\framebox[\textwidth]{
\centering
\begin{minipage}{.98\textwidth}
\begin{center}
Step 1: Block Proposal

\end{center}

Instructions for every user $i\in PK^{r-k}$: User $i$ starts his own Step 1 of round $r$ as soon as he knows $B^{r-1}$.

\begin{itemize}

\item
User $i$ computes $Q^{r-1}$ from the third component of $B^{r-1}$ and checks whether $i\in SV^{r,1}$ or not.

\item
If $i\notin SV^{r,1}$, then $i$ stops his own execution of Step 1 right away.

\item
If $i\in SV^{r, 1}$, that is, if $i$ is a potential leader, then he collects the round-$r$ payments that have been propagated to him so far and computes
a maximal payset $PAY^r_i$ from them.
Next, he computes his ``candidate block''
$B^r_i = (r, PAY^r_i, SIG_i(Q^{r-1}), H(B^{r-1}))$.
Finally, he computes the message $m^{r, 1}_i = (B^r_i, esig_i(H(B^r_i)), \sigma^{r, 1}_i)$,
destroys his ephemeral secret key $sk^{r, 1}_i$,
and then propagates $m^{r, 1}_i$.

\end{itemize}
\end{minipage}
}
\end{center}

\paragraph{Remark.}
In practice, to shorten the global execution of Step 1,
it is important that
the $(r, 1)$-messages are {\em selectively propagated}.
That is,
for every user $i$ in the system,
for the first $(r, 1)$-message that he ever receives and successfully verifies,%
\footnote{That is, all the signatures are correct and
both the block and its hash are valid ---although $i$ does not check whether
the included payset is maximal for its proposer or not.}
player $i$ propagates it as usual.
For all the other $(r, 1)$-messages that player $i$ receives and successfully verifies,
he propagates it
only if the hash value of
the credential it contains is the {\em smallest} among the hash values of the credentials contained in all
$(r, 1)$-messages he has received and successfully verified so far.
Furthermore, as suggested by Georgios Vlachos,
it is useful that
each potential leader $i$
also propagates his credential  $\sigma^{r, 1}_i$ separately:
those small messages travel faster than blocks,
ensure timely propagation of the $m^{r, 1}_j$'s where the contained credentials have small hash values,
while make those with large hash values disappear quickly.

\begin{center}
\framebox[\textwidth]{
\centering
\begin{minipage}{.98\textwidth}
\begin{center}
Step 2: The First Step of the Graded Consensus Protocol $GC$
\end{center}

Instructions for every user $i\in PK^{r-k}$:
User $i$ starts his own Step 2 of round $r$ as soon as he knows $B^{r-1}$.

\begin{itemize}

\item
User $i$ computes $Q^{r-1}$ from the third component of $B^{r-1}$ and checks whether $i\in SV^{r,2}$ or not.

\item
If $i\notin SV^{r, 2}$ then $i$ stops his own execution of Step 2 right away.

\item
If $i\in SV^{r, 2}$, then after waiting an amount of time  $t_2 \triangleq \lambda +\Lambda$, $i$ acts as follows.

\begin{newitemize}

\item[1.]
He finds the user $\ell$ such that $H(\sigma^{r, 1}_{\ell})\leq H(\sigma^{r, 1}_{j})$ for  all credentials $\sigma^{r, 1}_{j}$ that are part of the successfully verified $(r,1)$-messages he has received so far.%
\footnote{Essentially, user $i$ privately decides that the leader of round $r$ is user $\ell$.}

\item[2.]
If he has received from $\ell$ a valid message $m^{r, 1}_{\ell} = (B^r_{\ell}, esig_{\ell}(H(B^r_{\ell})), \sigma^{r, 1}_{\ell})$,%
\footnote{Again, player $\ell$'s signatures and the hashes are all successfully verified, and $PAY^r_{\ell}$ in $B^r_{\ell}$ is a valid payset for round $r$ ---although $i$ does not check whether $PAY^r_{\ell}$ is maximal for $\ell$ or not.}
 then $i$ sets $v'_i \triangleq H(B^r_{\ell})$; otherwise $i$ sets $v'_i \triangleq \bot$.

\item[3.]
$i$  computes the message
$m^{r, 2}_i \triangleq (ESIG_i(v'_i), \sigma^{r, 2}_i)$,%
\footnote{The message
$m^{r, 2}_i$ signals that player $i$ considers $v'_i$ to be the hash of the next block, or considers the next block to be empty.} destroys his ephemeral secret key $sk^{r, 2}_i$, and then propagates $m^{r, 2}_i$.

\end{newitemize}

\end{itemize}
\end{minipage}
}
\end{center}

\begin{center}
\framebox[\textwidth]{
\centering
\begin{minipage}{.98\textwidth}
\begin{center}
Step 3: The Second Step of $GC$
\end{center}

Instructions for every user $i\in PK^{r-k}$:
User $i$ starts his own Step 3 of round $r$ as soon as he knows $B^{r-1}$.

\begin{itemize}

\item
User $i$ computes $Q^{r-1}$ from the third component of $B^{r-1}$ and checks whether $i\in SV^{r,3}$ or not.

\item
If $i\notin SV^{r, 3}$, then $i$ stops his own execution of Step 3 right away.

\item
If $i\in SV^{r, 3}$,  then after waiting an amount of time  $t_3 \triangleq
t_2+2\lambda = 3\lambda+\Lambda$,
$i$ acts as follows.
\begin{newitemize}

\item[1.]
If there exists a value $v'\neq \bot$ such that,
among all the valid messages $m_j^{r, 2}$ he has received,
more than 2/3 of them are of the form
$(ESIG_j(v'), \sigma_j^{r, 2})$, without any contradiction,%
\footnote{That is,
he has not received two valid messages containing $ESIG_j(v')$ and a different $ESIG_j(v'')$ respectively, from a player $j$.
Here and from here on, except in the Ending Conditions defined later, whenever an honest player wants messages of a given form, messages contradicting each other are never counted or considered valid.}
then he
computes the message $m_i^{r, 3} \triangleq (ESIG_i(v'), \sigma^{r, 3}_i)$.
Otherwise, he computes $m_i^{r, 3} \triangleq (ESIG_i(\bot), \sigma^{r, 3}_i)$.

\item[2.]
$i$ destroys his ephemeral secret key $sk^{r, 3}_i$, and then propagates $m_i^{r, 3}$.

\end{newitemize}

\end{itemize}
\end{minipage}
}
\end{center}

\begin{center}
\framebox[\textwidth]{
\centering
\begin{minipage}{.98\textwidth}
\begin{center}
Step 4: Output of $GC$ and The First Step of $\BBA*$

\end{center}

Instructions for every user $i\in PK^{r-k}$:
User $i$ starts his own Step 4 of round $r$ as soon as he knows $B^{r-1}$.

\begin{itemize}

\item
User $i$ computes $Q^{r-1}$ from the third component of $B^{r-1}$ and checks whether $i\in SV^{r,4}$ or not.

\item
If $i\notin SV^{r, 4}$, then $i$ his stops his own execution of Step 4  right away.

\item
If $i\in SV^{r, 4}$,  then after waiting an amount of time  $t_4 \triangleq
t_3+2\lambda = 5\lambda+\Lambda$,
$i$ acts as follows.

\begin{newitemize}

\item[1.] He computes $v_i$ and $g_i$, the output of GC, as follows.

\begin{newitemize}

\item[(a)]
If there exists a value
$v'\neq \bot$ such that,
among all the valid messages $m_j^{r, 3}$ he has received,
more than 2/3 of them are of the form
$(ESIG_j(v'), \sigma_j^{r, 3})$,
then
he sets $v_i \triangleq v'$ and $g_i \triangleq 2$.

\item[(b)]
Otherwise, if
there exists a
value
$v'\neq \bot$ such that,
among all the valid messages $m_j^{r, 3}$ he has received,
more than 1/3 of them are of the form
$(ESIG_j(v'), \sigma_j^{r, 3})$,
then he sets $v_i \triangleq v'$ and $g_i \triangleq 1$.%
\footnote{It can be proved that the $v'$ in case (b), if exists, must be unique.}

\item[(c)]
Else, he sets $v_i \triangleq H(B^r_{\epsilon})$ and $g_i \triangleq 0$.

\end{newitemize}

\item[2.]
He computes $b_i$, the input of $\BBA*$, as follows:

  $b_i \triangleq 0$ if $g_i = 2$, and $b_i \triangleq 1$ otherwise.

\item[3.]
He computes the message $m^{r, 4}_i \triangleq (ESIG_i(b_i), ESIG_i(v_i), \sigma^{r, 4}_i)$, destroys his ephemeral secret key $sk^{r, 4}_i$, and then propagates $m^{r, 4}_i$.

\end{newitemize}

\end{itemize}
\end{minipage}
}
\end{center}

\begin{center}
\framebox[\textwidth]{
\centering
\begin{minipage}{.98\textwidth}
\begin{center}
Step $s$, $5\leq s\leq m+2$, $s -2 \equiv 0 \mod 3$: A Coin-Fixed-To-0 Step of $\BBA*$

\end{center}

Instructions for every user $i\in PK^{r-k}$:
User $i$ starts his own Step $s$ of round $r$ as soon as he knows $B^{r-1}$.

\begin{newitemize}

\item
User $i$ computes $Q^{r-1}$ from the third component of $B^{r-1}$ and checks whether $i\in SV^{r,s}$.

\item
If $i\notin SV^{r, s}$, then $i$ stops his own execution of Step $s$ right away.

\item
If $i\in SV^{r, s}$ then he acts as follows.

\begin{newitemize}

\item[--]
 He waits until an amount of
time $t_s \triangleq t_{s-1}+2\lambda = (2s-3)\lambda+\Lambda$
has passed.

\item[--] {\em Ending Condition 0:}
If, during such waiting and at any point of time,  there exists a string $v\neq \bot$ and a step $s'$ such that

\begin{newitemize}

 \item[(a)] $5\leq s'\leq s$, $s'-2 \equiv 0 \mod 3$ ---that is, Step $s'$ is a Coin-Fixed-To-0 step,

\item[(b)] $i$ has received at least $t_H= \frac{2n}{3}+1$ valid messages
 $m^{r, s'-1}_j = (ESIG_j(0),$ $ESIG_j(v), \sigma^{r, s'-1}_j)$,%
\footnote{Such a message from player $j$ is counted even if player $i$ has also received a message from $j$ signing for 1. Similar things for Ending Condition 1. As shown in the analysis, this is done to ensure that all honest users know $B^r$ within time $\lambda$ from each other.}
  and

\item[(c)] $i$ has received a valid message $m^{r, 1}_j = (B^r_j, esig_j(H(B^r_j)), \sigma^{r, 1}_j)$ with $v = H(B^r_j)$,

\end{newitemize}

then, $i$ stops his own execution of Step $s$ (and in fact of round $r$) right away without propagating anything;
sets $B^r = B^r_j$; and  sets his own $CERT^r$ to be the set of messages $m^{r, s'-1}_j$ of sub-step (b).%
\footnote{User $i$ now knows $B^r$ and his own round $r$ finishes.
He still helps propagating messages as a generic user, but does not initiate any propagation as a $(r, s)$-verifier.
In particular,
he has helped propagating all messages in his $CERT^r$, which is enough for our protocol.
Note that he should also set $b_i \triangleq 0$ for the binary BA protocol, but $b_i$ is not needed in this case anyway.
Similar things for all future instructions.}

\item[--] {\em Ending Condition 1:}
If, during such waiting and at any point of time, there exists a step $s'$ such that

\begin{newitemize}

\item[(a')] $6\leq s'\leq s$, $s'-2\equiv 1 \mod 3$ ---that is, Step $s'$ is a Coin-Fixed-To-1 step, and

\item[(b')]
$i$ has received at least $t_H$ valid messages $m^{r, s'-1}_j = (ESIG_j(1), ESIG_j(v_j),$ $\sigma^{r, s'-1}_j)$,%
\footnote{In this case, it does not matter what the $v_j$'s are.}

\end{newitemize}
then, $i$  stops his own execution of Step $s$ (and in fact of round $r$) right away without propagating anything; sets $B^r = B^r_\epsilon$; and sets his own $CERT^r$ to be the set of messages $m^{r, s'-1}_j$  of sub-step (b').

\item[--]
Otherwise, at the end of the wait, user $i$ does the following.

He sets $v_i$ to be the majority vote of the $v_j$'s in the second components
of all the valid $m^{r, s-1}_j$'s he has received.

He computes $b_i$ as follows.

\begin{newitemize}
\item[]
If more than 2/3 of all the valid  $m^{r, s-1}_j$'s he has received are
of the form $(ESIG_j(0), ESIG_j(v_j), \sigma^{r, s-1}_j)$,
then he sets $b_i \triangleq 0$.

\item[]
Else, if more than 2/3 of all the valid  $m^{r, s-1}_j$'s he has received are
of the form $(ESIG_j(1), ESIG_j(v_j), \sigma^{r, s-1}_j)$,
then he sets $b_i \triangleq 1$.

\item[]
Else, he sets $b_i\triangleq 0$.
\end{newitemize}

He computes the message $m^{r, s}_i \triangleq (ESIG_i(b_i), ESIG_i(v_i), \sigma^{r, s}_i)$,
destroys his ephemeral secret key $sk^{r, s}_i$, and then propagates $m^{r, s}_i$.

\end{newitemize}

\end{newitemize}

\end{minipage}
}
\end{center}

\begin{center}
\framebox[\textwidth]{
\centering
\begin{minipage}{.98\textwidth}
\begin{center}
Step $s$, $6\leq s\leq m+2$, $s -2 \equiv 1 \mod 3$: A Coin-Fixed-To-1 Step of $\BBA*$

\end{center}

Instructions for every user $i\in PK^{r-k}$:
User $i$ starts his own Step $s$ of round $r$ as soon as he knows $B^{r-1}$.

\begin{itemize}

\item
User $i$ computes $Q^{r-1}$ from the third component of $B^{r-1}$ and checks whether $i\in SV^{r,s}$ or not.

\item
If $i\notin SV^{r, s}$, then $i$ stops his own execution of Step $s$ right away.

\item
If $i\in SV^{r, s}$ then he does the follows.

\begin{newitemize}

\item[--]
 He waits until an amount of
time $t_s \triangleq (2s-3)\lambda + \Lambda$ has passed.

\item[--] {\em Ending Condition 0:} The same instructions as Coin-Fixed-To-0 steps.

\item[--] {\em Ending Condition 1:} The same instructions as Coin-Fixed-To-0 steps.

\item[--]
Otherwise, at the end of the wait, user $i$ does the following.

\smallskip
He sets $v_i$ to be the majority vote of the $v_j$'s in the second components
of all the valid $m^{r, s-1}_j$'s he has received.

\smallskip
He computes $b_i$ as follows.

\begin{newitemize}

\item[]
If more than 2/3 of all the valid  $m^{r, s-1}_j$'s he has received are
of the form $(ESIG_j(0), ESIG_j(v_j), \sigma^{r, s-1}_j)$,
then he sets $b_i \triangleq 0$.

\item[]
Else, if more than 2/3 of all the valid  $m^{r, s-1}_j$'s he has received are
of the form $(ESIG_j(1), ESIG_j(v_j), \sigma^{r, s-1}_j)$,
then he sets $b_i \triangleq 1$.

\item[]
Else, he sets $b_i\triangleq 1$.
\end{newitemize}

\smallskip
He computes the message
$m^{r, s}_i \triangleq (ESIG_i(b_i), ESIG_i(v_i), \sigma^{r, s}_i)$,
destroys his ephemeral secret key $sk^{r, s}_i$,
and then propagates $m^{r, s}_i$.

\end{newitemize}
\end{itemize}
\end{minipage}
}
\end{center}

\begin{center}
\framebox[\textwidth]{
\centering
\begin{minipage}{.98\textwidth}
\begin{center}
Step $s$, $7\leq s\leq m+2$, $s -2 \equiv 2 \mod 3$: A Coin-Genuinely-Flipped Step of $\BBA*$

\end{center}

Instructions for every user $i\in PK^{r-k}$:
User $i$ starts his own Step $s$ of round $r$ as soon as he knows $B^{r-1}$.

\begin{itemize}

\item
User $i$ computes $Q^{r-1}$ from the third component of $B^{r-1}$ and checks whether $i\in SV^{r,s}$ or not.

\item
If $i\notin SV^{r, s}$, then $i$ stops his own execution of Step $s$ right away.

\item
If $i\in SV^{r, s}$ then he does the follows.

\begin{newitemize}

\item[--]
 He waits until an amount of
time $t_s \triangleq (2s-3)\lambda + \Lambda$ has passed.

\item[--] {\em Ending Condition 0:} The same instructions as Coin-Fixed-To-0 steps.

\item[--] {\em Ending Condition 1:} The same instructions as Coin-Fixed-To-0 steps.

\item[--]
Otherwise, at the end of the wait, user $i$ does the following.

\smallskip
He sets $v_i$ to be the majority vote of the $v_j$'s in the second components
of all the valid $m^{r, s-1}_j$'s he has received.

\smallskip
He computes $b_i$ as follows.

\begin{newitemize}

\item[]
If more than 2/3 of all the valid  $m^{r, s-1}_j$'s he has received are
of the form $(ESIG_j(0), ESIG_j(v_j), \sigma^{r, s-1}_j)$,
then he sets $b_i \triangleq 0$.

\item[]
Else, if more than 2/3 of all the valid  $m^{r, s-1}_j$'s he has received are
of the form $(ESIG_j(1), ESIG_j(v_j), \sigma^{r, s-1}_j)$,
then he sets $b_i \triangleq 1$.

\item[]
Else,
let $SV^{r, s-1}_i$ be the set of $(r, s-1)$-verifiers from whom he has received a valid message $m^{r, s-1}_j$.
He sets $b_i \triangleq \lsb(\min_{j\in SV^{r, s-1}_i} H(\sigma^{r, s-1}_j))$.

\end{newitemize}

He computes the message $m^{r, s}_i \triangleq (ESIG_i(b_i), ESIG_i(v_i), \sigma^{r, s}_i)$,
destroys his ephemeral secret key $sk^{r, s}_i$,
 and then propagates $m^{r, s}_i$.
\end{newitemize}

\end{itemize}
\end{minipage}
}
\end{center}

\begin{center}
\framebox[\textwidth]{
\centering
\begin{minipage}{.98\textwidth}
\begin{center}
Step $m+3$: The Last Step of $\BBA*$ %
\footnote{With overwhelming probability $\BBA*$ has ended before this step, and we specify this step for completeness.}

\end{center}

Instructions for every user $i\in PK^{r-k}$:
User $i$ starts his own Step $m+3$ of round $r$ as soon as he knows $B^{r-1}$.
\begin{itemize}

\item
User $i$ computes $Q^{r-1}$ from the third component of $B^{r-1}$ and checks whether $i\in SV^{r,m+3}$ or not.

\item
If $i\notin SV^{r, m+3}$, then $i$ stops his own execution of Step $m+3$ right away.

\item
If $i\in SV^{r, m+3}$ then he does the follows.

\begin{newitemize}

\item[--]
 He waits until an amount of
time $t_{m+3} \triangleq t_{m+2}+2\lambda = (2m+3)\lambda + \Lambda$ has passed.

\item[--] {\em Ending Condition 0:} The same instructions as Coin-Fixed-To-0 steps.

\item[--] {\em Ending Condition 1:} The same instructions as Coin-Fixed-To-0 steps.

\item[--]
Otherwise, at the end of the wait, user $i$ does the following.

\smallskip
He sets $out_i \triangleq 1$ and $B^r \triangleq B^r_{\epsilon}$.

\smallskip
He computes the message $m^{r, m+3}_i = (ESIG_i(out_i), ESIG_i(H(B^r)), \sigma^{r, m+3}_i)$,
destroys his ephemeral secret key $sk^{r, m+3}_i$,
and then propagates $m^{r, m+3}_i$ to certify $B^r$.%
\footnote{A certificate from Step $m+3$ does not have to include $ESIG_i(out_i)$.
We include it for uniformity only:
the certificates now have a uniform format no matter in which step they are
generated.}

\end{newitemize}

\end{itemize}
\end{minipage}
}
\end{center}

\begin{center}
\framebox[\textwidth]{
\centering
\begin{minipage}{.98\textwidth}
\begin{center}
Reconstruction of the Round-$r$ Block by Non-Verifiers

\end{center}

Instructions for every user $i$ in the system:
User $i$ starts his own round $r$ as soon as he knows $B^{r-1}$, and waits for block information as follows.

\begin{itemize}

\item[--]
If, during such waiting and at any point of time,  there exists a string $v$ and a step $s'$ such that

\begin{newitemize}

 \item[(a)]
 $5\leq s'\leq m+3$ with $s'-2\equiv 0 \mod 3$,

 \item[(b)] $i$ has received at least $t_H$ valid messages
 $m^{r, s'-1}_j = (ESIG_j(0), ESIG_j(v), \sigma^{r, s'-1}_j)$,
  and

\item[(c)] $i$ has received a valid message $m^{r, 1}_j = (B^r_j, esig_j(H(B^r_j)), \sigma^{r, 1}_j)$ with $v = H(B^r_j)$,

\end{newitemize}

then, $i$ stops his own execution of round $r$ right away;
sets $B^r = B^r_j$; and  sets his own $CERT^r$ to be the set of messages $m^{r, s'-1}_j$ of sub-step (b).

\item[--]
If, during such waiting and at any point of time, there exists a step $s'$ such that

\begin{newitemize}

\item[(a')]
$6\leq s'\leq m+3$ with $s'-2\equiv 1\mod 3$, and

\item[(b')]
$i$ has received at least $t_H$ valid messages $m^{r, s'-1}_j = (ESIG_j(1), ESIG_j(v_j), \sigma^{r, s'-1}_j)$,
 \end{newitemize}

then, $i$  stops his own execution of round $r$ right away;
sets $B^r = B^r_\epsilon$; and sets his own $CERT^r$ to be the set of messages $m^{r, s'-1}_j$  of sub-step (b').

\item[--]
If, during such waiting and at any point of time,
$i$ has received at least $t_H$ valid messages $m^{r, m+3}_j = (ESIG_j(1), ESIG_j(H(B^r_\epsilon)), \sigma^{r, m+3}_j)$,
then $i$  stops his own execution of round $r$ right away,
sets $B^r = B^r_\epsilon$, and sets his own $CERT^r$ to be the set of messages $m^{r, m+3}_j$ for $1$ and $H(B^r_\epsilon)$.

\end{itemize}
\end{minipage}
}
\end{center}

\subsection{Analysis of $\bm{\alg_1}$}

We introduce the following notations for each round $r\geq 0$, used in the analysis.

\begin{newitemize}

\item
Let $T^r$ be the time when the first honest user knows $B^{r-1}$.

\item
Let $I^{r+1}$ be the interval $[T^{r+1}, T^{r+1}+\lambda]$.

\end{newitemize}
Note that $T^0 = 0$ by the initialization of the protocol.
For each $s\geq 1$ and $i\in SV^{r, s}$, recall that $\alpha^{r, s}_i$ and $\beta^{r, s}_i$ are respectively the starting time and the ending time of player $i$'s step $s$.
Moreover, recall that $t_s = (2s-3)\lambda + \Lambda$ for each $2\leq s\leq m+3$. In addition, let $I^0\triangleq \{0\}$ and $t_1\triangleq 0$.

Finally,
recall that
$L^r\leq m/3$ is  a random variable
representing the number of Bernoulli trials needed to see a 1, when each trial is 1 with probability
$\frac{p_h}{2}$ and there are at most $m/3$ trials. If all trials fail then $L^r\triangleq m/3$.

In the analysis we ignore computation time, as it is in fact  negligible relative to the time needed to propagate messages.
In any case, by using slightly larger $\lambda$ and $\Lambda$, the computation time can be incorporated into the analysis directly.
Most of the statements below hold ``with overwhelming probability,'' and we may not repeatedly emphasize this fact in the analysis.

\subsection{Main Theorem}

\begin{theorem}\label{thm:main}
The following properties hold with overwhelming probability for each round $r\geq 0$:

\begin{newitemize}
\item[1.]
All honest users agree on the same block $B^r$.

\item[2.]
When the leader $\ell^r$ is honest,
the block $B^r$ is generated by $\ell^r$,
$B^r$ contains a maximal payset received by $\ell^r$ by time $\alpha^{r, 1}_{\ell^r}$,
$T^{r+1} \leq T^r + 8\lambda + \Lambda$ and
all honest users know $B^r$ in the time interval $I^{r+1}$.

\item[3.]
When the leader $\ell^r$ is malicious,
$T^{r+1} \leq T^{r} + (6L^r+10)\lambda + \Lambda$ and
all honest users know $B^r$ in the time interval $I^{r+1}$.

\item[4.] $p_h= h^2(1+h-h^2)$ for $L^r$, and the leader $\ell^r$ is honest with probability at least $p_h$.
\end{newitemize}
\end{theorem}
Before proving our main theorem, let us make two remarks.

\paragraph{Remarks.}

\begin{itemize}

\item {\em Block-Generation and True Latency.}
The time to generate block $B\upr$ is defined to be $T\uprp - T\upr$. That is, it is defined to be the difference between the first time some
honest user learns $B\upr$ and the first time some honest user learns $B\uprm$.
When the round-$r$ leader is honest, Property~2 our main theorem guarantees that the {\em exact} time to generate $B\upr$ is   $8\lambda +\Lambda$ time,
no matter what the precise value of $h>2/3$ may be. When the leader is malicious, Property~3 implies
that the {\em expected} time to generate $B\upr$ is upperbounded by $(\frac{12}{p_h}+10)\lambda+\Lambda$,
again no matter the precise value of $h$.%
\footnote{Indeed,
$\bE[T^{r+1}-T^r] \leq (6\bE[L^r]+10)\lambda + \Lambda = (6\cdot\frac{2}{p_h}+10)\lambda + \Lambda = (\frac{12}{p_h}+10)\lambda+\Lambda$.}
However, the expected time to generate $B\upr$ depends on the precise value of $h$. Indeed, by Property 4,
$p_h=h^2(1+h-h^2)$ and
the leader is honest with probability at least $p_h$,
thus
$$\bE[T^{r+1}-T^r]\leq h^2(1+h-h^2)\cdot (8\lambda+\Lambda) + (1-h^2(1+h-h^2))((\frac{12}{h^2(1+h-h^2)}+10)\lambda+\Lambda).$$
For instance, if $h = 80\%$, then $\bE[T^{r+1}-T^r] \leq 12.7\lambda +\Lambda.$

\item
{\em $\lambda$ vs. $\Lambda$.}
Note that the size of the messages sent
by the verifiers in a step $\alg$ is dominated by the length of the digital signature keys, which can remain fixed,  even when the number of users is enormous. Also note that, in any step $s>1$, the same expected number $n$ of verifiers can be used whether the number of users is 100K, 100M, or 100M. This is so because $n$ solely depends on $h$ and $F$.
In sum, therefore, barring a sudden need to increase secret key length, the value of $\lambda$ should remain the same no matter how large the number of users may be in the foreseeable future.

By contrast, for any transaction rate, the number of transactions grows with the number of users. Therefore, to process all new transactions in a timely fashion, the size of a block should also grow with the number of users, causing  $\Lambda$ to grow too.
Thus,  in the long run, we should have $\lambda <<\Lambda$.
Accordingly, it is proper to have a larger coefficient for $\lambda$, and actually a coefficient of 1 for $\Lambda$.

\end{itemize}

\begin{proof}[Proof of Theorem \ref{thm:main}]
We prove Properties 1--3 by induction:
assuming they
hold for round $r-1$
(without loss of generality, they automatically hold for ``round -1'' when $r=0$),
 we prove them for round $r$.

Since $B^{r-1}$ is uniquely defined by the inductive hypothesis, the set $SV^{r, s}$ is uniquely defined for each step $s$ of round $r$.
By the choice of $n_1$, $SV^{r, 1}\neq \emptyset$
with overwhelming probability.
We now state the following two lemmas,
proved in Sections \ref{subsec:complete}
and \ref{subsec:sound}.
Throughout the induction and in the proofs of the two lemmas, the analysis for round 0 is almost the same as the inductive step, and we will highlight the differences when they occur.

\begin{lemma}\label{lem:honestleader_roundr}[Completeness Lemma]
Assuming Properties 1--3 hold for round $r-1$, when the leader $\ell^r$ is honest,
with overwhelming probability,
\begin{newitemize}
\item
All honest users agree on the same block $B^r$,
which is generated by $\ell^r$ and
contains a maximal payset received by $\ell^r$ by time $\alpha^{r, 1}_{\ell^r}\in I^r$; and

\item
$T^{r+1} \leq T^r + 8\lambda + \Lambda$ and
all honest users know $B^r$ in the time interval $I^{r+1}$.
\end{newitemize}

\end{lemma}

\begin{lemma}\label{lem:leadermalicious_roundr}[Soundness Lemma]
Assuming Properties 1--3 hold for round $r-1$,
when the leader $\ell^r$ is malicious,
with overwhelming probability, all honest users agree on the same block $B^r$,
$T^{r+1} \leq T^r+(6L^r+10)\lambda + \Lambda$ and
all honest users know $B^r$ in the time interval $I^{r+1}$.
\end{lemma}

Properties 1--3 hold by applying Lemmas \ref{lem:honestleader_roundr} and \ref{lem:leadermalicious_roundr}
 to $r=0$ and to the inductive step.
Finally, we restate Property 4 as the following lemma, proved in Section \ref{subsec:security}.
\begin{lemma}\label{lem:security}
Given Properties 1--3 for each round before $r$,
$p_h = h^2(1+h-h^2)$ for $L^r$, and
the leader $\ell^r$ is honest with probability at least $p_h$.
\end{lemma}
Combining the above three lemmas together, Theorem \ref{thm:main} holds.
\end{proof}

\medskip
The lemma below states several important properties about round $r$ given the inductive hypothesis,
and will be used in the proofs of the above three lemmas.

\begin{lemma}\label{lem:basic}
Assume Properties 1--3 hold for round $r-1$.
For each step $s\geq 1$ of round $r$ and each honest verifier $i\in HSV^{r, s}$, we have that
\begin{newitemize}
\item[(a)] $\alpha^{r, s}_i \in I^r$;

\item[(b)] if player $i$ has waited an amount of time $t_s$,
then $\beta^{r, s}_i \in [T^r+t_s, T^r+\lambda+t_s]$ for $r>0$ and $\beta^{r, s}_i = t_s$ for $r=0$; and

\item[(c)] if player $i$ has waited an amount of time $t_s$,
then by time $\beta^{r, s}_i$, he has received all messages sent by all honest verifiers $j\in HSV^{r, s'}$ for all steps $s'<s$.
\end{newitemize}
Moreover, for each step $s\geq 3$, we have that
\begin{newitemize}
\item[(d)] there do not exist two different players $i, i'\in SV^{r, s}$
and two different values $v, v'$ of the same length,
 such that
both players have waited an amount of time $t_s$,
more than 2/3 of all the valid messages $m^{r, s-1}_j$ player $i$ receives
have signed for $v$, and
more than 2/3 of all the valid messages $m^{r, s-1}_j$ player $i'$ receives
have signed for $v'$.

\end{newitemize}
\end{lemma}

\begin{proof}
Property (a) follows directly from the inductive hypothesis, as player $i$ knows $B^{r-1}$ in the time interval $I^r$ and starts his own step $s$ right away.
Property (b) follows directly from (a): since player $i$
has waited  an amount of time $t_s$ before acting,
 $\beta^{r, s}_i = \alpha^{r, s}_i+t_s$.
 Note that $\alpha^{r, s}_i = 0$ for $r=0$.

We now prove Property (c). If $s=2$, then by Property (b),
for all verifiers $j\in HSV^{r, 1}$
 we have
$$\beta^{r, s}_i = \alpha^{r, s}_i + t_s \geq T^r+t_s = T^r+\lambda + \Lambda \geq \beta^{r, 1}_j + \Lambda.$$
Since each verifier $j\in HSV^{r, 1}$ sends his message
at time $\beta^{r, 1}_j$ and
the message reaches all honest users in at most $\Lambda$ time,
by time $\beta^{r, s}_i$ player $i$ has received the messages sent by all verifiers in $HSV^{r, 1}$ as desired.

If $s>2$, then $t_s = t_{s-1}+2\lambda$. By Property (b),
for all steps $s'<s$ and all verifiers $j\in HSV^{r, s'}$,
$$\beta^{r, s}_i = \alpha^{r, s}_i + t_s \geq T^r+t_s = T^r+t_{s-1}+2\lambda \geq T^r+t_{s'}+2\lambda
=T^r+\lambda + t_{s'}+\lambda \geq \beta^{r, s'}_j + \lambda.$$
Since each verifier $j\in HSV^{r, s'}$ sends his message
at time $\beta^{r, s'}_j$ and the message reaches all honest users in at most $\lambda$ time,
by time $\beta^{r, s}_i$ player $i$ has received all messages sent by all honest verifiers in $HSV^{r, s'}$ for all $s'<s$. Thus Property (c) holds.

Finally, we prove Property (d).
Note that the verifiers $j\in SV^{r, s-1}$
sign at most two things  in Step $s-1$
using their ephemeral secret keys:
a value $v_j$ of the same length as the output of the hash function,
and also a bit $b_j\in \{0, 1\}$ if $s-1\geq 4$.
That is why in the statement of the lemma
we require that $v$ and $v'$ have the same length:
many verifiers may have signed both a hash value $v$ and a bit $b$,
thus both pass the $2/3$ threshold.

Assume for the sake of contradiction that there exist
the desired verifiers $i, i'$ and values $v, v'$.
Note that some malicious verifiers in $MSV^{r, s-1}$ may have signed both $v$ and $v'$,
but each honest verifier in $HSV^{r, s-1}$ has signed at most one of them.
By Property (c), both $i$ and $i'$ have received
all messages sent by all honest verifiers in $HSV^{r, s-1}$.

Let $HSV^{r, s-1}(v)$ be the set of honest
$(r, s-1)$-verifiers who have signed $v$,
$MSV_i^{r, s-1}$ the set of malicious
$(r, s-1)$-verifiers from whom  $i$ has received a valid message,
and
$MSV_i^{r, s-1}(v)$ the subset of $MSV_i^{r, s-1}$
 from whom  $i$ has received a valid message signing $v$.
By the requirements for $i$ and $v$,
we have
\begin{equation}\label{equ:6-1}
ratio\triangleq \frac{|HSV^{r, s-1}(v)|+|MSV_i^{r, s-1}(v)|}{|HSV^{r, s-1}| + |MSV_i^{r, s-1}|}> \frac{2}{3}.
\end{equation}
We first show
\begin{equation}\label{equ:6-2}
|MSV_i^{r, s-1}(v)|\leq |HSV^{r, s-1}(v)|.
\end{equation}
 Assuming otherwise,
 by the relationships among the parameters,
with overwhelming probability $|HSV^{r, s-1}|> 2|MSV^{r, s-1}|\geq 2|MSV_i^{r, s-1}|$,
thus
$$ratio <
\frac{|HSV^{r, s-1}(v)|+|MSV_i^{r, s-1}(v)|}{3|MSV_i^{r, s-1}|}
< \frac{2|MSV_i^{r, s-1}(v)|}{3|MSV_i^{r, s-1}|}\leq \frac{2}{3},$$
contradicting Inequality \ref{equ:6-1}.

Next, by Inequality \ref{equ:6-1} we have
\begin{eqnarray*}
& & 2|HSV^{r, s-1}| + 2|MSV_i^{r, s-1}| < 3|HSV^{r, s-1}(v)|+ 3|MSV_i^{r, s-1}(v)| \\
&\leq& 3|HSV^{r, s-1}(v)|+ 2|MSV_i^{r, s-1}|+|MSV_i^{r, s-1}(v)|.
\end{eqnarray*}
Combining with Inequality \ref{equ:6-2},
$$2|HSV^{r, s-1}| < 3|HSV^{r, s-1}(v)|+ |MSV_i^{r, s-1}(v)| \leq 4|HSV^{r, s-1}(v)|,$$
which implies
$$|HSV^{r, s-1}(v)|> \frac{1}{2}|HSV^{r, s-1}|.$$
Similarly, by the requirements for $i'$ and $v'$, we have
$$|HSV^{r, s-1}(v')|> \frac{1}{2}|HSV^{r, s-1}|.$$
Since an honest verifier $j\in HSV^{r, s-1}$ destroys his ephemeral
secret key $sk^{r, s-1}_j$
before propagating his message,
the Adversary cannot forge $j$'s signature for a value that $j$ did not sign,
after learning that $j$ is a verifier.
Thus, the two inequalities above imply
$|HSV^{r, s-1}| \geq |HSV^{r, s-1}(v)|+ |HSV^{r, s-1}(v')|> |HSV^{r, s-1}|$, a contradiction.
Accordingly, the desired $i, i', v, v'$ do not exist, and Property (d) holds.
\end{proof}

\subsection{The Completeness Lemma}\label{subsec:complete}

\noindent
{\bf Lemma \ref{lem:honestleader_roundr}.} [Completeness Lemma, restated]
{\em Assuming Properties 1--3 hold for round $r-1$, when the leader $\ell^r$ is honest,
with overwhelming probability,
\begin{newitemize}
\item
All honest users agree on the same block $B^r$,
which is generated by $\ell^r$ and
contains a maximal payset received by $\ell^r$ by time $\alpha^{r, 1}_{\ell^r}\in I^r$; and

\item
$T^{r+1} \leq T^r + 8\lambda + \Lambda$ and
all honest users know $B^r$ in the time interval $I^{r+1}$.
\end{newitemize}
}

\begin{proof}

By the inductive hypothesis and Lemma \ref{lem:basic},
for each step $s$ and verifier $i\in HSV^{r, s}$,
$\alpha^{r, s}_i\in I^r$.
Below we analyze the protocol step by step.

\paragraph{Step 1.}
By definition, every honest verifier $i\in HSV^{r, 1}$
propagates the desired message $m^{r, 1}_i$
at time $\beta^{r, 1}_i = \alpha^{r, 1}_i$,
where
$m^{r, 1}_i = (B^r_i, esig_i(H(B^r_i)), \sigma^{r, 1}_i)$,
$B^r_i = (r, PAY^r_i, SIG_i(Q^{r-1}), H(B^{r-1}))$,
and $PAY^r_i$ is a maximal payset among all payments that $i$ has seen by time $\alpha^{r, 1}_i$.

\paragraph{Step 2.}
Arbitrarily fix an honest verifier $i\in HSV^{r, 2}$.
By Lemma \ref{lem:basic},
when player $i$ is done waiting at time $\beta^{r, 2}_i = \alpha^{r, 2}_i + t_2$,
he has received all messages
sent by verifiers in $HSV^{r, 1}$,
including
$m^{r, 1}_{\ell^r}$.
By the definition of $\ell^r$,
there does not exist another player in $PK^{r-k}$ whose credential's hash value is smaller than $H(\sigma^{r, 1}_{\ell^r})$.
Of course, the Adversary can corrupt $\ell^r$ after seeing that $H(\sigma^{r, 1}_{\ell^r})$ is very small,
but by that time player $\ell^r$ has destroyed his ephemeral key and the message $m^{r, 1}_{\ell^r}$ has been propagated.
Thus verifier $i$ sets his own leader to be player $\ell^r$.
Accordingly, at time $\beta^{r, 2}_i$,
 verifier $i$ propagates
$m^{r, 2}_i =  (ESIG_i(v'_i), \sigma^{r, 2}_i)$,
where
$v'_i = H(B^r_{\ell^r})$.
When $r=0$, the only difference is that $\beta^{r, 2}_i = t_2$ rather than being in a range.
Similar things can be said for future steps and we will not emphasize them again.

\paragraph{Step 3.}
Arbitrarily fix an honest verifier $i\in HSV^{r, 3}$.
By Lemma \ref{lem:basic},
when player $i$ is done waiting at time $\beta^{r, 3}_i = \alpha^{r, 3}_i + t_3$, he
has received all messages sent by verifiers in $HSV^{r, 2}$.

By the relationships among the parameters, with overwhelming probability $|HSV^{r, 2}|> 2|MSV^{r, 2}|$.
Moreover, no honest verifier would sign contradicting messages, and
the Adversary cannot forge a signature of an honest verifier after the latter has destroyed his corresponding ephemeral secret key.
Thus more than $2/3$ of all the valid $(r, 2)$-messages $i$ has received
are from honest verifiers and of the form
$m^{r, 2}_j=(ESIG_j(H(B^r_{\ell^r})), \sigma^{r, 2}_j)$, with no contradiction.

Accordingly, at time $\beta^{r, 3}_i$
player $i$ propagates
$m^{r, 3}_i = (ESIG_i(v'), \sigma^{r, 3}_i)$, where $v' = H(B^r_{\ell^r})$.

\paragraph{Step 4.}
Arbitrarily fix an honest verifier $i\in HSV^{r, 4}$.
By Lemma \ref{lem:basic},
player $i$ has received all messages sent by verifiers in $HSV^{r, 3}$ when he is done waiting at time
$\beta^{r, 4}_i = \alpha^{r, 4}_i + t_4$.
Similar to Step 3,
more than $2/3$ of all the valid $(r, 3)$-messages $i$ has received
are from honest verifiers and of the form
$m^{r, 3}_j=(ESIG_j(H(B^r_{\ell^r})), \sigma^{r, 3}_j)$.

Accordingly, player $i$ sets $v_i = H(B^r_{\ell^r})$, $g_i = 2$ and $b_i = 0$.
At time $\beta^{r, 4}_i = \alpha^{r, 4}_i + t_4$
he propagates
$m^{r, 4}_i = (ESIG_i(0), ESIG_i(H(B^r_{\ell^r})), \sigma^{r, 4}_i)$.

\paragraph{Step 5.}
Arbitrarily fix an honest verifier $i\in HSV^{r, 5}$.
By Lemma \ref{lem:basic},
player $i$ would have received all messages sent by the verifiers in $HSV^{r, 4}$ if he has waited till time $\alpha^{r, 5}_i+t_5$.
Note that
$|HSV^{r, 4}|\geq t_H$.%
\footnote{Strictly speaking, this happens with very high probability but not necessarily overwhelming.
However, this probability slightly effects the running time of the protocol, but does not affect its correctness. When $h=80\%$, then $|HSV^{r, 4}|\geq t_H$ with probability
 $1-10^{-8}$.
If this event does not occur, then the protocol will continue for another 3 steps.
As the probability that this does not occur in two steps is negligible,
the protocol will finish at Step 8. In expectation, then, the number of steps needed is almost 5.}
Also note that
all verifiers in $HSV^{r, 4}$ have signed for $H(B^r_{\ell^r})$.

As $|MSV^{r, 4}|< t_H$,
there does not exist any $v'\neq H(B^r_{\ell^r})$ that could have been signed by $t_H$ verifiers in $SV^{r, 4}$ (who would necessarily be malicious),
so player $i$ does not stop before
he has received $t_H$ valid messages
$m^{r, 4}_j = (ESIG_j(0), ESIG_j(H(B^r_{\ell^r})), \sigma^{r, 4}_j)$.
Let $T$ be the time when the latter event happens.
Some of those messages may be from malicious players, but because $|MSV^{r, 4}|< t_H$,
at least one of them is from an honest verifier in $HSV^{r, 4}$ and is sent after time $T^r+t_4$.
Accordingly, $T\geq T^r+t_4 > T^r + \lambda+\Lambda\geq \beta^{r, 1}_{\ell^r}+\Lambda$, and by time $T$ player $i$ has also received the message
$m^{r, 1}_{\ell^r}$.
By the construction of the protocol, player $i$ stops at time $\beta^{r, 5}_i = T$
 without propagating anything;
 sets
$B^r = B^r_{\ell^r}$; and sets his own $CERT^r$ to be the set of $(r, 4)$-messages for $0$ and $H(B^r_{\ell^r})$ that he has received.

\paragraph{Step $s>5$.}
Similarly,
for any step $s>5$ and any verifier $i\in HSV^{r, s}$,
player $i$ would have received all messages sent by the
verifiers in $HSV^{r, 4}$ if he has waited till time $\alpha^{r, s}_i+t_s$.
By the same analysis,
player~$i$ stops without propagating anything, setting $B^r = B^r_{\ell^r}$ (and setting his own $CERT^r$ properly).
Of course, the malicious verifiers may not stop and may propagate arbitrary messages,
but because $|MSV^{r, s}|<t_H$,
by induction
no other $v'$ could be signed by $t_H$ verifiers in any step $4\leq s'<s$,
thus the honest verifiers only stop because they have received $t_H$ valid $(r, 4)$-messages for $0$ and $H(B^r_{\ell^r})$.

\paragraph{Reconstruction of the Round-$r$ Block.}
The analysis of Step 5 applies to a generic honest user $i$ almost without any change.
Indeed, player $i$ starts his own round $r$ in the interval $I^r$ and will only stop at a time $T$ when he has received $t_H$ valid
$(r, 4)$-messages for $H(B^r_{\ell^r})$.
Again because at least one of those messages are from honest verifiers and are sent after time $T^r+t_4$,
player $i$ has also received $m^{r, 1}_{\ell^r}$ by time $T$.
Thus he sets $B^r = B^r_{\ell^r}$ with the proper $CERT^r$.

\medskip
It only remains to show that all honest users finish their round $r$ within the time interval $I^{r+1}$.
By the analysis of Step 5, every honest verifier $i\in HSV^{r, 5}$ knows $B^r$ on or before $\alpha^{r, 5}_i + t_5\leq T^r+\lambda + t_5 = T^r + 8\lambda +\Lambda$.
Since $T^{r+1}$ is the time when the first honest user $i^r$ knows $B^r$,
we have
$$T^{r+1}\leq T^r + 8\lambda +\Lambda$$
 as desired.
Moreover, when player $i^r$ knows $B^r$, he has already helped propagating the messages in his $CERT^r$.
Note that all those messages will be received by all honest users within time $\lambda$,
even if player $i^r$ were the first player to propagate them.
Moreover, following the analysis above we have $T^{r+1}\geq T^r + t_4\geq \beta^{r, 1}_{\ell^r} + \Lambda$, thus all honest users have received $m^{r, 1}_{\ell^r}$ by time $T^{r+1}+\lambda$.
Accordingly, all honest users know $B^r$ in the time interval $I^{r+1} = [T^{r+1}, T^{r+1}+\lambda]$.

Finally, for $r=0$ we actually have $T^1\leq t_4+\lambda = 6\lambda + \Lambda$.
Combining everything together, Lemma \ref{lem:honestleader_roundr} holds.
\end{proof}

\subsection{The Soundness Lemma}\label{subsec:sound}

\noindent
{\bf Lemma \ref{lem:leadermalicious_roundr}.} [Soundness Lemma, restated]
{\em Assuming Properties 1--3 hold for round $r-1$,
when the leader $\ell^r$ is malicious,
with overwhelming probability, all honest users agree on the same block $B^r$,
$T^{r+1} \leq T^r+(6L^r+10)\lambda + \Lambda$ and
all honest users know $B^r$ in the time interval $I^{r+1}$.
}

\begin{proof}
We consider the two parts of the protocol, GC and $\BBA*$, separately.
\paragraph{GC.}
By the inductive hypothesis and by Lemma \ref{lem:basic},
for any step $s\in \{2, 3, 4\}$ and any honest verifier $i\in HSV^{r, s}$,
when player $i$ acts at time $\beta^{r, s}_i = \alpha^{r, s}_i+t_s$,
he has received all messages sent by all the honest verifiers in steps $s'<s$.
We distinguish two possible cases for step 4.
\begin{newitemize}
\item[Case 1.]
{\em No verifier $i\in HSV^{r, 4}$ sets $g_i = 2$.}

In this case, by definition $b_i = 1$ for all verifiers $i\in HSV^{r, 4}$. That is, they start with an agreement on $1$ in the binary BA protocol.
They may not have an agreement on their $v_i$'s, but this does not matter as we will see in the binary BA.

\item[Case 2.]
{\em There exists a verifier $\hat{i}\in HSV^{r, 4}$ such that $g_{\hat{i}} = 2$.}

In this case, we show that
\begin{newitemize}
\item[(1)]
 $g_i\geq 1$ for all $i\in HSV^{r, 4}$,

\item[(2)] there exists a  value $v'$ such that $v_i = v'$ for all $i\in HSV^{r, 4}$,
and

\item[(3)] there exists a valid message $m^{r, 1}_{\ell}$ from some verifier $\ell\in SV^{r, 1}$ such that $v'= H(B^r_\ell)$.
\end{newitemize}

Indeed, since player $\hat{i}$ is honest and sets $g_{\hat{i}} = 2$,
more than $2/3$ of all the valid messages $m^{r, 3}_j$ he has received
are for the same value
$v'\neq \bot$, and he has set $v_{\hat{i}} = v'$.

By Property (d) in Lemma \ref{lem:basic},
for any other honest
$(r, 4)$-verifier $i$,
it cannot be that
more than $2/3$ of all the valid messages $m^{r, 3}_j$ that $i'$ has received
are for the same value $v''\neq v'$.
Accordingly, if $i$ sets $g_i = 2$,
it must be that $i$ has seen $>2/3$ majority for $v'$ as well and
set $v_i = v'$, as desired.

\smallskip
Now consider an arbitrary verifier $i\in HSV^{r, 4}$ with $g_i <2$.
Similar to the analysis of Property (d) in Lemma \ref{lem:basic},
because player $\hat{i}$ has seen $>2/3$ majority for $v'$,
more than $\frac{1}{2}|HSV^{r, 3}|$ honest $(r, 3)$-verifiers have signed $v'$.
Because $i$ has received all messages by honest $(r, 3)$-verifiers by time
$\beta^{r, 4}_i = \alpha^{r, 4}_i + t_4$, he has in particular
received more than $\frac{1}{2}|HSV^{r, 3}|$ messages  from them  for $v'$.
Because $|HSV^{r, 3}|>2|MSV^{r, 3}|$, $i$ has seen $>1/3$ majority for $v'$.
Accordingly, player $i$ sets $g_i = 1$, and Property (1) holds.

\smallskip
Does player $i$ necessarily set $v_i = v'$?
Assume there exists a different value $v''\neq \bot$ such that player $i$ has also
seen $>1/3$ majority for $v''$.
Some of those messages may be from malicious verifiers,
but at least one of them is from some honest verifier $j\in HSV^{r, 3}$:
indeed, because $|HSV^{r, 3}|>2|MSV^{r, 3}|$ and $i$ has received all messages from $HSV^{r, 3}$,
the set of malicious verifiers from whom $i$ has received a valid $(r, 3)$-message counts for $<1/3$ of all the valid messages he has received.

By definition, player $j$ must have seen
$>2/3$ majority for $v''$ among all the valid $(r, 2)$-messages he has received.
However, we already have that some other honest $(r, 3)$-verifiers have
seen $>2/3$ majority for $v'$ (because they signed $v'$).
By Property (d) of Lemma \ref{lem:basic}, this cannot happen and such a value $v''$ does not exist.
Thus player $i$ must have set $v_i = v'$ as desired, and Property (2) holds.

\smallskip
Finally, given that some honest $(r, 3)$-verifiers have seen $>2/3$ majority for $v'$,
some (actually, more than half of)
honest $(r, 2)$-verifiers
have signed for $v'$ and propagated their messages.
By the construction of the protocol,
those honest $(r, 2)$-verifiers
must have received a valid message $m^{r, 1}_\ell$ from some player $\ell\in SV^{r, 1}$ with $v' = H(B^r_\ell)$,
thus Property (3) holds.
\end{newitemize}

\paragraph{$\bm{\BBA*}$.}
We again distinguish two cases.
\begin{newitemize}

\item[Case 1.] {\em All verifiers $i\in HSV^{r, 4}$ have $b_i=1$.}

This happens following Case 1 of GC.
As $|MSV^{r, 4}|<t_H$,
in this case no verifier in $SV^{r, 5}$ could collect or generate
$t_H$ valid $(r, 4)$-messages for bit $0$.
Thus, no honest verifier in $HSV^{r, 5}$ would stop because he knows a non-empty block $B^r$.

Moreover, although there are at least $t_H$ valid $(r, 4)$-messages for bit $1$,
$s'=5$ does not satisfy $s'-2 \equiv 1 \mod 3$,
thus no honest verifier in $HSV^{r, 5}$
would stop because he knows $B^r = B^r_\epsilon$.

Instead, every verifier $i\in HSV^{r, 5}$ acts at time $\beta^{r, 5}_i = \alpha^{r, 5}_i + t_5$,
by when he has received all messages sent by $HSV^{r, 4}$ following Lemma \ref{lem:basic}.
Thus player $i$ has seen $>2/3$ majority for $1$ and sets $b_i = 1$.

\smallskip
In Step 6 which is a Coin-Fixed-To-1 step,
although $s'=5$ satisfies $s'-2\equiv 0 \mod 3$,
there do not exist $t_H$ valid $(r, 4)$-messages for bit 0,
thus no verifier in $HSV^{r, 6}$ would stop because he knows a non-empty block $B^r$.
However, with $s'=6$, $s'-2 \equiv 1 \mod 3$ and
there do exist
$|HSV^{r, 5}|\geq t_H$
valid $(r, 5)$-messages for bit 1
from $HSV^{r, 5}$.

For every verifier $i\in HSV^{r, 6}$, following Lemma \ref{lem:basic},
on or before time $\alpha^{r, 6}_i + t_6$ player~$i$ has received all messages from $HSV^{r, 5}$, thus
$i$ stops without propagating anything and sets $B^r = B^r_\epsilon$.
His $CERT^r$ is the set of $t_H$ valid $(r, 5)$-messages $m^{r, 5}_j = (ESIG_j(1), ESIG_j(v_j), \sigma^{r, 5}_j)$
received by him when he stops.

\smallskip
Next, let player $i$ be either an honest verifier in a step $s>6$ or a generic honest user (i.e., non-verifier).
Similar to the proof of Lemma \ref{lem:honestleader_roundr},
player $i$ sets $B^r = B^r_\epsilon$ and sets his own $CERT^r$ to be the set of
$t_H$ valid $(r, 5)$-messages $m^{r, 5}_j = (ESIG_j(1), ESIG_j(v_j), \sigma^{r, 5}_j)$ he has received.

\smallskip
Finally, similar to Lemma \ref{lem:honestleader_roundr},
$$T^{r+1}\leq \min_{i\in HSV^{r, 6}} \alpha^{r, 6}_i+t_6 \leq T^r+\lambda+t_6 = T^r+10\lambda+\Lambda,$$
and all honest users know $B^r$ in the time interval $I^{r+1}$,
because the first honest user $i$ who knows $B^r$ has helped propagating the $(r, 5)$-messages in his $CERT^r$.

\item[Case 2.]
{\em There exists a verifier $\hat{i}\in HSV^{r, 4}$ with $b_{\hat{i}} = 0$.}

This happens following Case 2 of GC and
is the more complex case.
By the analysis of GC, in this case there exists a valid message
$m^{r, 1}_\ell$ such that $v_i = H(B^r_\ell)$ for all $i\in HSV^{r, 4}$.
Note that the verifiers in $HSV^{r, 4}$ may not have an agreement on their $b_i$'s.

For any step $s\in \{5, \dots, m+3\}$ and verifier $i\in HSV^{r, s}$,
by Lemma \ref{lem:basic} player $i$ would have received all messages sent by all honest verifiers in $HSV^{r, 4}\cup\cdots\cup HSV^{r, s-1}$ if he has waited for time~$t_s$.

\smallskip
We now consider the following event $E$:
{\em there exists a step $s^*\geq 5$ such that,
for the first time in the binary BA, some player
$i^*\in SV^{r, s^*}$ (whether malicious or honest)
should stop without propagating anything.}
We use ``should stop'' to emphasize the fact that, if player~$i^*$ is malicious,
then he may pretend that he should not stop according to the protocol and propagate messages of the Adversary's choice.

\smallskip
Moreover, by the construction of the protocol,  either
\begin{newitemize}
\item[($E.a$)]  $i^*$
is able to collect or generate at least $t_H$ valid messages
$m^{r, s'-1}_j = (ESIG_j(0), ESIG_j(v),$ $\sigma^{r, s'-1}_j)$ for the same $v$ and $s'$,
with $5\leq s'\leq s^*$ and $s'-2\equiv 0\mod 3$;
or

\item[($E.b$)] $i^*$ is able to collect or generate at least
$t_H$ valid messages $m^{r, s'-1}_j= (ESIG_j(1), ESIG_j(v_j),$ $\sigma^{r, s'-1}_j)$ for the same $s'$,
with
$6\leq s'\leq s^*$ and $s'-2\equiv 1 \mod 3$.
\end{newitemize}
Because the honest $(r, s'-1)$-messages are received by all honest $(r, s')$-verifiers before they are done waiting in Step $s'$,
and because the Adversary receives everything no later than the honest users,
without loss of generality we have $s' = s^*$ and player $i^*$ is malicious.
Note that we did not require the value $v$ in $E.a$ to be the hash of a valid block: as it will become clear in the analysis,
$v=H(B^r_\ell)$ in this sub-event.

\smallskip
Below we first analyze Case 2 following event $E$,
and then show that the value of $s^*$
is essentially distributed accordingly to $L^r$
(thus event $E$ happens before Step $m+3$ with overwhelming probability
given the relationships for parameters).
To begin with, for any step $5\leq s<s^*$,
every honest verifier $i\in HSV^{r, s}$ has waited  time $t_{s}$
and set $v_{i}$ to be the majority vote of the valid $(r, s-1)$-messages he has received.
Since player $i$ has received all honest $(r, s-1)$-messages following Lemma \ref{lem:basic},
since all honest verifiers in $HSV^{r, 4}$ have signed $H(B^r_\ell)$ following Case 2 of GC,
and since $|HSV^{r, s-1}|> 2|MSV^{r, s-1}|$ for each $s$,
by induction we have that player $i$ has set
$$v_{i} = H(B^r_\ell).$$
The same holds for every honest verifier $i\in HSV^{r, s^*}$ who does not stop without propagating anything.
Now we consider Step $s^*$ and distinguish four subcases.

\begin{newitemize}
\item[Case 2.1.a.]
{\em Event $E.a$ happens and
there exists an honest verifier $i'\in HSV^{r, s^*}$ who should also stop without propagating anything.}

In this case, we have $s^*-2\equiv 0 \mod 3$ and  Step $s^*$ is a Coin-Fixed-To-0 step.
By definition, player
$i'$ has received at least $t_H$ valid $(r, s^*-1)$-messages
of the form $(ESIG_j(0), ESIG_j(v), \sigma^{r, s^*-1}_j)$.
Since all verifiers in $HSV^{r, s^*-1}$ have signed $H(B^r_\ell)$
and $|MSV^{r, s^*-1}|< t_H$,
we have $v = H(B^r_\ell)$.

Since
at least
$t_H - |MSV^{r, s^*-1}|\geq 1$
of the $(r, s^*-1)$-messages received by $i'$ for $0$ and $v$
are sent by verifiers in $HSV^{r, s^*-1}$
after time $T^r+t_{s^*-1}\geq T^r+t_4 \geq T^r+\lambda + \Lambda \geq \beta^{r, 1}_\ell + \Lambda$,
player $i'$ has received $m^{r, 1}_\ell$ by the time he receives those $(r, s^*-1)$-messages.
Thus player $i'$ stops without propagating anything; sets $B^r = B^r_\ell$; and sets his own $CERT^r$ to be the set of valid $(r, s^*-1)$-messages for 0 and $v$ that he has received.

\medskip
Next, we show that, any other verifier $i\in HSV^{r, s^*}$ has either stopped with $B^r = B^r_\ell$,
or has set $b_{i} = 0$ and propagated $(ESIG_{i}(0), ESIG_{i}(H(B^r_\ell)), \sigma^{r, s}_{i})$.
Indeed, because Step $s^*$ is the first time some verifier should stop without propagating anything,
there does not exist a step $s'< s^*$ with $s'-2\equiv 1 \mod 3$
such that $t_H$ $(r, s'-1)$-verifiers have signed $1$.
Accordingly, no verifier in $HSV^{r, s^*}$ stops with $B^r = B^r_\epsilon$.

Moreover, as all honest verifiers in steps $\{4, 5, \dots, s^*-1\}$ have signed $H(B^r_\ell)$,
there does not exist a step $s'\leq s^*$ with $s'-2\equiv 0 \mod 3$
such that $t_H$ $(r, s'-1)$-verifiers have signed some $v''\neq H(B^r_\ell)$ ---indeed,
$|MSV^{r, s'-1}|<t_H$.
Accordingly, no verifier in $HSV^{r, s^*}$ stops with $B^r \neq B^r_\epsilon$ and $B^r\neq B^r_\ell$.
That is, if a player $i\in HSV^{r, s^*}$ has stopped without propagating anything, he must have set
$B^r = B^r_\ell$.

If a player $i\in HSV^{r, s^*}$ has waited time $t_{s^*}$ and propagated a message at time $\beta^{r, s^*}_i
= \alpha^{r, s^*}_i+t_{s^*}$,
he has received all messages from $HSV^{r, s^*-1}$,
including at least $t_H-|MSV^{r, s^*-1}|$ of them for
$0$ and $v$.
If $i$ has seen $>2/3$ majority for $1$,
then he has seen more than $2(t_H-|MSV^{r, s^*-1}|)$ valid $(r, s^*-1)$-messages
for $1$,
with more than
$2t_H - 3|MSV^{r, s^*-1}|$ of them from honest $(r, s^*-1)$-verifiers.
However, this implies $|HSV^{r, s^*-1}|\geq
t_H-|MSV^{r, s^*-1}| + 2t_H - 3|MSV^{r, s^*-1}|> 2n-4|MSV^{r, s^*-1}|$,
contradicting the fact that
$$|HSV^{r, s^*-1}| + 4|MSV^{r, s^*-1}|< 2n,$$
 which comes
from the relationships for the parameters.
Accordingly, $i$ does not see $>2/3$ majority for 1, and he sets $b_i = 0$
because Step $s^*$ is a Coin-Fixed-To-0 step.
As we have seen, $v_i = H(B^r_\ell)$.
Thus $i$ propagates
$(ESIG_{i}(0), ESIG_{i}(H(B^r_\ell)), \sigma^{r, s}_{i})$ as we wanted to show.

\medskip
For Step $s^*+1$, since player $i'$ has helped propagating the messages in his $CERT^r$ on or before time $\alpha^{r, s^*}_{i'}+t_{s^*}$, all honest verifiers in $HSV^{r, s^*+1}$
have received at least
$t_H$ valid $(r, s^*-1)$-messages
for bit $0$ and value $H(B^r_\ell)$ on or before they are done waiting.
Furthermore, verifiers in $HSV^{r, s^*+1}$ will not stop before receiving those $(r, s^*-1)$-messages,
because there do not exist any other $t_H$ valid
$(r, s'-1)$-messages for bit $1$
with
$s'-2\equiv 1\mod 3$ and $6\leq s'\leq s^*+1$, by the definition of Step $s^*$.
In particular, Step $s^*+1$ itself is a Coin-Fixed-To-1 step,
but no honest verifier in $HSV^{r, s^*}$ has propagated a message for 1,
and $|MSV^{r, s^*}|<t_H$.

Thus all honest verifiers in $HSV^{r, s^*+1}$ stop without propagating anything and set $B^r = B^r_\ell$:
as before, they have received $m^{r, 1}_\ell$ before they receive the desired $(r, s^*-1)$-messages.%
\footnote{If $\ell$ is malicious, he might
send out $m^{r, 1}_\ell$ late, hoping that some honest users/verifiers have not received
$m^{r, 1}_\ell$ yet when they receive the desired certificate for it.
However, since verifier $\hat{i}\in HSV^{r, 4}$ has set $b_{\hat{i}}=0$ and $v_{\hat{i}}=H(B^r_\ell)$,
as before we have that more than half of honest verifiers $i\in HSV^{r, 3}$ have set $v_i = H(B^r_\ell)$.
This further implies more than half of honest verifiers $i\in HSV^{r, 2}$ have set $v_i = H(B^r_\ell)$, and those $(r, 2)$-verifiers have all
received $m^{r, 1}_\ell$. As the Adversary cannot distinguish a verifier from a non-verifier,
he cannot target the propagation of $m^{r, 1}_\ell$ to $(r, 2)$-verifiers without having the non-verifiers seeing it.
In fact, with high probability, more than half (or a good constant fraction) of all honest users have seen $m^{r, 1}_\ell$
after waiting for $t_2$ from the beginning of their own round $r$. From here on,
the time $\lambda'$ needed for $m^{r, 1}_\ell$ to reach the remaining honest users
is much smaller than $\Lambda$, and for simplicity we do not write it out in the analysis.
If $4\lambda\geq \lambda'$ then the analysis goes through without any change: by the end of Step 4, all honest users would have received $m^{r, 1}_\ell$.
If the size of the block becomes enormous and $4\lambda<\lambda'$, then in Steps 3 and~4, the protocol could ask each verifier to wait for $\lambda'/2$ rather than $2\lambda$,
and the analysis continues to hold.}
The same can be said for all honest verifiers in future steps and all honest users in general.
In particular,
 they all know $B^r = B^r_\ell$ within the time interval $I^{r+1}$
and
$$T^{r+1}\leq \alpha^{r, s^*}_{i'}+t_{s^*} \leq T^r+\lambda + t_{s^*}.$$

\item[Case 2.1.b.]
{\em Event $E.b$ happens and there exists
an honest verifier $i'\in HSV^{r, s^*}$ who should also stop without propagating anything.}

In this case we have $s^*-2\equiv 1\mod 3$ and
Step $s^*$ is a Coin-Fixed-To-1 step.
The analysis is similar to Case 2.1.a and many details have been omitted.

As before,
player $i'$ must have received
at least $t_H$ valid $(r, s^*-1)$-messages
of the form $(ESIG_j(1), ESIG_j(v_j), \sigma^{r, s^*-1}_j)$.
Again by the definition of $s^*$, there does not exist
a step $5\leq s'<s^*$ with $s'-2\equiv 0\mod 3$, where
at least $t_H$ $(r, s'-1)$-verifiers have signed 0 and the same $v$.
Thus player $i'$ stops without propagating anything;
sets $B^r = B^r_\epsilon$; and sets his own $CERT^r$ to be the set of valid
$(r, s^*-1)$-messages for bit 1 that he has received.

\smallskip
Moreover,
any other verifier $i\in HSV^{r, s^*}$ has either stopped with $B^r = B^r_\epsilon$, or has set $b_{i} = 1$ and propagated
$(ESIG_{i}(1), ESIG_{i}(v_{i}), \sigma^{r, s^*}_{i})$.
Since player $i'$ has helped propagating the
$(r, s^*-1)$-messages in his $CERT^r$ by time $\alpha^{r, s^*}_{i'}+t_{s^*}$,
again all honest verifiers in $HSV^{r, s^*+1}$ stop without propagating anything and set $B^r = B^r_\epsilon$.
Similarly, all honest users know $B^r = B^r_\epsilon$ within the time interval $I^{r+1}$
and
$$T^{r+1}\leq \alpha^{r, s^*}_{i'}+t_{s^*} \leq T^r+\lambda + t_{s^*}.$$

\item[Case 2.2.a.]
{\em Event $E.a$ happens and there does not exist an honest verifier $i'\in HSV^{r, s^*}$ who
should also stop without propagating anything.}

In this case, note that
player $i^*$ could have a valid $CERT^r_{i^*}$ consisting of the $t_H$ desired $(r, s^*-1)$-messages
the Adversary is able to collect or generate.
However,
the malicious verifiers may not help propagating those messages,
so we cannot conclude that the honest users will receive them in time $\lambda$.
In fact, $|MSV^{r, s^*-1}|$ of those messages may be from malicious $(r, s^*-1)$-verifiers,
who did not propagate their messages at all and only send them to the malicious verifiers in step $s^*$.

\medskip
Similar to Case 2.1.a, here we have
$s^*-2\equiv 0\mod 3$, Step $s^*$ is a Coin-Fixed-To-0 step,
and
the $(r, s^*-1)$-messages in $CERT^r_{i^*}$
are for bit $0$
and $v=H(B^r_\ell)$. Indeed, all honest $(r, s^*-1)$-verifiers sign $v$,
thus the Adversary cannot generate $t_H$ valid $(r, s^*-1)$-messages
for a different $v'$.

Moreover,
all honest $(r, s^*)$-verifiers have waited time $t_{s^*}$
and do not see $>2/3$ majority
for bit $1$,
again because
$|HSV^{r, s^*-1}| + 4|MSV^{r, s^*-1}|< 2n$.
Thus every honest verifier $i\in HSV^{r, s^*}$ sets $b_{i} = 0$, $v_{i} = H(B^r_\ell)$ by the majority vote,
and propagates $m^{r, s^*}_{i} = (ESIG_{i}(0), ESIG_{i}(H(B^r_\ell)), \sigma^{r, s^*}_{i})$ at time $\alpha^{r, s^*}_{i}+t_{s^*}$.

\medskip
Now consider the honest verifiers in Step $s^*+1$ (which is a Coin-Fixed-To-1 step).
If the Adversary actually sends the messages in $CERT^r_{i^*}$
 to some of them and causes them to stop,
then similar to Case 2.1.a, all honest users know $B^r = B^r_\ell$ within the time interval $I^{r+1}$
and
$$T^{r+1} \leq T^r+\lambda + t_{s^*+1}.$$

Otherwise, all honest verifiers in Step $s^*+1$ have
received all the $(r, s^*)$-messages
for $0$ and
$H(B^r_\ell)$ from $HSV^{r, s^*}$
after waiting time $t_{s^*+1}$,
which leads to $>2/3$ majority, because $|HSV^{r, s^*}|>2|MSV^{r, s^*}|$.
Thus all the verifiers in $HSV^{r, s^*+1}$
propagate their messages for $0$ and $H(B^r_\ell)$ accordingly.
Note that the verifiers in $HSV^{r, s^*+1}$ do not stop with
$B^r = B^r_\ell$,
because Step $s^*+1$
is not a Coin-Fixed-To-0 step.

\medskip
Now consider the honest verifiers in Step $s^*+2$ (which is
a Coin-Genuinely-Flipped step).
If the Adversary sends the messages in $CERT^r_{i^*}$ to some of them and causes them to stop,
then again all honest users know $B^r = B^r_\ell$ within the time interval $I^{r+1}$ and
$$T^{r+1}\leq T^r+\lambda + t_{s^*+2}.$$

Otherwise, all honest verifiers in Step $s^*+2$ have received
all the $(r, s^*+1)$-messages for $0$ and $H(B^r_\ell)$
from $HSV^{r, s^*+1}$
after waiting time $t_{s^*+2}$,
which leads to $>2/3$ majority.
Thus all of them propagate
 their messages for $0$ and $H(B^r_\ell)$ accordingly:
that is they do not ``flip a coin'' in this case.
Again, note that they do not stop without propagating,
because Step $s^*+2$ is not a Coin-Fixed-To-0 step.

\medskip
Finally, for the honest verifiers in Step $s^*+3$ (which is another Coin-Fixed-To-0 step),
all of them would have received
at least $t_H$ valid messages for $0$ and $H(B^r_\ell)$ from $HSV^{s^*+2}$,
if they really wait time $t_{s^*+3}$.
Thus, whether or not the Adversary sends the messages in $CERT^r_{i^*}$ to
any of them,
all verifiers in $HSV^{r, s^*+3}$ stop with $B^r = B^r_\ell$, without propagating anything.
Depending on how the Adversary acts, some of them may have their own $CERT^r$ consisting of those $(r, s^*-1)$-messages in $CERT^r_{i^*}$,
and the others have their own $CERT^r$ consisting of those $(r, s^*+2)$-messages.
In any case, all honest users know $B^r = B^r_\ell$ within the time interval $I^{r+1}$ and
$$T^{r+1}\leq T^r+\lambda + t_{s^*+3}.$$

\item[Case 2.2.b.]
{\em Event $E.b$ happens and there does not exist an honest verifier $i'\in HSV^{r, s^*}$ who
should also stop without propagating anything.}

The analysis in this case is similar to those in Case 2.1.b and Case 2.2.a, thus many details
have been omitted.
In particular, $CERT^r_{i^*}$ consists of the $t_H$ desired
$(r, s^*-1)$-messages for bit 1
that the Adversary is able to collect or generate,
$s^*-2\equiv 1 \mod 3$,
Step $s^*$ is a Coin-Fixed-To-1 step,
 and no honest $(r, s^*)$-verifier could have
 seen $>2/3$ majority for $0$.

\medskip
Thus, every verifier $i\in HSV^{r, s^*}$ sets $b_{i}=1$ and
propagates $m^{r, s^*}_{i} = (ESIG_{i}(1), ESIG_{i}(v_{i}),$ $\sigma^{r, s^*}_{i})$ at time
$\alpha^{r, s^*}_{i}+t_{s^*}$.
Similar to Case 2.2.a,
in at most 3 more steps
(i.e., the protocol reaches Step $s^*+3$, which is another Coin-Fixed-To-1 step),
all honest users know $B^r = B^r_\epsilon$ within the time interval $I^{r+1}$.
Moreover, $T^{r+1}$ may be $\leq T^r+\lambda + t_{s^*+1}$,
or $\leq T^r+\lambda + t_{s^*+2}$,
or $\leq T^r+\lambda + t_{s^*+3}$,
depending on when is the first time an honest verifier is able to stop without propagating.

\end{newitemize}

Combining the four sub-cases,
we have that
all honest users know $B^r$ within the time interval
$I^{r+1}$, with
\begin{newitemize}
\item[]
$T^{r+1}\leq T^r+\lambda + t_{s^*}$
in Cases 2.1.a and 2.1.b,
and

\item[]
$T^{r+1}\leq T^r+\lambda + t_{s^*+3}$
in Cases 2.2.a and 2.2.b.
\end{newitemize}
It remains to upper-bound $s^*$ and thus $T^{r+1}$
for Case 2,
and we do so by considering how many times the
Coin-Genuinely-Flipped steps
are actually executed in the protocol: that is,
some honest verifiers actually have flipped a coin.

\smallskip
In particular, arbitrarily fix a Coin-Genuinely-Flipped
step $s'$ (i.e., $7\leq s'\leq m+2$ and $s'-2 \equiv 2 \mod 3$),
and let
$\ell' \triangleq \argmin_{j\in SV^{r, s'-1}} H(\sigma^{r, s'-1}_j)$.
For now let us assume $s'<s^*$, because otherwise
no honest verifier actually flips a coin in Step $s'$,
according to previous discussions.

\smallskip
By the definition of $SV^{r, s'-1}$,
the hash value of the credential of $\ell'$
is also the smallest among all users in $PK^{r-k}$.
Since the hash function is a random oracle,
ideally player $\ell'$ is honest with
probability at least $h$.
As we will show later, even if the Adversary tries his best to predict the output of the random oracle and tilt the probability,
player $\ell'$ is still honest with probability at least $p_h = h^2(1+h-h^2)$.
Below we consider the case when that indeed happens:
that is, $\ell'\in HSV^{r, s'-1}$.

\smallskip
Note that every honest verifier $i\in HSV^{r, s'}$ has received all
messages from $HSV^{r, s'-1}$ by time $\alpha^{r, s'}_i+t_{s'}$.
If player $i$ needs to flip a coin
(i.e., he has not seen $>2/3$ majority
for the same bit $b\in \{0, 1\}$),
then he sets $b_i = \lsb(H(\sigma^{r, s'-1}_{\ell'}))$.
If there exists another honest verifier $i'\in HSV^{r, s'}$
who
has seen $>2/3$ majority for
a bit $b\in \{0, 1\}$,
then by Property (d) of Lemma~\ref{lem:basic},
no honest verifier in $HSV^{r, s'}$ would have
seen $>2/3$ majority for a
bit $b'\neq b$.
Since $\lsb(H(\sigma^{r, s'-1}_{\ell'})) = b$ with probability $1/2$,
all honest verifiers in $HSV^{r, s'}$ reach an agreement on $b$ with probability $1/2$.
Of course, if
such a verifier $i'$ does not exist,
then all honest verifiers in $HSV^{r, s'}$
agree on the bit $\lsb(H(\sigma^{r, s'-1}_{\ell'}))$ with probability 1.

\smallskip
Combining the probability for $\ell'\in HSV^{r, s'-1}$,
we have that the honest verifiers in $HSV^{r, s'}$
reach an agreement on a bit $b\in \{0,1\}$
with probability at least $\frac{p_h}{2} = \frac{h^2(1+h-h^2)}{2}$.
Moreover, by induction on the majority vote as before, all honest verifiers in $HSV^{r, s'}$
have their $v_i$'s set to be $H(B^r_\ell)$.
Thus, once an agreement on $b$ is reached in Step $s'$,
$T^{r+1}$ is
$$\mbox{either } \leq T^r+\lambda +t_{s'+1} \mbox{ or } \leq T^r+\lambda + t_{s'+2},$$
depending on whether $b=0$ or $b=1$,
following the analysis of Cases 2.1.a and 2.1.b. In particular, no further Coin-Genuinely-Flipped step will be executed:
that is, the verifiers in such steps still check that they are
the verifiers and thus wait, but they will all stop without propagating anything.
Accordingly, before Step $s^*$, the number of times the Coin-Genuinely-Flipped steps
are executed
is distributed according to the random variable $L^r$.
Letting Step $s'$ be the last Coin-Genuinely-Flipped step according to $L^r$,
by the construction of the protocol we have
$$s' = 4+3L^r.$$

\smallskip
When should the Adversary make Step $s^*$ happen if he wants to delay $T^{r+1}$ as much
as possible? We can even assume that the Adversary knows the realization of $L^r$ in advance.
If $s^*>s'$ then it is useless, because
the honest verifiers have already reached an agreement in Step $s'$.
To be sure, in this case $s^*$ would be $s'+1$ or $s'+2$, again depending on whether
$b=0$ or $b=1$.
However, this is actually Cases 2.1.a and 2.1.b,
and the resulting $T^{r+1}$ is exactly the same as in that case.
More precisely,
$$T^{r+1}\leq T^r+\lambda+t_{s^*} \leq T^r+\lambda+t_{s'+2}.$$

If $s^*<s'-3$ ---that is, $s^*$ is before the second-last Coin-Genuinely-Flipped step--- then
by the analysis of Cases 2.2.a and 2.2.b,
$$T^{r+1}\leq T^r+\lambda + t_{s^*+3}<T^r+\lambda+t_{s'}.$$
That is,
the Adversary is actually
making the agreement on $B^r$ happen faster.

\smallskip
If $s^* = s'-2$ or $s'-1$ ---that is,
the Coin-Fixed-To-0 step or the Coin-Fixed-To-1 step immediately before Step $s'$---
then
by the analysis of the four sub-cases,
the honest verifiers in Step $s'$ do not get to flip coins anymore,
 because they have either stopped without propagating,
or have seen $>2/3$ majority
for the same bit $b$.
Therefore we have
$$T^{r+1}\leq T^r+\lambda+t_{s^*+3}\leq T^r+\lambda+t_{s'+2}.$$

In sum, no matter what $s^*$ is, we have
\begin{eqnarray*}
& & T^{r+1}\leq T^r+\lambda+t_{s'+2}= T^r+\lambda + t_{3L^r+6} \\
&=&
T^r + \lambda + (2(3L^r+6)-3)\lambda + \Lambda \\
&=& T^r +(6L^r+10)\lambda + \Lambda,
\end{eqnarray*}
as we wanted to show. The worst case is when $s^*=s'-1$ and Case 2.2.b happens.
\end{newitemize}

Combining Cases 1 and 2 of the binary BA protocol, Lemma \ref{lem:leadermalicious_roundr} holds.
\end{proof}

\subsection{Security of the Seed $\bm{Q^r}$ and Probability of An Honest Leader}\label{subsec:security}
It remains to prove Lemma \ref{lem:security}.
Recall that the verifiers in round $r$
are taken from $PK^{r-k}$ and are chosen according to the quantity $Q^{r-1}$.
The reason for introducing the look-back parameter $k$ is to make sure that, back at round $r-k$, when the Adversary is able to add new malicious users to $PK^{r-k}$,
he cannot predict the quantity $Q^{r-1}$ except with negligible probability.
Note that the hash function is a random oracle and $Q^{r-1}$ is one of its inputs when selecting verifiers for
round~$r$.
Thus, no matter how malicious users are added to $PK^{r-k}$, from the Adversary's point of view each one of them is still
selected to be a verifier in a step of round $r$ with the required probability $p$ (or $p_1$ for Step 1).
More precisely, we have the following lemma.

\begin{lemma}\label{lem:Q}
With $k = O(\log_{1/2} F)$, for each round $r$, with overwhelming probability the Adversary did not query $Q^{r-1}$
to the random oracle back at round $r-k$.
\end{lemma}

\begin{proof}

We proceed by induction.
Assume that
for each round $\gamma<r$,
the Adversary did not query $Q^{\gamma-1}$ to the random oracle back at round $\gamma-k$.%
\footnote{As $k$ is a small integer, without loss of generality one can assume that
the first $k$ rounds of the protocol are run under a safe environment and the inductive hypothesis holds
for those rounds.}
Consider the following mental game played by the Adversary at round $r-k$, trying to predict $Q^{r-1}$.

In Step 1 of each round $\gamma  = r-k,\dots, r-1$,
given a specific $Q^{\gamma-1}$ not queried to the random oracle,
by ordering the players
$i\in PK^{\gamma-k}$
according to the hash values $H(SIG_i(\gamma, 1, Q^{\gamma-1}))$ increasingly,
we  obtain a random permutation over $PK^{\gamma-k}$.
By definition, the leader $\ell^{\gamma}$ is the first user in the permutation and
is honest with probability $h$.
Moreover, when $PK^{\gamma-k}$ is large enough,
for any integer $x\geq 1$, the probability that the first $x$ users in the permutation
are all malicious but the $(x+1)$st is honest
is $(1-h)^{x}h$.

If $\ell^{\gamma}$ is honest, then $Q^{\gamma} = H(SIG_{\ell^{\gamma}}(Q^{\gamma-1}), \gamma)$.
As the Adversary cannot forge the signature of~$\ell^{\gamma}$,
$Q^{\gamma}$ is distributed uniformly at random
from the Adversary's point of view and,
except with exponentially small probability,%
\footnote{That is, exponential in the length of the output of $H$. Note that this probability is way smaller than $F$.}
was not queried to $H$ at round $r-k$.
Since
each $Q^{\gamma+1}, Q^{\gamma+2}, \dots, Q^{r-1}$
respectively is the output of $H$ with $Q^{\gamma}, Q^{\gamma+1}, \dots, Q^{r-2}$ as one of the inputs,
they all look random to the Adversary and
the Adversary could not have queried $Q^{r-1}$ to $H$ at round $r-k$.

Accordingly, the only case where the Adversary can predict $Q^{r-1}$
with good probability at round $r-k$ is when all the leaders
$\ell^{r-k}, \dots, \ell^{r-1}$ are malicious.
Again consider a round $\gamma \in \{r-k\dots, r-1\}$
and the random permutation over $PK^{\gamma-k}$ induced by the corresponding hash values.
If for some $x\geq 2$, the first $x-1$ users in the permutation
are all malicious and the $x$-th is honest,
then the Adversary has $x$ possible choices for $Q^{\gamma}$:
either of the form
$H(SIG_i(Q^{\gamma-1}, \gamma))$, where $i$ is one of
the first $x-1$ malicious users, by making player $i$
the actually leader of round $\gamma$;
or $H(Q^{\gamma-1}, \gamma)$, by forcing $B^{\gamma} = B^{\gamma}_\epsilon$.
Otherwise, the leader of round $\gamma$ will be the first honest user in the permutation and
$Q^{r-1}$ becomes unpredictable to the Adversary.

Which of the above $x$ options of $Q^{\gamma}$ should the Adversary pursue?
To help the Adversary answer this question, in the mental game we actually make him more powerful than he actually is, as follows.
First of all, in reality, the Adversary cannot
compute the hash of a honest user's signature, thus cannot
decide, for each $Q^\gamma$,
the number $x(Q^\gamma)$ of malicious users at the beginning of the random permutation in round $\gamma+1$ induced by $Q^\gamma$.
In the mental game,
we give him the numbers $x(Q^\gamma)$ for free.
Second of all, in reality,
having the first $x$ users in the permutation all being malicious
does not necessarily mean they can all be made into the leader,
because the hash values of their signatures must also be less than $p_1$.
We have ignored this constraint in the mental game,
giving the Adversary even more advantages.

It is easy to see that in the mental game,
the optimal option for the Adversary, denoted by $\hat{Q}^{\gamma}$, is the one that
produces the longest sequence of malicious users
at the beginning of the random permutation in round $\gamma+1$.
Indeed, given a specific $Q^{\gamma}$, the protocol does not depend on $Q^{\gamma-1}$ anymore
and the Adversary can solely focus on the new permutation in round $\gamma+1$,
which has the same distribution for the number of malicious users at the beginning.
Accordingly, in each round $\gamma$,
 the above mentioned $\hat{Q}^{\gamma}$
 gives him the largest number of options for $Q^{\gamma+1}$ and thus
 maximizes the probability that
the consecutive leaders are all malicious.

Therefore, in the mental game
the Adversary is following a Markov Chain from round $r-k$ to round $r-1$,
with the state space being $\{0\}\cup\{x: x\geq 2\}$.
State $0$ represents the fact that the first user in the random permutation in the current round $\gamma$
is honest, thus the Adversary fails the game
for predicting $Q^{r-1}$;
 and
each state $x\geq 2$ represents the fact that the first $x-1$ users in the permutation are malicious
and the $x$-th is honest,
thus the Adversary has $x$ options for $Q^{\gamma}$.
The transition probabilities $P(x, y)$ are as follows.
\begin{newitemize}
\item
$P(0, 0) = 1$ and $P(0, y) = 0$ for any $y\geq 2$. That is, the Adversary fails the game once the first user in the permutation becomes honest.

\item
$P(x, 0) = h^x$ for any $x\geq 2$. That is, with probability $h^x$, all the $x$ random permutations have their first users being
honest, thus the Adversary fails the game in the next round.

\item
For any $x\geq 2$ and $y\geq 2$, $P(x, y)$ is the probability that,
among the $x$ random permutations induced by the $x$ options of $Q^{\gamma}$,
the longest sequence of malicious users at the beginning of some of them is $y-1$,
thus the Adversary has $y$ options
for $Q^{\gamma+1}$ in the next round.
That is,
$$P(x, y) = \left(\sum_{i=0}^{y-1}(1-h)^i h\right)^x - \left(\sum_{i=0}^{y-2}(1-h)^i h\right)^x=(1-(1-h)^y)^x - (1-(1-h)^{y-1})^x.$$

\end{newitemize}
Note that state 0 is
the unique absorbing state in the transition matrix $P$,
 and every other state $x$ has a positive probability of going to 0.
We are interested in upper-bounding the number $k$ of rounds needed for the
Markov Chain to converge to 0 with overwhelming probability: that is,
no matter which state the chain starts at,
with overwhelming probability the Adversary loses the game and fails to predict $Q^{r-1}$ at round $r-k$.

Consider the transition matrix $P^{(2)} \triangleq P\cdot P$ after two rounds.
It is easy to see that $P^{(2)}(0, 0) = 1$ and $P^{(2)}(0, x) = 0$ for any $x\geq 2$.
For any $x\geq 2$ and $y\geq 2$, as $P(0, y) = 0$, we have
$$P^{(2)}(x, y) = P(x, 0)P(0, y) + \sum_{z\geq 2}P(x, z)P(z, y)= \sum_{z\geq 2}P(x, z)P(z, y).$$
Letting $\bar{h} \triangleq 1-h$, we have
$$P(x, y) = (1-\bar{h}^y)^x - (1-\bar{h}^{y-1})^x$$
 and
$$P^{(2)}(x, y) = \sum_{z\geq 2}[(1-\bar{h}^z)^x - (1-\bar{h}^{z-1})^x][(1-\bar{h}^y)^z - (1-\bar{h}^{y-1})^z].$$
Below we compute the limit of
$\frac{P^{(2)}(x, y)}{P(x, y)}$ as $h$ goes to $1$ ---that is, $\bar{h}$ goes to 0.
Note that the highest order of $\bar{h}$ in $P(x, y)$ is $\bar{h}^{y-1}$, with coefficient $x$.
Accordingly,
\begin{eqnarray*}
&&\lim_{h\rightarrow 1} \frac{P^{(2)}(x, y)}{P(x, y)}
= \lim_{\bar{h}\rightarrow 0} \frac{P^{(2)}(x, y)}{P(x, y)}
= \lim_{\bar{h}\rightarrow 0} \frac{P^{(2)}(x, y)}{x \bar{h}^{y-1} + O(\bar{h}^{y})} \\
&=& \lim_{\bar{h}\rightarrow 0} \frac{\sum_{z\geq 2} [x \bar{h}^{z-1} + O(\bar{h}^{z})][z \bar{h}^{y-1} + O(\bar{h}^{y})]}{x \bar{h}^{y-1} + O(\bar{h}^{y})}
= \lim_{\bar{h}\rightarrow 0} \frac{2x\bar{h}^{y} + O(\bar{h}^{y+1})}{x \bar{h}^{y-1} + O(\bar{h}^{y})} \\
&=&  \lim_{\bar{h}\rightarrow 0} \frac{2x\bar{h}^y}{x\bar{h}^{y-1}}
= \lim_{\bar{h}\rightarrow 0} 2\bar{h} = 0.
\end{eqnarray*}
When $h$ is sufficiently close to $1$,%
\footnote{For example, $h = 80\%$ as suggested by the specific choices of parameters.}
we have
$$\frac{P^{(2)}(x, y)}{P(x, y)}\leq \frac{1}{2}$$
 for any $x\geq 2$ and $y\geq 2$.
By induction, for any $k> 2$,
$P^{(k)} \triangleq P^k$ is such that
\begin{newitemize}
\item
$P^{(k)}(0, 0) = 1$, $P^{(k)}(0, x) = 0$ for any $x\geq 2$, and

\item
for any $x\geq 2$ and $y\geq 2$,
\begin{eqnarray*}
&&P^{(k)}(x, y) = P^{(k-1)}(x, 0) P(0, y) + \sum_{z\geq 2} P^{(k-1)}(x, z) P(z, y)
= \sum_{z\geq 2} P^{(k-1)}(x, z) P(z, y) \\
&\leq& \sum_{z\geq 2} \frac{P(x, z)}{2^{k-2}} \cdot P(z, y)
= \frac{P^{(2)}(x, y)}{2^{k-2}} \leq \frac{P(x, y)}{2^{k-1}}.
\end{eqnarray*}
\end{newitemize}
As $P(x, y)\leq 1$, after $1-\log_2 F$ rounds,
the transition probability into any state $y\geq 2$ is negligible, starting with any state $x\geq 2$.
Although there are many such states $y$, it is easy to see that
$$\lim_{y\rightarrow +\infty} \frac{P(x, y)}{P(x, y+1)} =
\lim_{y\rightarrow +\infty} \frac{(1-\bar{h}^y)^x - (1-\bar{h}^{y-1})^x}{(1-\bar{h}^{y+1})^x - (1-\bar{h}^{y})^x}
=\lim_{y\rightarrow +\infty} \frac{\bar{h}^{y-1}-\bar{h}^y}{\bar{h}^y-\bar{h}^{y+1}} = \frac{1}{\bar{h}} = \frac{1}{1-h}.$$
Therefore each row $x$ of the transition matrix $P$ decreases as a geometric sequence with rate $\frac{1}{1-h}>2$
when $y$ is large enough,
and the same holds for $P^{(k)}$.
Accordingly, when $k$ is large enough but still on the order of
$\log_{1/2} F$,
$\sum_{y\geq 2}P^{(k)}(x, y) < F$ for any $x\geq 2$.
That is, with overwhelming probability the Adversary loses the game and fails to predict $Q^{r-1}$ at round $r-k$.
For $h\in (2/3, 1]$, a more complex analysis shows that there exists
a constant $C$ slightly larger than $1/2$,
 such that it suffices to take $k = O(\log_C F)$.
Thus Lemma \ref{lem:Q} holds.
\end{proof}

\medskip
\noindent
{\bf Lemma \ref{lem:security}.} (restated)
{\em
Given Properties 1--3 for each round before $r$,
$p_h = h^2(1+h-h^2)$ for $L^r$, and
the leader $\ell^r$ is honest with probability at least $p_h$.
}

\begin{proof}
Following Lemma \ref{lem:Q}, the Adversary cannot predict
$Q^{r-1}$ back at round $r-k$ except with negligible probability.
Note that this does not mean
the probability of an honest leader is $h$ for each round.
Indeed, given $Q^{r-1}$,
 depending on how many malicious users are
at the beginning of the random permutation of $PK^{r-k}$,
the Adversary may have more than one options for $Q^r$ and thus
can increase the probability of a malicious leader in round $r+1$ ---again we
 are giving him some unrealistic advantages as in Lemma \ref{lem:Q}, so as
 to simplify the analysis.

However, for each $Q^{r-1}$ that was not queried to $H$ by the Adversary
back at round $r-k$, for any $x\geq 1$,
with probability $(1-h)^{x-1} h$
the first honest user occurs at position $x$
in the resulting random permutation of $PK^{r-k}$.
When $x = 1$, the probability of an honest leader in round $r+1$ is indeed $h$;
while when $x = 2$, the Adversary has two options for $Q^r$ and
the resulting probability is $h^2$.
Only by considering these two cases,
we have that the probability of an honest leader in round $r+1$ is at least
$h\cdot h + (1-h)h \cdot h^2 = h^2(1+h-h^2)$ as desired.

Note that the above probability only considers the randomness in the protocol
from round $r-k$ to round $r$.
When all the randomness from round 0 to round $r$ is taken into consideration,
$Q^{r-1}$ is even less predictable to the Adversary and  the probability of
an honest leader in round $r+1$ is at least $h^2(1+h-h^2)$. Replacing $r+1$ with $r$ and shifts
everything back by one round,
the leader $\ell^r$ is honest with probability at least $h^2(1+h-h^2)$, as desired.

Similarly, in each Coin-Genuinely-Flipped step $s$, the ``leader'' of that step ---that is the verifier in $SV^{r,s}$ whose credential has the smallest hash value, is honest with probability
at least $h^2(1+h-h^2)$. Thus $p_h = h^2(1+h-h^2)$ for $L^r$ and Lemma \ref{lem:security} holds.
\end{proof}

\section{$\bm{\alg_2}$}

In this section, we construct a version of  $\alg$ working under the following assumption.
\begin{description}
  \item[{\sc Honest Majority of Users  Assumption:}] {\em More than 2/3 of the users in each $PK\upr$ are honest.}

\end{description}
In Section \ref{sec:hmm}, we show how to replace the above assumption with the desired Honest Majority of Money assumption.

\subsection{Additional Notations  and Parameters for $\bm{\alg_2}$}\label{sec:para}

\paragraph{Notations}

\begin{newitemize}

\item
$\mu\in \bZ^+$: a pragmatic upper-bound to the
number of steps that, with overwhelming probability,  will actually taken in one round. (As we shall see, parameter $\mu$ controls how many ephemeral keys a user prepares in advance for each round.)

\item
$L^r$:
a random variable
representing the number of Bernoulli trials needed to see a 1,
when each trial is 1 with probability $\frac{p_h}{2}$.
$L^r$ will be used to upper-bound the time needed to generate block $B^r$.

\item
$t_H$: a lower-bound for the number of honest verifiers in a step $s>1$ of round $r$, such that
with overwhelming probability (given $n$ and $p$), there are $>t_H$ honest verifiers in
$SV^{r, s}$.

\end{newitemize}

\paragraph{Parameters}

\begin{newitemize}

\item
{\em Relationships among various parameters.}

\begin{newitemize}

\item[---] For each step $s>1$ of round $r$,
$n$
is chosen so that, with overwhelming probability,
\begin{newcenter}
$|HSV^{r, s}|>t_H$ \quad and \quad $|HSV^{r, s}|+2|MSV^{r, s}|<2t_H$.
\end{newcenter}
Note that the two inequalities above together imply $|HSV^{r, s}|>2|MSV^{r, s}|$: that is, there
is a $2/3$ honest majority among selected verifiers.

The closer to 1 the value of $h$ is, the smaller $n$ needs to be.
In particular, we use (variants of) Chernoff bounds to ensure the desired conditions hold with overwhelming probability.

\end{newitemize}

\item
{\em Example choices of important parameters.}
\begin{newitemize}

\item[---]
$F=10^{-18}$.

\item[---]
$n\approx 4000$, $t_H\approx 0.69n$,
 $k=70$.

  \end{newitemize}

\end{newitemize}

\subsection{Implementing Ephemeral Keys in $\alg_2$}\label{sec:ephemeral}

Recall  that a verifier $i\in SV\rs$
digitally signs his message $m_i\rs$ of step $s$ in round $r$,
relative to an ephemeral public key $pk_i\rs$, using an ephemeral secrete key $sk^{r, s}_i$ that
 he promptly destroys after using.
When the number of possible steps that a round may take is capped by a given integer $\mu$, we have already seen how
to practically handle ephemeral keys. For example, as we have explained in $\alg_1$ (where $\mu=m+3$),  to handle all his possible ephemeral keys, from a round $r'$ to a round $r'+10^6$,
$i$ generates a pair $(PMK,SMK)$, where $PMK$ public master key of an identity based signature scheme, and $SMK$ its corresponding secret master key. User $i$ publicizes $PMK$ and uses $SMK$ to generate the secret key of each possible ephemeral public key (and destroys $SMK$ after having done so).
The set of $i$'s ephemeral public keys for the relevant rounds is
$S=\{i\}\times \{r',\ldots,r'+10^6\}\times \{1, \ldots, \mu\}$.
(As discussed, as the round $r'+10^6$ approaches, $i$ ``refreshes" his pair $(PMK,SMK)$.)

In practice, if $\mu$ is large enough, a round of $\alg_2$
will not take more than $\mu$ steps. In principle, however, there is the remote possibility that, for some round $r$ the number of steps actually taken will exceed $\mu$. When this happens, $i$ would be unable to sign his message $m_i\rs$ for any step $s>\mu$, because he has prepared in advance only $\mu$ secret keys for round $r$.
Moreover, he could not prepare and publicize a new stash of ephemeral keys, as discussed before. In fact, to do so, he would need to insert a new public master key $PMK'$ in a new block. But, should round $r$ take more and more steps, no new blocks would be generated.

However,  solutions exist. For instance, $i$ may use the last ephemeral key of round $r$, $pk_i^{r,\mu}$, as follows.
He generates another stash of key-pairs for round $r$ ---e.g., by (1) generating another master key pair $(\overline{PMK}, \overline{SMK})$; (2) using this pair to generate another, say, $10^6$ ephemeral keys, $\overline{sk}_i^{r, \mu+1},\ldots,\overline{sk}_i^{r, \mu+10^6}$, corresponding to steps $\mu+1,...,\mu+10^6$ of round $r$;  (3)  using $sk_i^{r,\mu}$ to digitally sign $\overline{PMK}$ (and any $(r,\mu)$-message if $i\in SV^{r,\mu}$), relative to $pk_i^{r,\mu}$; and (4) erasing $\overline{SMK}$ and $sk_i^{r,\mu}$.
Should $i$ become a verifier in a step $\mu+s$ with $s\in \{1,\dots, 10^6\}$,
then $i$ digitally signs his $(r, \mu+s)$-message $m^{r, \mu+s}_i$ relative to his new key $\overline{pk}_i^{r, \mu+s} = (i, r, \mu+s)$.
Of course, to verify this signature of~$i$, others need to be certain that this public key
corresponds to
$i$'s new public master key $\overline{PMK}$.
Thus, in addition to this signature, $i$ transmits
his digital signature  of $\overline{PMK}$ relative to $pk_i^{r, \mu}$.

Of course, this approach can be repeated, as many times as necessary, should round $r$ continue for more and more steps!
The last ephemeral secret key is used to authenticate a new master public key, and thus another stash of ephemeral keys for round $r$. And so on.

\subsection{The Actual Protocol $\bm{\alg_2}$}

Recall again that, in each step $s$ of a round $r$, a verifier $i\in SV^{r, s}$
uses his long-term public-secret key pair to produce his credential, $\sigma_i^{r,s}\triangleq SIG_i(r, s, Q^{r-1})$,
as well as $SIG_i\left(Q^{r-1}\right)$ in case $s=1$. Verifier $i$ uses his ephemeral key pair, ($pk_i^{r,s},sk_i^{r,s})$,
to sign any other message $m$ that may be required.
For simplicity, we write $esig_i(m)$, rather than $sig_{pk^{r, s}_i}(m)$, to denote $i$'s proper ephemeral signature of $m$ in this step, and write $ESIG_i(m)$ instead of $SIG_{pk_i^{r,s}}(m) \triangleq (i, m, esig_i(m))$.

\begin{center}
\framebox[\textwidth]{
\centering
\begin{minipage}{.98\textwidth}
\begin{center}
Step 1: Block Proposal

\end{center}

Instructions for every user $i\in PK^{r-k}$: User $i$ starts his own Step 1 of round $r$ as soon as he has $CERT^{r-1}$, which
allows $i$ to unambiguously compute $H(B^{r-1})$ and $Q^{r-1}$.

\begin{itemize}

\item
User $i$ uses $Q^{r-1}$ to check
whether $i\in SV^{r,1}$ or not. If $i\notin SV^{r,1}$,
he does nothing for Step 1.

\item
If $i\in SV^{r, 1}$, that is, if $i$ is a potential leader, then he does the following.

 \begin{itemize}
\item[(a)]
If $i$ has seen $B^0,\ldots,B^{r-1}$ himself (any $B^j=B^j_\epsilon$ can be easily derived from its hash value in $CERT^j$ and is thus assumed ``seen''),
then he collects the round-$r$ payments that have been propagated to him so far and computes
a maximal payset $PAY^r_i$ from them.

\item[(b)]
If $i$ hasn't seen all $B^0,\ldots, B^{r-1}$ yet, then he sets $PAY^r_i = \emptyset$.

\item[(c)]
Next, $i$ computes his ``candidate block''
$B^r_i = (r, PAY^r_i, SIG_i(Q^{r-1}), H(B^{r-1}))$.

\item[(c)]
Finally, $i$ computes the message $m^{r, 1}_i = (B^r_i, esig_i(H(B^r_i)), \sigma^{r, 1}_i)$,
destroys his ephemeral secret key $sk^{r, 1}_i$,
and then propagates two messages,
$m^{r, 1}_i$ and $(SIG_i(Q^{r-1}), \sigma^{r, 1}_i)$, separately but simultaneously.%
\footnote{When $i$ is the leader,
 $SIG_i(Q^{r-1})$ allows others to compute $Q^r = H(SIG_i(Q^{r-1}), r)$.}

\end{itemize}
\end{itemize}
\end{minipage}
}
\end{center}

\begin{center}
\framebox[\textwidth]{
\centering
\begin{minipage}{.98\textwidth}
\begin{center}
Selective Propagation

\end{center}
To shorten the global execution of Step 1 and the whole round,
it is important that the $(r,1)$-messages
are {\em selectively propagated}.
That is, for every user $j$ in the system,

\begin{itemize}
\item
For the first $(r, 1)$-message
that he ever receives and successfully verifies,%
\footnote{That is, all the signatures are correct and, if it is of the form $m^{r, 1}_i$,
both the block and its hash are valid ---although $j$ does not check whether
the included payset is maximal for $i$ or not.}
whether it contains a block or is just a credential and a signature of $Q^{r-1}$,
player $j$ propagates it as usual.

\item
For all the other $(r, 1)$-messages
that player $j$ receives and successfully verifies,
he propagates it
only if the hash value of
the credential it contains is the {\em smallest} among the hash values of the credentials contained in all
$(r, 1)$-messages
he has received and successfully verified so far.

\item
However, if $j$ receives two different messages of the form $m^{r, 1}_i$ from the same player $i$,%
\footnote{Which means $i$ is malicious.}
 he discards the second one no matter what the hash value of $i$'s credential is.
\end{itemize}

Note that, under selective propagation it is useful that
each potential leader $i$
propagates his credential  $\sigma^{r, 1}_i$ separately from $m^{r, 1}_i$:%
\footnote{We thank Georgios Vlachos for suggesting this.}
those small messages travel faster than blocks,
ensure timely propagation of the $m^{r, 1}_i$'s where the contained credentials have small hash values,
while make those with large hash values disappear quickly.

\end{minipage}
}
\end{center}

\begin{center}
\framebox[\textwidth]{
\centering
\begin{minipage}{.98\textwidth}
\begin{center}
Step 2: The First Step of the Graded Consensus Protocol $GC$
\end{center}

Instructions for every user $i\in PK^{r-k}$:
User $i$ starts his own Step 2 of round $r$ as soon as he
has $CERT^{r-1}$.

\begin{itemize}

\item
User $i$ waits a maximum
amount of time  $t_2 \triangleq \lambda +\Lambda$. While waiting, $i$ acts as follows.

\begin{newitemize}

\item[1.]
After waiting for time $2\lambda$, he finds the user $\ell$ such that $H(\sigma^{r, 1}_{\ell})\leq H(\sigma^{r, 1}_{j})$ for  all credentials $\sigma^{r, 1}_{j}$ that are part of the successfully verified $(r,1)$-messages he has received so far.%
\footnote{Essentially, user $i$ privately decides that the leader of round $r$ is user $\ell$.}

\item[2.]
If he has received a block $B^{r-1}$, which matches the hash value $H(B^{r-1})$ contained in $CERT^{r-1}$,%
\footnote{Of course, if $CERT^{r-1}$ indicates that $B^{r-1}=B^{r-1}_\epsilon$, then $i$ has already ``received'' $B^{r-1}$ the moment he has $CERT^{r-1}$.}
and if he has received from $\ell$ a valid message $m^{r, 1}_{\ell} = (B^r_{\ell}, esig_{\ell}(H(B^r_{\ell})), \sigma^{r, 1}_{\ell})$,%
\footnote{Again, player $\ell$'s signatures and the hashes are all successfully verified, and $PAY^r_{\ell}$ in $B^r_{\ell}$ is a valid payset for round $r$ ---although $i$ does not check whether $PAY^r_{\ell}$ is maximal for $\ell$ or not.
If $B^r_\ell$ contains an empty payset, then there is actually no need for $i$ to see $B^{r-1}$ before verifying whether $B^r_\ell$ is valid or not.}
 then $i$ stops waiting and sets $v'_i \triangleq (H(B^r_{\ell}), \ell)$.

\item[3.] Otherwise, when time $t_2$ runs out,
$i$ sets $v'_i \triangleq \bot$.

\item[4.]
When the value of $v'_i$ has been set,
$i$ computes $Q^{r-1}$ from $CERT^{r-1}$
and checks whether $i\in SV^{r,2}$ or not.

\item[5.]
If $i\in SV^{r,2}$, $i$  computes the message
$m^{r, 2}_i \triangleq (ESIG_i(v'_i), \sigma^{r, 2}_i)$,%
\footnote{The message
$m^{r, 2}_i$ signals that player $i$ considers the first component of $v'_i$ to be the hash of the next block, or considers the next block to be empty.} destroys his ephemeral secret key $sk^{r, 2}_i$, and then propagates $m^{r, 2}_i$.
Otherwise,
$i$ stops without propagating anything.

\end{newitemize}

\end{itemize}
\end{minipage}
}
\end{center}

\begin{center}
\framebox[\textwidth]{
\centering
\begin{minipage}{.98\textwidth}
\begin{center}
Step 3: The Second Step of $GC$
\end{center}

Instructions for every user $i\in PK^{r-k}$:
User $i$ starts his own Step 3 of round $r$ as soon as he
has $CERT^{r-1}$.

\begin{itemize}

\item
User $i$ waits a maximum amount of time  $t_3 \triangleq
t_2+2\lambda = 3\lambda+\Lambda$. While waiting,
$i$ acts as follows.
\begin{newitemize}

\item[1.]
If there exists a value $v$ such that
he has received at least $t_H$ valid messages $m^{r, 2}_j$ of the form $(ESIG_j(v), \sigma^{r, 2}_j)$,
without any contradiction,%
\footnote{That is,
he has not received two valid messages containing $ESIG_j(v)$ and a different $ESIG_j(\hat{v})$ respectively, from a player $j$.
Here and from here on, except in the Ending Conditions defined later, whenever an honest player wants messages of a given form, messages contradicting each other are never counted or considered valid.}
then he
stops waiting and sets $v'=v$.

\item[2.]
Otherwise, when time $t_3$ runs out, he sets $v'=\bot$.

\item[3.]
When the value of $v'$ has been set,
$i$ computes $Q^{r-1}$ from
$CERT^{r-1}$
and checks whether $i\in SV^{r,3}$ or not.

\item[4.]
If $i\in SV^{r, 3}$, then
$i$ computes the message
 $m_i^{r, 3} \triangleq (ESIG_i(v'), \sigma^{r, 3}_i)$,
destroys his ephemeral secret key $sk^{r, 3}_i$, and then propagates $m_i^{r, 3}$. Otherwise,
$i$ stops without propagating anything.

\end{newitemize}

\end{itemize}
\end{minipage}
}
\end{center}

\begin{center}
\framebox[\textwidth]{
\centering
\begin{minipage}{.98\textwidth}
\begin{center}
Step 4: Output of $GC$ and The First Step of $\BBA*$

\end{center}

Instructions for every user $i\in PK^{r-k}$:
User $i$ starts his own Step 4 of round $r$ as soon as
he finishes his own Step 3.

\begin{itemize}

\item
User $i$ waits a maximum amount of time $2\lambda$.%
\footnote{Thus, the maximum {\em total} amount of
time since $i$ starts his Step 1 of round $r$
could be $t_4 \triangleq
t_3+2\lambda = 5\lambda+\Lambda$.}
While waiting,
$i$ acts as follows.

\begin{newitemize}

\item[1.] He computes $v_i$ and $g_i$, the output of GC, as follows.

\begin{newitemize}

\item[(a)]
If there exists a value
$v'\neq \bot$
such that he has received at least $t_H$ valid messages
$m_j^{r, 3} = (ESIG_j(v'), \sigma_j^{r, 3})$,
then
he stops waiting and sets $v_i \triangleq v'$ and $g_i \triangleq 2$.

\item[(b)]
If he has received at least $t_H$ valid messages
$m_j^{r, 3} = (ESIG_j(\bot), \sigma_j^{r, 3})$,
then he stops waiting and sets $v_i \triangleq \bot$
and $g_i\triangleq 0$.%
\footnote{Whether Step (b) is in the protocol or not does not affect its correctness. However, the presence of Step (b) allows Step 4 to end in less than $2\lambda$ time if sufficiently many Step-3 verifiers have ``signed $\bot$.''}

\item[(c)]
Otherwise, when time $2\lambda$ runs out, if
there exists a
value $v'\neq \bot$
such that he
has received at least $\lceil \frac{t_H}{2} \rceil$ valid messages
$m^{r, j}_j = (ESIG_j(v'), \sigma_j^{r, 3})$,
then he sets $v_i \triangleq v'$ and $g_i \triangleq 1$.%
\footnote{It can be proved that the $v'$ in this case, if exists, must be unique.}

\item[(d)]
Else, when time $2\lambda$ runs out, he sets $v_i \triangleq \bot$
 and $g_i \triangleq 0$.

\end{newitemize}

\item[2.]
When the values $v_i$ and $g_i$ have been set, $i$ computes $b_i$,
the input of $\BBA*$, as follows:

  $b_i \triangleq 0$ if $g_i = 2$, and $b_i \triangleq 1$ otherwise.

\item[3.]
$i$ computes $Q^{r-1}$ from
$CERT^{r-1}$ and
checks whether $i\in SV^{r,4}$ or not.

\item[4.]
If $i\in SV^{r, 4}$, he computes the message $m^{r, 4}_i \triangleq (ESIG_i(b_i), ESIG_i(v_i), \sigma^{r, 4}_i)$, destroys his ephemeral secret key $sk^{r, 4}_i$, and propagates $m^{r, 4}_i$. Otherwise, $i$ stops without propagating anything.

\end{newitemize}

\end{itemize}
\end{minipage}
}
\end{center}

\begin{center}
\framebox[\textwidth]{
\centering
\begin{minipage}{.98\textwidth}
\begin{center}
Step $s$, $5\leq s\leq m+2$, $s -2 \equiv 0 \mod 3$: A Coin-Fixed-To-0 Step of $\BBA*$

\end{center}

Instructions for every user $i\in PK^{r-k}$:
User $i$ starts his own Step $s$ of round $r$ as soon as he
finishes his own Step $s-1$.

\begin{newitemize}

\item
User $i$ waits a maximum amount of time $2\lambda$.%
\footnote{Thus, the maximum {\em total} amount of time since $i$ starts his Step 1 of round $r$ could be
$t_s \triangleq t_{s-1}+2\lambda = (2s-3)\lambda+\Lambda$.}
While waiting, $i$ acts as follows.

\begin{newitemize}

\item[--] {\em Ending Condition 0:}
If
at any point
there exists a string $v\neq \bot$
and a step $s'$ such that

\begin{newitemize}

 \item[(a)] $5\leq s'\leq s$, $s'-2 \equiv 0 \mod 3$ ---that is, Step $s'$ is a Coin-Fixed-To-0 step,

\item[(b)] $i$ has received at least $t_H$
valid messages
 $m^{r, s'-1}_j = (ESIG_j(0),$ $ESIG_j(v), \sigma^{r, s'-1}_j)$,%
\footnote{Such a message from player $j$ is counted even if player $i$ has also received a message from $j$ signing for 1. Similar things for Ending Condition 1. As shown in the analysis, this is to ensure that all honest users know
$CERT^r$
within time $\lambda$ from each other.}
  and

\item[(c)] $i$ has received a valid message $(SIG_j(Q^{r-1}), \sigma^{r,1}_j)$ with $j$ being the second component of $v$,

\end{newitemize}

then, $i$ stops waiting and ends his own execution of Step $s$ (and in fact of round $r$) right away without propagating anything as a $(r,s)$-verifier;
sets $H(B^r)$ to be the first component of $v$;
and  sets his own $CERT^r$ to be the set of messages $m^{r, s'-1}_j$ of step (b) together with $(SIG_j(Q^{r-1}), \sigma^{r,1}_j)$.%
\footnote{User $i$ now knows
$H(B^r)$
and his own round $r$ finishes. He just needs to wait until the actually block $B^r$ is propagated to him, which may take some additional time.
He still helps propagating messages as a generic user, but does not initiate any propagation as a $(r, s)$-verifier.
In particular,
he has helped propagating all messages in his $CERT^r$, which is enough for our protocol.
Note that he should also set $b_i \triangleq 0$ for the binary BA protocol, but $b_i$ is not needed in this case anyway.
Similar things for all future instructions.}

\item[--] {\em Ending Condition 1:}
If
at any point there exists a step $s'$ such that

\begin{newitemize}

\item[(a')] $6\leq s'\leq s$, $s'-2\equiv 1 \mod 3$ ---that is, Step $s'$ is a Coin-Fixed-To-1 step, and

\item[(b')]
$i$ has received at least $t_H$ valid messages $m^{r, s'-1}_j = (ESIG_j(1), ESIG_j(v_j),$ $\sigma^{r, s'-1}_j)$,%
\footnote{In this case, it does not matter what the $v_j$'s are.}

\end{newitemize}
then, $i$  stops waiting and ends his own execution of Step $s$ (and in fact of round $r$) right away without propagating anything as a $(r, s)$-verifier; sets $B^r = B^r_\epsilon$; and sets his own $CERT^r$ to be the set of messages $m^{r, s'-1}_j$  of sub-step (b').

\item[--]
If at any point he has received at least $t_H$ valid
$m^{r, s-1}_j$'s
of the form $(ESIG_j(1), ESIG_j(v_j), \sigma^{r, s-1}_j)$,
then he stops waiting and sets $b_i \triangleq 1$.

\item[--]
If at any point he has received at least $t_H$ valid $m^{r, s-1}_j$'s
of the form $(ESIG_j(0), ESIG_j(v_j), \sigma^{r, s-1}_j)$, but
they do not agree on the same $v$, then
he stops waiting and sets $b_i\triangleq 0$.

\item[--]
Otherwise, when time $2\lambda$ runs out,
$i$ sets $b_i\triangleq 0$.

\item[--]
When the value $b_i$ has been set,
$i$ computes $Q^{r-1}$ from $CERT^{r-1}$
and
checks whether $i\in SV^{r,s}$.

\item[--]
If $i\in SV^{r, s}$,
$i$ computes the message $m^{r, s}_i \triangleq (ESIG_i(b_i), ESIG_i(v_i), \sigma^{r, s}_i)$ with $v_i$ being the value he has computed in Step 4,
destroys his ephemeral secret key $sk^{r, s}_i$, and then propagates $m^{r, s}_i$. Otherwise, $i$ stops without propagating anything.

\end{newitemize}

\end{newitemize}

\end{minipage}
}
\end{center}

\begin{center}
\framebox[\textwidth]{
\centering
\begin{minipage}{.98\textwidth}
\begin{center}
Step $s$, $6\leq s\leq m+2$, $s -2 \equiv 1 \mod 3$: A Coin-Fixed-To-1 Step of $\BBA*$

\end{center}

Instructions for every user $i\in PK^{r-k}$:
User $i$ starts his own Step $s$ of round $r$ as soon as he
finishes his own Step $s-1$.

\begin{itemize}

\item
User $i$ waits a maximum amount of time $2\lambda$.
While waiting, $i$ acts as follows.

\begin{newitemize}

\item[--] {\em Ending Condition 0:} The same instructions as in a Coin-Fixed-To-0 step.

\item[--] {\em Ending Condition 1:} The same instructions as in a Coin-Fixed-To-0 step.

\item[--]
If at any point he has received at least $t_H$ valid
$m^{r, s-1}_j$'s
of the form $(ESIG_j(0), ESIG_j(v_j), \sigma^{r, s-1}_j)$,
then he stops waiting and sets $b_i \triangleq 0$.%
\footnote{Note that receiving $t_H$ valid $(r, s-1)$-messages signing for 1 would mean Ending Condition 1.}

\item[--]
Otherwise, when time $2\lambda$ runs out, $i$ sets $b_i\triangleq 1$.

\item[--]
When the value $b_i$ has been set,
$i$ computes $Q^{r-1}$ from $CERT^{r-1}$
and
checks whether $i\in SV^{r,s}$.

\item[--]
If $i\in SV^{r, s}$,
$i$ computes the message $m^{r, s}_i \triangleq (ESIG_i(b_i), ESIG_i(v_i), \sigma^{r, s}_i)$ with $v_i$ being the value he has computed in Step 4,
destroys his ephemeral secret key $sk^{r, s}_i$, and then propagates $m^{r, s}_i$. Otherwise, $i$ stops without propagating anything.

\end{newitemize}
\end{itemize}
\end{minipage}
}
\end{center}

\begin{center}
\framebox[\textwidth]{
\centering
\begin{minipage}{.98\textwidth}
\begin{center}
Step $s$, $7\leq s\leq m+2$, $s -2 \equiv 2 \mod 3$: A Coin-Genuinely-Flipped Step of $\BBA*$

\end{center}

Instructions for every user $i\in PK^{r-k}$:
User $i$ starts his own Step $s$ of round $r$ as soon as he
finishes his own step $s-1$.

\begin{itemize}

\item
User $i$ waits a maximum amount of time $2\lambda$.
While waiting, $i$ acts as follows.

\begin{newitemize}

\item[--] {\em Ending Condition 0:} The same instructions as in a Coin-Fixed-To-0 step.

\item[--] {\em Ending Condition 1:} The same instructions as in a Coin-Fixed-To-0 step.

\item[--]
If at any point he has received at least $t_H$ valid
$m^{r, s-1}_j$'s
of the form $(ESIG_j(0), ESIG_j(v_j), \sigma^{r, s-1}_j)$,
then he stops waiting and sets $b_i \triangleq 0$.

\item[--]
If at any point he has received at least $t_H$ valid
$m^{r, s-1}_j$'s
of the form $(ESIG_j(1), ESIG_j(v_j), \sigma^{r, s-1}_j)$,
then he stops waiting and sets $b_i \triangleq 1$.

\item[--]
Otherwise, when time $2\lambda$ runs out,
letting $SV^{r, s-1}_i$ be the set of $(r, s-1)$-verifiers from whom he has received a valid message $m^{r, s-1}_j$,
$i$ sets $b_i \triangleq \lsb(\min_{j\in SV^{r, s-1}_i} H(\sigma^{r, s-1}_j))$.

\item[--]
When the value $b_i$ has been set,
$i$ computes $Q^{r-1}$ from $CERT^{r-1}$
and
checks whether $i\in SV^{r,s}$.

\item[--]
If $i\in SV^{r, s}$,
$i$ computes the message $m^{r, s}_i \triangleq (ESIG_i(b_i), ESIG_i(v_i), \sigma^{r, s}_i)$ with $v_i$ being the value he has computed in Step 4,
destroys his ephemeral secret key $sk^{r, s}_i$, and then propagates $m^{r, s}_i$. Otherwise, $i$ stops without propagating anything.

\end{newitemize}

\end{itemize}
\end{minipage}
}
\end{center}

\paragraph{Remark.} In principle, as considered in subsection \ref{sec:ephemeral}, the protocol may take arbitrarily many steps in some round.
Should this happens, as discussed, a user $i\in SV\rs$ with $s>\mu$ has exhausted his stash of pre-generated ephemeral keys and has to authenticate his $(r,s)$-message $m^{r, s}_i$
by a ``cascade" of ephemeral keys.
Thus $i$'s message
becomes a bit longer
and transmitting these longer messages will take a bit more time.
Accordingly, after so many steps of a given round, the value of the parameter $\lambda$ will automatically increase slightly. (But it reverts to the original $\lambda$ once a new block is produced and a new round starts.)

\begin{center}
\framebox[\textwidth]{
\centering
\begin{minipage}{.98\textwidth}
\begin{center}
Reconstruction of the Round-$r$ Block by Non-Verifiers

\end{center}

Instructions for every user $i$ in the system:
User $i$ starts his own round $r$ as soon as he has $CERT^{r-1}$.

\begin{itemize}

\item
$i$ follows the instructions of each step of the protocol, participates the propagation of all messages,
but does not initiate any propagation in a step if he is not a verifier in it.

\item
$i$ ends his own round $r$ by entering either Ending Condition 0 or Ending Condition 1 in some step, with the corresponding $CERT^{r}$.

\item
From there on, he starts his round $r+1$ while waiting to receive the actual block $B^r$ (unless he has already received it), whose hash $H(B^r)$ has been pinned down by $CERT^r$. Again, if $CERT^r$ indicates that $B^r = B^r_\epsilon$, the $i$ knows $B^r$ the moment he has $CERT^r$.
\end{itemize}

\end{minipage}
}
\end{center}

\subsection{Analysis of $\bm{\alg_2}$}

The analysis of $\alg_2$ is easily derived from that of $\alg_1$. Essentially, in $\alg_2$, with overwhelming probability, (a)
all honest users agree on the same block $B^r$; the leader of a new block is honest with probability at least $p_h= h^2(1+h-h^2)$.

\section{Handling  Offline Honest users}\label{sec:LazyHonesty}

As we said, a honest user follows all his prescribed instructions, which include that of being online and running the protocol. This is not a major burden in Algorand, since the computation and bandwidth required from a honest user are quite modest. Yet, let us point out that Algorand can be easily modified so as to work in two models, in which honest users are allowed to be offline in great numbers.

Before discussing these two models, let us point out that, if the percentage of honest players were 95\%,  Algorand could still be run setting all parameters assuming instead that $h=80\%$. Accordingly,  Algorand would continue to work properly even if  at most half of the honest players chose to go offline (indeed, a major case of``absenteeism"). In fact, at any point in time, at least 80\% of the players online would be honest.

\paragraph{From Continual Participation to Lazy Honesty}

As we saw, $\alg_1$ and $\alg_2$ choose the look-back parameter $k$. Let us now show that choosing $k$ properly large enables one to remove the Continual Participation requirement. This requirement ensures a crucial property: namely, that the underlying BA protocol $\BBA*$ has a proper honest majority. Let us now explain how lazy honesty provides an alternative and attractive way to satisfy this property.

Recall that a user $i$ is lazy-but-honest if (1) he follows all his prescribed
instructions, when he is asked to participate to the protocol, and (2) he is asked to participate
to the protocol only very rarely ---e.g., once a week---  with suitable advance notice, and potentially receiving significant rewards when he participates.

To allow Algorand to work with such players, it just suffices to ``choose the verifiers of the current round among the users already in the system in a much earlier round."
Indeed, recall that the verifiers for a round $r$ are chosen from users in round $r-k$, and the selections are made based on the quantity $Q^{r-1}$.
Note that a week consists of roughly 10,000 minutes, and assume that a round takes roughly (e.g., on average) 5 minutes, so  a week has roughly 2,000 rounds.
Assume that,
at some point of time,
a user~$i$
wishes
to plan his time and know whether he is going to be a verifier in the coming week.
The protocol now chooses the verifiers for a round $r$ from users in round $r-k-2,000$,
and the selections are based on $Q^{r-2,001}$.
At round $r$, player $i$
already
knows the values $Q^{r-2,000}, \ldots, Q^{r-1}$,
since they are actually part of the blockchain. Then,
for each $M$ between 1 and 2,000,
$i$ is a verifier in a step $s$ of round $r+M$ if and only if
$$.H\left(SIG_i\left(r+M, s, Q^{r+M-2,001}\right)\right)\leq p \enspace. $$
Thus, to check whether he is going to be called to act as
a verifier in the next 2,000 rounds, $i$ must
compute $\sigma^{M,s}_i = SIG_i\left(r+M, s, Q^{r+M-2,001} \right)$
for $M = 1$ to $2,000$ and for each step $s$,
and check whether $.H(\sigma^{M, s}_i) \leq p$ for some of them.
If computing a digital signature takes a millisecond,
then this entire operation will take him about 1 minute
of computation.
If he is not selected as a verifier in any of these rounds, then he can go off-line with an ``honest
conscience". Had he continuously participated,
he would have essentially taken 0 steps in the next
2,000 rounds anyway! If, instead, he is selected to be a verifier in one of these rounds, then he readies himself (e.g., by obtaining all the information necessary) to act as an honest verifier at the proper round.

By so acting, a lazy-but-honest
potential verifier $i$ only misses participating to the
propagation of messages. But message propagation is typically robust.
Moreover, the payers and
the payees of recently propagated payments are expected to be online to watch what happens to their
payments, and thus they will participate to message propagation, if they are honest.

\section{Protocol $\bm{\alg}$ with Honest Majority of Money}\label{sec:hmm}

 We now, finally, show how to replace the Honest Majority of Users assumption  with the much more meaningful Honest Majority of Money  assumption.
The basic idea  is (in a proof-of-stake flavor)
``to select a user $i\in PK^{r-k}$
to belong to $SV\rs$ with a weight
(i.e., decision power)
proportional to the amount of money owned by $i$.''%
\footnote{We should say $PK^{r-k-2,000}$ so as to
replace continual participation. For simplicity, since one may wish to require continual participation anyway, we use $PK^{r-k}$
as before, so as to carry one less parameter.}

By our HMM assumption, we can choose whether that amount should be owned at round $r-k$ or at (the start of) round $r$. Assuming that we do not mind continual participation, we opt for the latter choice. (To remove continual participation, we would have opted for the former choice. {\em Better said, for the amount of money owned at round $r-k-2,000$.})

There are many ways to implement this  idea.
The simplest way would be to have each key hold at most 1 unit of money and then select at random $n$ users $i$ from $PK^{r-k}$ such that $a_i^{(r)}=1$.

\subsubsection*{The Next Simplest Implementation}

The next simplest implementation may be to demand that each public key
owns a maximum amount of money $M$, for some fixed $M$.
The value $M$ is small enough compared with the total amount of money in the system,
such that the probability a key belongs to
the verifier set of more than one step in ---say--- $k$ rounds is negligible.
Then, a key $i\in PK^{r-k}$, owning an amount of money $a_i^{(r)}$ in round $r$,  is chosen to belong to $SV\rs$
if
$$ .H\left(SIG_i\left(r,s,Q\uprm \right)\right)\leq p \cdot \frac{a_i^{(r)}}{M} \enspace.$$
And all proceeds as before.

\subsubsection*{A More Complex Implementation}

The last implementation  ``forced a rich participant in the system to own many keys".

An alternative implementation, described below, generalizes the notion of status
and consider each user $i$ to consist of $K+1$ {\em copies} $(i,v)$, each of which is independently selected to be a verifier, and will own his own ephemeral key $(pk_{i,v}^{r,s}, sk_{i,v}^{r,s})$ in a step $s$ of a round $r$.
The value $K$ depends on the amount of money $a_i^{(r)}$ owned by $i$ in round $r$.

Let us now see how such a system works in greater detail.

\paragraph{Number of  Copies}

Let  $n$ be the targeted expected cardinality of each verifier set,
and let $a_i^{(r)}$ be the amount of money owned by a user $i$ at round~$r$. Let~$A\upr$ be the total amount of money owned by the users in $PK^{r-k}$ at round $r$, that is,
$$A\upr=\sum_{i\in PK^{r-k}}a_i^{(r)}.$$
If $i$ is an user in $PK^{r-k}$,
then
$i$'s copies are $(i,1),\ldots,(i,K+1)$, where
$$K=\left\lfloor \frac{n\cdot a_i^{(r)}} {A\upr}\right\rfloor \enspace.$$

\medskip

\scparagraph{Example.} Let $n=1,000$, $A\upr=10^9$, and $a_i^{(r)}= 3.7$ millions.
Then,
$$K=\left\lfloor \frac{10^3 \cdot (3.7\cdot 10^6)}{10^9}\right\rfloor=\lfloor 3.7\rfloor=3\enspace.$$

\paragraph{Verifiers and Credentials}

Let $i$ be a user in $PK^{r-k}$
with $K+1$
copies.

For each $v=1,\ldots,K$, copy $(i,v)$  belongs to $SV\rs$ automatically.
That is, $i$'s credential is $\sigma^{r, s}_{i,v} \triangleq SIG_i((i, v), r, s, Q^{r-1})$, but the corresponding condition becomes
$.H(\sigma^{r, s}_{i, v})\leq 1$, which is always true.

\medskip

For copy $(i,K+1)$,
for each Step $s$ of round $r$,
$i$ checks whether
$$  .H\big( SIG_i \big(\, (i,K+1), r, s,Q\uprm \, \big)\big)\leq a_i^{(r)} \frac{n}{A\upr} - K  \enspace. $$
If so, copy $(i,K+1)$ belongs to $SV\rs$. To prove it, $i$ propagates the credential
$$\sigma_{i, K+1}^{r,1} = SIG_i \big(\, (i,K+1), r, s,Q\uprm \, \big).$$

\medskip

\scparagraph{Example.}
As in the previous example, let $n=1K$,  $a_i^{(r)}=3.7M$, $A\upr=1B$, and $i$
has 4 copies: $(i,1),\ldots,(i,4)$.
Then, the first 3 copies belong to $SV\rs$ automatically. For the 4th one, conceptually, $\alg$
independently rolls a biased coin, whose probability of Heads is  0.7.
Copy $(i,4)$ is selected if and only if the coin toss is Heads.

(Of course, this biased coin flip is implemented by hashing, signing, and comparing ---as we have done all along in this paper--- so as to enable $i$ to prove his result.)

\bigskip

\paragraph{Business as Usual}
Having explained how verifiers are selected and
how their credentials are computed at each step of a round $r$, the execution of a round
is similar to that already explained.

\section{Handling Forks}

Having reduced the probability of forks to $10^{-12}$ or $10^{-18}$, it is practically unnecessary to handle them in the remote chance that they occur.
Algorand, however, can also employ various  fork resolution procedures, with or without proof of work.

One possible way of instructing the users to resolve forks is as follows:
\begin{itemize}
\item
Follow the longest chain if a user sees multiple chains.

\item
If there are more than one longest chains, follow the one with a non-empty block at the end.
If all of them have empty blocks at the end, consider their second-last blocks.

\item
If there are more than one longest chains with non-empty blocks at the end, say the chains are of length $r$,
follow the one whose leader of block $r$ has the smallest credential.
If there are ties, follow the one whose block $r$ itself has the smallest hash value.
If there are still ties, follow the one whose block $r$ is ordered the first lexicographically.
\end{itemize}

\section{Handling Network Partitions}\label{sec:partition}
As said, we assume the propagation times
of messages among all users in the network are upper-bounded by $\lambda$ and $\Lambda$. This is not a strong assumption,
as today's Internet is fast and robust, and the actual values of these parameters are quite reasonable.
Here, let us point out that $\alg_2$ continues to work even if the Internet occasionally got partitioned into two parts. The case when the Internet is partitioned into more than two parts is similar.

\subsection{Physical Partitions}
First of all, the partition may be caused by physical reasons. For example, a huge earthquake may end up completely breaking down
the connection between Europe and America.
In this case, the malicious users
are also partitioned and there is no communication between the two parts.
Thus there will be two Adversaries, one for part 1 and the other for part 2.
Each Adversary still tries to break the protocol in its own part.

Assume the partition happens in the middle of round $r$. Then each user is still selected as a verifier based on $PK^{r-k}$, with the same probability as before.
Let $HSV_i^{r, s}$ and $MSV_i^{r, s}$ respectively be the set of honest and malicious
verifiers in a step $s$ in part $i\in \{1, 2\}$.
We have
$$|HSV_1^{r, s}| + |MSV_1^{r, s}| + |HSV_2^{r, s}|+|MSV_2^{r, s}| = |HSV^{r, s}| + |MSV^{r, s}|.$$
Note that $|HSV^{r, s}| + |MSV^{r, s}|< |HSV^{r, s}|+2|MSV^{r, s}|<2t_H$ with overwhelming probability.

If some part $i$ has $|HSV_i^{r, s}| + |MSV_i^{r, s}|\geq t_H$ with non-negligible probability, e.g., $1\%$, then
the probability that $|HSV_{3-i}^{r, s}| + |MSV_{3-i}^{r, s}|\geq t_H$ is very low, e.g., $10^{-16}$ when $F=10^{-18}$.
In this case, we may as well treat the smaller part as going offline, because there will not be enough verifiers in this part to generate $t_H$ signatures
to certify a block.

Let us consider the larger part, say  part 1 without loss of generality.
Although $|HSV^{r, s}|<t_H$ with negligible probability in each step $s$,
when the network is partitioned, $|HSV_1^{r, s}|$ may be
less than $t_H$ with some non-negligible probability.
In this case the Adversary may, with some other non-negligible probability,
force the binary BA protocol into a fork in round $r$,
with a non-empty block $B^r$ and the empty block $B^r_\epsilon$ both having $t_H$ valid signatures.%
\footnote{Having a fork with two non-empty blocks is not possible with or without partitions, except with negligible probability.}
For example, in a Coin-Fixed-To-0 step $s$, all verifiers in $HSV_1^{r, s}$ signed for bit 0 and $H(B^r)$, and propagated their messages.
All verifiers in $MSV_1^{r, s}$ also signed 0 and $H(B^r)$, but withheld their messages.
Because $|HSV_1^{r, s}|+|MSV_1^{r, s}|\geq t_H$,
the system has enough signatures to certify $B^r$.
However, since the malicious verifiers withheld their signatures, the users enter step $s+1$,
which is a Coin-Fixed-To-1 step.
Because $|HSV_1^{r, s}|<t_H$ due to the partition,
the verifiers in $HSV_1^{r, s+1}$ did not see $t_H$ signatures for bit 0
and they all signed for bit 1.
All verifiers in $MSV_1^{r, s+1}$ did the same.
Because $|HSV_1^{r, s+1}|+|MSV_1^{r, s+1}|\geq t_H$,
the system has enough signatures to certify $B^r_\epsilon$.
The Adversary then creates a fork by releasing the signatures of $MSV_1^{r, s}$ for 0 and $H(B^r)$.

Accordingly, there will be two $Q^{r}$'s, defined by the corresponding blocks of round $r$.
However, the fork will not continue and only one of the two branches may grow in round $r+1$.

\paragraph{Additional Instructions for $\bm{\alg_2}$.}
{\em When seeing a non-empty block $B^r$ and the empty block $B^r_\epsilon$,
follow the non-empty one (and the $Q^r$ defined by it).}

\bigskip

Indeed, by instructing the users to go with the non-empty block in the protocol,
if a large amount of honest users in $PK^{r+1-k}$ realize there is a fork at the beginning of round $r+1$, then
 the empty block will not have enough followers and will not grow.
Assume the Adversary manages to partition the honest users so that some honest users see $B^r$ (and perhaps $B^r_\epsilon$), and
some only see $B^r_\epsilon$.
Because the Adversary cannot tell which one of them will be a verifier following $B^r$
and which will be a verifier following $B^r_\epsilon$, the honest users are randomly partitioned and
each one of them
still becomes a verifier (either with respect to $B^r$ or with respect to $B^r_\epsilon$) in a step $s>1$ with probability $p$.
For the malicious users, each one of them may have two chances to become a verifier,
one with $B^r$ and the other with $B^r_\epsilon$, each with probability $p$ independently.

Let $HSV_{1; B^r}^{r+1, s}$ be the set of honest verifiers in step $s$ of round $r+1$ following $B^r$.
Other notations such as $HSV_{1; B^r_\epsilon}^{r+1, s}$, $MSV_{1; B^r}^{r+1, s}$ and $MSV_{1; B^r_\epsilon}^{r+1, s}$ are similarly defined.
By Chernoff bound, it is easy to see that with overwhelming probability,
$$|HSV_{1; B^r}^{r+1, s}| + |HSV_{1; B^r_\epsilon}^{r+1, s}| + |MSV_{1; B^r}^{r+1, s}| + |MSV_{1; B^r_\epsilon}^{r+1, s}| < 2t_H.$$
Accordingly, the two branches cannot both have $t_H$ proper signatures certifying a block for round $r+1$ in the same step $s$.
Moreover, since the selection probabilities for two steps $s$ and $s'$ are the same and the selections are independent,
also with overwhelming probability
$$|HSV_{1; B^r}^{r+1, s}| + |MSV_{1; B^r}^{r+1, s}| + |HSV_{1; B^r_\epsilon}^{r+1, s'}| + |MSV_{1; B^r_\epsilon}^{r+1, s'}| < 2t_H,$$
for any two steps $s$ and $s'$.
When $F=10^{-18}$, by the union bound,
as long as the Adversary cannot partition the honest users for a long time (say $10^4$ steps, which is more than 55 hours with $\lambda = 10$ seconds%
\footnote{Note that a user finishes a step $s$ without waiting for $2\lambda$ time only if he has seen at least $t_H$ signatures for the same message.
When there are not enough signatures, each step will last for $2\lambda$ time.}),
with high probability (say $1-10^{-10}$)
at most one branch will have $t_H$ proper signatures to certify a block in round $r+1$.

Finally, if the physical partition has created two parts with roughly the same size, then
the probability that $|HSV_i^{r, s}| + |MSV_i^{r, s}|\geq t_H$ is small for each part $i$.
Following a similar analysis, even if the Adversary manages to create a fork with some non-negligible probability in each part for round $r$,
at most one of the four branches may grow in round $r+1$.

\subsection{Adversarial Partition}
Second of all, the partition may be caused by the Adversary, so that the messages propagated by the honest users in one part will not reach the honest
users in the other part directly, but the Adversary is able to forward messages between the two parts.
Still, once a message from one part
reaches an honest user in the other part, it will be propagated in the latter as usual. If the Adversary is willing to spend a lot of money,
it is conceivable that he may be able to hack the Internet and partition it like this for a while.

The analysis is similar to that for the larger part in the physical partition above (the smaller part can be considered as having population 0):
the Adversary may be able to create a fork and each honest user only sees one of the branches,
but at most one branch may grow.

\subsection{Network Partitions in Sum}
Although network partitions can happen and a fork in one round may occur under partitions,
there is no lingering ambiguity:
a fork is very short-lived, and in fact lasts for at most a single round.
In all parts of the partition except for at most one, the users cannot generate a new block and thus (a) realize there is a
partition in the network and (b) never rely on blocks that will ``vanish''.

\section*{Acknowledgements}\label{Ack}

We would like to first acknowledge Sergey Gorbunov, coauthor of the cited Democoin system.

Most sincere thanks go to Maurice Herlihy, for many enlightening discussions, for pointing out that pipelining will improve Algorand's throughput performance,  and for greatly improving the exposition of an earlier version of this paper.
Many thanks to Sergio Rajsbaum, for his comments on an earlier version of this paper.
Many thanks to
Vinod Vaikuntanathan, for several deep discussions and insights.
Many thanks to Yossi Gilad, Rotem Hamo, Georgios Vlachos, and Nickolai Zeldovich
for starting to test these ideas, and for many helpful comments and discussions.

Silvio Micali would like to personally thank Ron Rivest for innumerable discussions and guidance in cryptographic research over more than 3 decades, for coauthoring the cited micropayment system that has inspired one of the verifier selection mechanisms of Algorand.

We hope to bring this technology to the next level. Meanwhile the travel and companionship are great fun, for which we are very grateful.

\end{document}